\documentclass[11pt,a4paper,reqno]{article}
\usepackage{amsmath}
\usepackage{graphicx}
\usepackage{amsfonts}
\usepackage{amssymb}
\usepackage{amsthm}
\usepackage[left=2.5cm, right=2.5cm, top=2.5cm, bottom=2.5cm]{geometry}
\usepackage{indentfirst}
\usepackage[all]{xy}
\usepackage[colorlinks=true,linkcolor=blue]{hyperref}
\usepackage{mathrsfs} 
\usepackage{tikz-cd}
\usetikzlibrary{graphs,decorations.pathmorphing,decorations.markings}
\usepackage{enumitem}
\usepackage[title]{appendix}
\setitemize{noitemsep,topsep=0pt,parsep=0pt,
partopsep=0pt,itemindent=12pt,leftmargin=0pt}
\makeatletter
\newcommand{\subjclass}[2][2020]{%
  \let\@oldtitle\@title%
  \gdef\@title{\@oldtitle\footnotetext{#1 \emph{Mathematics subject classification.} #2}}%
}
\newcommand{\keywords}[1]{%
  \let\@@oldtitle\@title%
  \gdef\@title{\@@oldtitle\footnotetext{\emph{Key words and phrases.} #1.}}%
}
\makeatother


\newtheorem{theorem}{Theorem}[section]
\newtheorem{lemma}[theorem]{Lemma}
\newtheorem{proposition}[theorem]{Proposition}
\newtheorem{corollary}[theorem]{Corollary}

\theoremstyle{definition}
\newtheorem{definition}[theorem]{Definition}
\newtheorem{example}[theorem]{Example}

\theoremstyle{remark}
\newtheorem{remark}[theorem]{Remark}

\newcommand{\build}[3]{\mathrel{\mathop{\kern 0pt#1}\limits_{#2}^{#3}}}

\newcommand\SU{{\mathrm{SU}}}
\newcommand\SO{{\mathrm{SO}}}
\newcommand\Sp{{\mathrm{Sp}}}
\newcommand\GL{{\mathrm{GL}}}

\newcommand\U{{\mathrm U}}
\newcommand\Z{{\mathbb Z}}

\newcommand\Id{\mathrm{Id}}

\newcommand\R{\mathbb{R}}
\newcommand\C{\mathbb{C}}
\newcommand\Tr{\mathrm{Tr}}

\newcommand\E{\mathbb{E}}

\newcommand\YM{\mathrm{YM}}
\newcommand\Hom{\mathrm{Hom}}

\newcommand\End{{\mathrm{End}}}
\newcommand\Fbb{{\mathbb{F}}}
\newcommand\TT{{\mathring{T}}}

\newcommand\Lc{{\mathcal{L}}}
\newcommand\B{{\mathcal{B}}}

\newcommand\Br{{\mathfrak{B}}}
\newcommand\Span{{\mathrm{Span}}}
\newcommand\M{{\mathcal{M}}}
\newcommand\Wil{{\mathrm{W}}}
\newcommand\HK{{\mathrm{HK}}}

\newcommand\EPE{{\mathrm{EPE}}}
\newcommand\Spec{{\mathrm{Spec}}}

\title{Universal dualities for Wilson loops in lattice Yang--Mills}

\author{Thibaut Lemoine\thanks{Universit\'e de Strasbourg, CNRS, UMR 7501 -- Institut de Recherche Math\'ematique Avanc\'ee, 7 rue Ren\'e Descartes, 67000 Strasbourg, France. thibaut.lemoine@math.unistra.fr}}



\begin{document}

\maketitle

\begin{abstract}
We identify a universal finite-$N$ structure underlying Wilson loop expectations in lattice Yang--Mills, in any dimension $d\geq 2$, for gauge group $\mathrm{U}(N)$, and for arbitrary smooth central plaquette actions. The starting point is a state-sum expansion in plaquette labels by irreducible representations, in which each term factorizes into an action-dependent spectral weight and an action-independent topological coefficient. We then analyze these coefficients in three exact ways: as a gauge/string expansion over decorated spanning surfaces, as a local spin-foam/channel model on the dual incidence graph, and as a universal finite-$N$ master loop equation that closes on the coefficient side. As a consequence, several recent Wilson-action results are recovered as specializations of our broader action-agnostic framework. In a companion work, this framework is used to construct the higher-dimensional heat-kernel master field at strong coupling and to prove an area-law bound.
\end{abstract}

\tableofcontents

\section{Introduction}\label{sec:intro}

\subsection{Lattice gauge theory}

Lattice gauge theory (LGT for short) is a model that aims to describe a discretized version of the quantum gauge theory underlying the Standard Model of physics. Although the construction of a genuine four-dimensional Euclidean quantum Yang--Mills theory remains a great challenge today \cite{JW}, there have been several parallel approaches, initiated around the year 1974 by pioneering works \cite{Wil74,Hoo74,Mig75}, all aimed at a better understanding of a discrete version of the theory. The common framework of these papers can be summarized as follows.

Consider a lattice\footnote{In fact, the construction also works for any 2-dimensional CW-complex embedded in a compact manifold. We keep everything embedded in $\Z^d$ by convention but not by limitation.} $\Lambda\subset\Z^d$, which may be finite or infinite, and a compact Lie group $G$. Let $E(\Lambda)$ (resp. $P(\Lambda)$) be the set of 1-cells (resp. 2-cells) of $\Lambda$, seen as a CW-complex. We call respectively \emph{edges} and \emph{plaquettes} (or \emph{faces}) the 1-cells and 2-cells of the lattice. A \emph{gauge configuration} is an assignment $U=(U_e)_{e\in E(\Lambda)}$ of $G$-valued random variables on each edge. For every plaquette $p\in P(\Lambda)$, let $Q_p:G\to(0,\infty)$ be a smooth central function. We call $Q=(Q_p)_{p\in P(\Lambda)}$ a family of \emph{plaquette weights}. The \emph{lattice Yang--Mills measure on $\Lambda$ with plaquette weights $Q$} is the measure
\begin{equation}\label{eq:lattice-YM}
d\mu_{\Lambda,Q}(U)=\frac{1}{Z_{\Lambda,Q}}\prod_{p\in P(\Lambda)}Q_p(U_{\partial p})dU,
\end{equation}
where $dU=\prod_{e\in E(\Lambda)}dU_e$ is the normalized Haar measure on $G^{E(\Lambda)}$, $U_{\partial p}=U_{e_1}\ldots U_{e_k}$ if the boundary of the plaquette $p$ is represented by $\partial p=e_1\ldots e_k$ as a concatenation of oriented edges, and $Z_{\Lambda,Q}$ is a normalization factor called the \emph{partition function}, ensuring that $\int_{G^{E(\Lambda)}}d\mu_{\Lambda,Q}(U)=1$. Note that $U_{\partial p}$ is defined up to conjugation for any $p\in P(\Lambda)$ because it depends on a parametrization of $\partial p$, but as $Q_p$ is central, the measure does not depend on this parametrization. A one-parameter action family $(Q_t)_{t>0}$ together with plaquette parameters $(t_p)_p$ is recovered by setting $Q_p=Q_{t_p}$; this includes the usual Wilson inverse-temperature and heat-kernel time parametrizations without identifying the two conventions. The main quantities of interest in the theory are the partition function $Z_{\Lambda,Q}$ and the \emph{Wilson loops}, defined for a family $\Lc=(\ell_1,\ldots,\ell_k)$ of loops in the 1-skeleton of $\Lambda$ by
\begin{equation}
W_{\Lambda,\Lc}(U)=\Tr(U_{\ell_1})\ldots\Tr(U_{\ell_k}),
\end{equation}
where $U_\ell$ is again defined as the product of the edge variables corresponding to the loop $\ell$.

Strictly speaking, the definition of Wilson loops requires a choice of a finite-dimensional unitary representation $\rho:G\to U(V)$, and one should write
\[
W_{\Lambda,\Lc}^{\rho}(U)=\prod_{i=1}^k \Tr\big(\rho(U_{\ell_i})\big).
\]
By the Peter--Weyl theorem, every compact Lie group admits a faithful finite-dimensional unitary representation, so after fixing one such representation we may, without loss of generality, regard $G$ as a compact subgroup of $\U(N)$ for some $N$ and use the ordinary matrix trace notation.

Wilson loops are named after Wilson \cite{Wil74}, who introduced them together with most of the settings of lattice gauge theory in order to study the confinement of quarks. Wilson considered the so-called \emph{Wilson action}\footnote{There is actually another definition, where $t\Re\Tr(x)$ is replaced by $-t\Re\Tr(I-x)$, but it only changes the normalization.} for $G$ a subgroup of a unitary group $\U(N)$
\[
Q_t^{\Wil}(x)=\exp(t\Re(\Tr(x))).
\]

Concurrently, 't Hooft \cite{Hoo74} considered a similar framework, and took $G=\SU(N)$ while letting the parameter $N$ tend to infinity, showing that the $1/N$ expansion of observables simplifies into a sum over planar Feynman diagrams in the limit. It is now known as the \emph{large-$N$} regime.

Almost at the same time, Migdal \cite{Mig75} considered lattice gauge theory with another action, called the \emph{heat kernel action} or \emph{Villain action}
\[
Q_t^{\HK}(x)=p_t(x),\quad \forall x\in G,
\]
where $(p_t)_{t>0}$ is the heat kernel on $G$ for a well-chosen $G$-invariant inner product on its Lie algebra $\mathfrak g$. He showed, in particular, that the semigroup property of the heat kernel translates into topological properties such as cutting and gluing properties, in particular in two-dimensional lattice gauge theory where the model is integrable.

These works have led to many remarkable improvements in several directions in the last fifty years, both in physics and in mathematics. Let us focus on the main mathematical contributions:
\begin{itemize}
\item A rigorous construction of the continuous two-dimensional Yang--Mills measure as a measure on holonomy fields \cite{Dri89,Sen97,Lev03,Lev10},
\item The formalization of the \emph{master field} \cite{Sing95} as a noncommutative process corresponding to the large-$N$ limit of Wilson loops, which is now essentially solved in two dimensions \cite{Dah16,Lev17,Hal18,DN,DL23,DL25,Lem25,Dah26}, and for which significant results were achieved in 't Hooft regime at strong coupling \cite{Cha19,Jaf16,BCSK24}.
\item Partial proofs of the \emph{gauge/string duality}, which is a general principle stating that observables of Yang--Mills theory should admit a ``dual string theory'', that is, a theory of random surfaces \cite{Lev08,Jaf16,Cha19,CPS25,LemMai25,LM2}.
\item Several partial rigorous constructions of the continuous Yang--Mills measure in terms of random distributional connections or in terms of Langevin dynamics \cite{Chev19,CCHS22,CCHS24,ChevShen26,DangNohra26}.
\end{itemize}
For more details, the reader might have a look at the following surveys: \cite{LevSen17,Cha19survey,Lev20survey,Lem26survey}.

It is natural to wonder why one action should be used rather than the others. From the continuum point of view, they should be understood as different lattice regularizations of the same formal Yang--Mills energy, defined formally as the squared $L^2$ norm of the curvature form of a connection $A$:
\[
S_{\YM}(A)=\frac{1}{2}\int \|F_A(x)\|^2\,dx .
\]
Indeed, if \(p\) is a small plaquette and \(U_{\partial p}\) denotes its holonomy, then formally
\[
U_{\partial p}\simeq \exp(\operatorname{area}(p)F_A(x_p)).
\]
Thus a central plaquette action whose logarithm has a non-degenerate quadratic minimum at the identity gives, to first order in the mesh size, a discretization of the \(L^2\)-norm of the curvature. For instance, the Wilson action realizes this through the expansion of \(\Re\Tr(I-U_{\partial p})\) near the identity, and the heat-kernel action through the small-time Gaussian behaviour of the heat kernel.

These actions nevertheless have different algebraic features. The Wilson action is the standard lattice action and is particularly well adapted to trace expansions and to the classical form of the master loop equation. The heat-kernel action is more rigid from the two-dimensional continuum viewpoint: its semigroup property is precisely the compatibility property needed when one subdivides plaquettes, and it is the graph-level action underlying the construction of the continuum two-dimensional Yang--Mills holonomy field. More generally, many reasonable plaquette actions are expected, after the appropriate scaling, to belong to the same continuum universality class. Several results in 2D confirm this expectation, identifying the heat kernel action as the natural continuum limit of various lattice actions \cite{ChevGar25,ChevShen26,DangNohra26}. For this reason, we believe that it is crucial to \emph{understand lattice gauge theory beyond the Wilson action}, and in particular for the heat kernel action.

The main motivation of this paper is therefore to identify a \emph{universal finite-$N$ structure} underlying Wilson loop expectations in lattice Yang--Mills, valid in arbitrary dimension and for arbitrary smooth central plaquette actions, that generalizes results previously known only in two dimensions or only for the Wilson action. The starting point is a spectral/topological splitting: after Fourier expansion of the plaquette action, Wilson loop expectations become sums over plaquette labels, with action-dependent spectral weights and action-independent topological coefficients. We then show that these coefficients unravel two complementary kinds of dualities: a gauge/string duality on a global geometric side, and a spin-foam duality on a local representation-theoretic side. We will also prove that they satisfy a kind of Schwinger--Dyson equation, the \emph{master loop equation}. In this way, surface sums, spin-foam-type dual models, and master loop equations appear as three manifestations of a single universal formalism that was previously understood for the Wilson action only. The conceptual point is the following: once the spectral/topological splitting is isolated, surface expansions, local dual models, and loop equations all emerge from the same action-independent topological coefficients, and the existing Wilson-action results are just a specialization of this universal formalism.

\subsection{A state-sum expansion}

The starting point of this framework is the following theorem.

\begin{theorem}[State-sum representation]\label{thm:state_sum_Wilson}
Let $\Lambda\subset\Z^d$ be a finite lattice, with oriented edge set $E(\Lambda)$ and oriented plaquette set $P(\Lambda)$. Let $U=(U_e)_{e\in E(\Lambda)}$ be a random configuration with distribution $\mu_{\Lambda,Q}$, where $Q=(Q_p)_{p\in P(\Lambda)}$ is a family of smooth central plaquette weights. For any family $\Lc=(\ell_1,\ldots,\ell_k)$ in the 1-skeleton of $\Lambda$, we have
\begin{equation}\label{eq:State_sum_Wilson}
\E\left[W_{\Lambda,\Lc}(U)\right]=\frac{1}{Z_{\Lambda,Q}}\sum_{\alpha:P(\Lambda)\to\widehat{G}}\kappa_{\Lambda,Q}(\alpha)\widehat{W}_{\Lambda,\Lc}(\alpha),
\end{equation}
where
\begin{equation}\label{eq:top_coef}
\kappa_{\Lambda,Q}(\alpha)=\prod_{p\in P(\Lambda)}\widehat Q_p(\alpha_p),\qquad
\widehat Q_p(\lambda):=\langle Q_p,\chi_\lambda\rangle_{L^2(G)},
\end{equation}
and
\begin{equation}\label{eq:topological-coefficient}
\widehat{W}_{\Lambda,\Lc}(\alpha)=\int_{G^{E(\Lambda)}} \left(\prod_{i=1}^k\Tr(U_{\ell_i})\right)\left(\prod_{p\in P(\Lambda)}\chi_{\alpha_p}(U_{\partial p})\right)dU.
\end{equation}
\end{theorem}

Simply stated, all Wilson loop expectations are sums over \emph{plaquette decorations} $\alpha:P(\Lambda)\to\widehat{G}$, with a product of two terms:
\begin{itemize}
\item A \emph{spectral coefficient} $\kappa_{\Lambda,Q}(\alpha)$, which depends on the plaquette weights but not on the loops,
\item A \emph{topological coefficient} $\widehat{W}_{\Lambda,\Lc}(\alpha)$ that depends on the loops but not on the plaquette weights.
\end{itemize}
The proof is immediate from the Peter--Weyl expansion; we include it to fix the normalization.

\begin{proof}[Proof of Theorem~\ref{thm:state_sum_Wilson}]
We start with
\[
\E[W_{\Lambda,\Lc}(U)]=\frac1{Z_{\Lambda,Q}}\int_{G^{E(\Lambda)}}\left(\prod_{i=1}^k\Tr(U_{\ell_i})\right)\prod_{p\in P(\Lambda)}Q_p(U_{\partial p})dU,
\]
assuming that $G$ is a compact subgroup of $\U(N)$ for some integer $N$. We expand each $Q_p$ into irreducible characters by the Peter--Weyl theorem:
\[
Q_p(x)=\sum_{\lambda\in\widehat{G}}\widehat Q_p(\lambda)\chi_\lambda(x),
\qquad
\widehat Q_p(\lambda)=\int_G Q_p(g)\overline{\chi_\lambda(g)}\,dg.
\]
Because $Q_p$ is smooth, this expansion is absolutely and uniformly convergent and
\[
\sum_{\lambda\in \widehat G}d_\lambda\,|\widehat Q_p(\lambda)|<\infty,
\]
where $d_\lambda=\chi_\lambda(1)$. Since $|\Tr(U_{\ell_i})|\leq N$ for all $1\leq i\leq k$, we obtain the uniform bound
\[
\sum_{\alpha:P(\Lambda)\to \widehat G}
\left|
W_{\Lambda,\Lc}(U)
\prod_{p\in P(\Lambda)}
\widehat Q_p(\alpha_p)\chi_{\alpha_p}(U_{\partial p})
\right|
\le
N^k\prod_{p\in P(\Lambda)}\sum_{\lambda\in\widehat G}d_\lambda|\widehat Q_p(\lambda)|.
\]
We may therefore exchange the sum and the integral. The resulting expression is exactly~\eqref{eq:State_sum_Wilson} after identifying~\eqref{eq:top_coef} and~\eqref{eq:topological-coefficient}.
\end{proof}

Despite its elementary proof, this result already shows the key idea of the paper: \emph{the combinatorial aspects of Wilson loop expectations do not depend on the plaquette weights.} This is the most important conceptual input, as it separates the action dependence from the combinatorics of Wilson observables. All subsequent results will be manifestations of the combinatorial structure of the coefficients $\widehat W_{\Lambda,\Lc}(\alpha)$ at fixed plaquette decoration $\alpha$. Let us insist on the fact that, albeit stated for sublattices of $\Z^d$, Theorem~\ref{thm:state_sum_Wilson} holds for any finite two-dimensional CW-complex, e.g. a topological map, or the triangulation of an arbitrary manifold of any dimension $d\geq 2$ and any set of loops in this complex. The arguments do not rely specifically on the fact that we work on a sublattice of $\Z^d$, but we focus on this case because it is standard in lattice gauge theory.

For special plaquette weights it is often useful to normalize the spectral coefficients into a reference probability law. We shall not do so at the level of the general theory, because the Fourier coefficients of a general smooth central weight need not be positive and their unweighted total mass need not provide a useful normalization. In the heat-kernel case, by contrast, the spectral coefficients are positive and lead naturally to the discrete Gaussian reference law used in the companion paper~\cite{Lem26prep}; see also~\cite{LemMai25,LM2,Lem25}.

With this representation in mind, we will now focus on the topological coefficients $\widehat{W}_{\Lambda,\Lc}(\alpha)$. They can be reinterpreted as integrals of products of irreducible characters (by remarking that the trace is the character of the fundamental representation), therefore the combinatorial simplification of Wilson loop expectations boils down to understanding such integrals in general.

In the sequel, we shall focus on $G=\U(N)$ and explore the combinatorial aspects of Wilson loop expectations by means of two complementary approaches:
\begin{enumerate}
\item A rational Weingarten calculus that generalizes the traditional Weingarten calculus developed by Collins--Sniady~\cite{ColSni06}, and which is based on the mixed Schur--Weyl duality initiated by Koike~\cite{Koi89} and revisited by Magee and Dahlqvist~\cite{Mag,Mag2,Dah26}.
\item A spin-network/tensor-network approach inspired by the works of Oeckl--Pfeiffer~\cite{OeckPfei01} that is more often used in quantum gravity~\cite{DMS16} but that has been proved relevant for two-dimensional Yang--Mills theory in two recent papers~\cite{Lev26,Dah26}.
\end{enumerate}

Most of our main theorems will be stated for Wilson loop expectations $\E[W_{\Lambda,\Lc}(U)]$ to fit into the usual language of lattice gauge theory, but, as the developments will reveal, they are actually consequences of formulas obtained specifically for the topological coefficient $\widehat{W}_{\Lambda,\Lc}(\alpha)$ at a fixed plaquette decoration $\alpha:P(\Lambda)\to\widehat{\U(N)}$.

\subsection{Gauge/string duality}

The first consequence of Theorem~\ref{thm:state_sum_Wilson} is that the Wilson loop expectations can be rewritten in terms of surface-sums.

\begin{theorem}[Exact decorated-surface expansion]\label{thm:surface-sum-Wilsonloops}
Let $\Lambda\subset\Z^d$ be a finite lattice, with oriented edge set $E(\Lambda)$ and oriented plaquette set $P(\Lambda)$. Let $U$ be a random configuration with distribution $\mu_{\Lambda,Q}$. For any family of loops $\Lc=(\ell_1,\ldots,\ell_k)$ in the 1-skeleton of $\Lambda$, we have
\begin{equation}\label{eq:surface_sum}
\E[W_{\Lambda,\Lc}(U)]
=\frac1{Z_{\Lambda,Q}}\sum_{\alpha:P(\Lambda)\to\widehat{\U(N)}}\kappa_{\Lambda,Q}(\alpha)
\sum_{\Sigma\in S_{\Lambda,\Lc}(\alpha)}\Omega_N(\Sigma;\alpha)N^{\chi(\Sigma)-h(\Sigma)},
\end{equation}
where $S_{\Lambda,\Lc}(\alpha)$ is a finite set of spanning surfaces decorated by $\alpha$ with boundary prescribed by $\Lc$, $\chi(\Sigma)$ is its Euler characteristic, $h(\Sigma)\geq0$ is a local defect term, and $\Omega_N(\Sigma;\alpha)$ is an explicit coefficient.
\end{theorem}

This theorem will follow from a rewriting of all topological coefficients $\widehat{W}_{\Lambda,\Lc}(\alpha)$ as sums over decorated surfaces in the spirit of Weingarten calculus. It can be interpreted as an exact decorated-surface expansion, or an exact gauge/string duality, understood in the broad sense initiated by 't~Hooft's genus expansion \cite{Hoo74} and made geometrically explicit in two-dimensional Yang--Mills by Gross and Taylor, who interpreted the large-$N$ expansion as a sum over maps of surfaces to the target surface, with later refinements incorporating twists and Wilson loops \cite{GT,GT2,CMR}. The particular feature of our expansion is that the surfaces live in a random plaquette-decorated environment, and it is a quite complicated task to exchange the sum and the expectation in order to absorb the random environment inside the surface-sum. It is nonetheless possible for the Wilson action, and we will show how to recover a surface-sum formula by Cao--Park--Sheffield \cite{CPS25} in this specialization. Let us mention that other exact or convergent expansions were obtained in two dimensions by L\'evy \cite{Lev08} and Novak \cite{Nov2024} via Schur--Weyl duality and ramified coverings, and asymptotic expansions by the author together with Ma\"ida \cite{LemMai25,LM2} while in higher-dimensional lattice gauge theory the strong-coupling large-$N$ works of Chatterjee \cite{Cha19} and Jafarov \cite{Jaf16} provided different realizations of the same gauge/string philosophy. Theorem~\ref{thm:surface-sum-Wilsonloops} should be viewed as an exact finite-$N$, action-agnostic extension of this circle of ideas: it produces a surface-sum expansion for all smooth central plaquette actions, before any large-$N$ or strong-coupling limit is taken.

It is worth emphasizing that this theorem is actually the lattice Yang--Mills specialization of two more general results of independent interest, namely Theorem~\ref{prop:character_weingarten} and Proposition~\ref{prop:character_surface}, which hold for arbitrary products of rational irreducible characters of $\U(N)$ evaluated on words in independent Haar unitaries. Theorem~\ref{prop:character_weingarten} gives a refined Weingarten expansion in this general setting, and Proposition~\ref{prop:character_surface} recasts it as an exact decorated-surface expansion. In this sense, the gauge/string expansion of lattice Yang--Mills derived here should be viewed as a target-lattice incarnation of a broader topological-expansion phenomenon for unitary integrals, initiated by the works of Magee--Puder \cite{MagPud19,Mag2,MdlS26} and Dahlqvist \cite{Dah26}, and related in spirit to the map-enumeration and topological-recursion viewpoints developed for matrix integrals and unitary models in \cite{GuionnetMaurelSegala2006,CollinsGuionnetMaurelSegala2009,GuionnetNovak2015,BucDalche2026}, but formulated here as an exact finite-$N$ statement naturally adapted to our context.

\subsection{Spin-foam duality}

A second main theme of the paper is that Wilson loop expectations admit, besides the global surface expansion of Theorem~\ref{thm:surface-sum-Wilsonloops}, a genuinely local combinatorial description. The guiding idea is that once the plaquette action has been expanded in irreducible characters, each plaquette carries a highest-weight label, and the remaining problem is to understand how these local representation-theoretic degrees of freedom interact through the edge integrations. This naturally suggests a dual formulation in which plaquettes are viewed as local representation
variables and edges as local gluing constraints. The outcome is an exact local model for Wilson loop expectations, in which the observable is computed by summing over plaquette labels together with additional finite-dimensional channel data living on plaquette-edge incidences.

This point of view is close in spirit to the spin-foam and duality formulations of lattice gauge theory initiated by Oeckl--Pfeiffer \cite{OeckPfei01} and later developed by Oeckl \cite{Oec03}, in which non-Abelian lattice gauge theory is rewritten as a local state sum of representation data. It is also related to later work on dual computations and spin-foam methods \cite{CCK07}. The novelty of the present framework is that the local dual model is extracted after the universal spectral/topological splitting of Theorem~\ref{thm:state_sum_Wilson}. As a consequence, the action dependence is entirely carried by the spectral weights of the plaquette labels, while the topological coefficients themselves admit an exact finite-range representation in terms of local channel variables. In particular, Wilson loop expectations become ratios of two partition functions built from the same bulk plaquette-channel background; the additional Wilson scalarization data are localized on the finitely many edge vertices traversed by the loop family.

This formulation is complementary to the surface expansion: the latter is global and geometric, and it organizes the coefficients in terms of spanning surfaces and makes the gauge/string aspect of the theory transparent. By contrast, the channel representation isolates a genuinely local bulk model on the dual incidence graph, together with a finite Wilson defect supported on the loop family. It is therefore better adapted to locality questions.

\begin{theorem}[Spin-foam representation]\label{thm:spin-foam}
Let $\Lambda\subset \mathbb Z^d$ be a finite connected lattice, let $T\subset E(\Lambda)$ be a spanning tree, and let $\Lc=(\ell_1,\dots,\ell_k)$ be a finite family of loops in the $1$-skeleton of $\Lambda$. For every plaquette decoration $\alpha:P(\Lambda)\to \widehat{\U(N)},$ there exists a finite set $\mathcal C^\Lc_\Lambda(\alpha)$ of scalar local-channel fields, together with local plaquette coefficients $A_p^N$ and local edge kernels $K_e^N$, such that
\begin{equation}
\widehat W_{\Lambda,\Lc}(\alpha)=\sum_{\Gamma\in \mathcal C^\Lc_\Lambda(\alpha)}
\Bigl(\prod_{p\in P(\Lambda)}A_p^N(\Gamma_{\partial p};\alpha_p)\Bigr)
\Bigl(\prod_{e\in E(\Lambda)\setminus T}K_e^N(\Gamma_e;\alpha,\Lc)\Bigr).
\end{equation}
Moreover:
\begin{enumerate}
\item the plaquette coefficient $A_p^N$ depends only on the single plaquette label $\alpha_p$ and on the channel tuple around $p$;
\item for a fixed Wilson scalarization datum, the edge kernel $K_e^N$ depends only on the incident plaquette channels and on the local Wilson factor at the edge-vertex $e$;
\item all Wilson-dependent scalarization data are confined to the finite defect supported on $D(\Lc)$, while the plaquette coefficients are independent of $\Lc$.
\end{enumerate}
\end{theorem}

Equivalently, topological coefficients are computed by a finite-range bulk model on plaquette labels and incidence-channel data on the dual incidence graph, coupled to a finite tensor-network defect carried by the Wilson insertion. Define the \emph{defect partition function}
\begin{equation}\label{eq:defect-pf}
\mathcal{Z}_{\Lambda,Q}^{(\Lc)}:=\sum_{\alpha:P(\Lambda)\to \widehat{\U(N)}}\kappa_{\Lambda,Q}(\alpha)
\sum_{\Gamma\in \mathcal C^\Lc_\Lambda(\alpha)}\Bigl(
\prod_{p\in P(\Lambda)}A_p^N(\Gamma_{\partial p};\alpha_p)\Bigr)\Bigl(\prod_{e\in E(\Lambda)\setminus T}K_e^N(\Gamma_e;\alpha,\Lc)\Bigr),
\end{equation}
and the defect support
\[
D(\Lc):=\{e\in E(\Lambda)\setminus T:\text{ some loop of }\Lc\text{ traverses }e\}.
\]
Wilson loop expectations can then be expressed as ratios of such partition functions.

\begin{theorem}[Exact defect-ratio representation]\label{thm:spin-foam2}
For any finite loop family $\Lc$,
\begin{equation}
\mathbb E[W_{\Lambda,\Lc}(U)]=\frac{\mathcal{Z}_{\Lambda,Q}^{(\Lc)}}{\mathcal{Z}_{\Lambda,Q}^{(0)}}.
\end{equation}
The bulk plaquette-channel data are the same in numerator and denominator, and all additional Wilson-dependent data are localized on edge vertices in the defect support $D(\Lc)$; away from $D(\Lc)$ the local edge kernel is the vacuum one.
\end{theorem}

Theorem~\ref{thm:spin-foam2} may be compared with the worldsheet/background decompositions considered in the spin-foam literature on confinement, notably by Cherrington and collaborators \cite{CCK07}. There, Wilson loop observables are organized as sums over worldsheets interacting with vacuum spin-foam backgrounds, typically with non-trivial degeneracy issues. In contrast, the present defect-ratio formula is an exact identity at finite $N$: numerator and denominator share the same bulk plaquette-channel background, while the additional Wilson scalarization data and the modified edge kernels are supported on the defect set $D(\Lc)$. This avoids any non-canonical decomposition of charged configurations into background plus spanning surface, while retaining a fully local dual description.

Together, Theorems~\ref{thm:spin-foam} and~\ref{thm:spin-foam2} answer a question of Cao--Park--Sheffield~\cite[\S 7, (16)]{CPS25}: compared to the global surface-sum organization, the spin-foam representation makes the \emph{local properties} of lattice Yang--Mills explicit. Such locality might be naturally relevant for infinite-volume limits or scaling limits, and it is precisely one of the features exploited in the companion large-$N$ application discussed below.

\subsection{Master loop equation}

Master loop equations occupy a central place in lattice Yang--Mills theory because they are the basic exact recursions satisfied by Wilson observables. In the Wilson-action setting, such equations lie behind the strong-coupling and large-$N$ analyses of Chatterjee and Jafarov \cite{Cha19,Jaf16}. Shen--Smith--Zhu \cite{SSZ24} recently rederived them at finite $N$ for $\SO(N)$, $\U(N)$ and $\SU(N)$ through a Langevin-dynamic approach. These results show that the Wilson action admits a remarkably rigid cut-and-join/deformation recursion directly at the level of Wilson loop expectations. From that point of view, the classical finite-$N$ master loop equation is already one of the fundamental exact identities of the subject.

Theorem~\ref{thm:state_sum_Wilson} reveals these equations as instances of a broader topological and algebraic recursion. Rather than starting from a particular action and working directly with Wilson expectations, we first pass to the coefficient side of the state-sum expansion and prove a universal integration-by-parts identity for the topological coefficients $\widehat W_{\Lambda,\Lc}(\alpha)$. The result is a coefficientwise master loop equation that closes on the family
\[
\bigl\{\widehat W_{\Lambda,\Lc}(\alpha): \Lc \text{ loop family},\ \alpha:P(\Lambda)\to\widehat{\U(N)}\bigr\}.
\]
It involves two local operators $\mathscr{L}_e$ and $\mathscr{B}_{e,p}$ that we will describe in Section~\ref{sec:master-loop}; the former is the standard cut-and-join operator and acts on loops directly, and the latter is a local loop--plaquette recoupling operator: it performs a local surgery on the loop family and a finite recoupling of the active plaquette label.

\begin{theorem}[Universal coefficientwise master loop equation]
\label{thm:universal-coefficientwise-master-equation}
For every oriented edge $e\in E(\Lambda)$, every loop sequence $\Lc$, and every plaquette decoration
$\alpha:P(\Lambda)\to \widehat{\U(N)}$, one has
\begin{equation}\label{eq:master-equation}
(\mathscr L_e \widehat{W}_{\Lambda,\Lc})(\alpha)+\sum_{p\ni e}( \mathscr{B}_{e,p}\widehat{W}_{\Lambda,\Lc})(\alpha)=0.
\end{equation}
In particular, it closes on all topological coefficients.
\end{theorem}

Although it seems more natural to us to state this master loop equation at the level of topological coefficients, it also admits a resummed form. Multiplying~\eqref{eq:master-equation} by the spectral weight $\kappa_{\Lambda,Q}(\alpha)$, summing over plaquette decorations and using~\eqref{eq:State_sum_Wilson}, we obtain
\begin{equation}
\mathscr L_e \mathbb E[W_{\Lambda,\Lc}(U)]
+\frac1{Z_{\Lambda,Q}}\sum_{p\ni e}\sum_{\alpha:P(\Lambda)\to\widehat{\U(N)}}
\kappa_{\Lambda,Q}(\alpha)(\mathscr B_{e,p}\widehat{W}_{\Lambda,\Lc})(\alpha)=0.
\end{equation}

We see that there are both a loop part, applied on Wilson loop expectation $W_{\Lambda,\Lc}(U)$, and a Fourier part, applied to the topological coefficients $\widehat{W}_{\Lambda,\Lc}(\alpha)$. The former involves the familiar cut-and-join operator, and the latter involves a local recoupling operator on the active Fourier label $\alpha_p$.

This point of view is important for two reasons. First, it shows that the mechanism behind the master loop equation is more universal than the classical Wilson-action formulas suggest: what is truly fundamental is a hybrid position/Fourier recursion, in which the loop sector stays geometric while the plaquette sector is transferred to representation space. Second, it opens the door to non-Wilson specializations. For the Wilson action, the local recoupling operator collapses after resummation to the usual deformation terms; for the heat kernel action, by contrast, it remains visible as an explicit spectral transfer operator on the active plaquette label. In this sense, Theorem~\ref{thm:universal-coefficientwise-master-equation} should be viewed as a common finite-$N$ ancestor of the usual lattice master loop equations and of the more spectral Makeenko--Migdal-type identities that arise in two-dimensional Yang--Mills \cite{SSZ25}.

\subsection{Other classical groups}

Although this paper is deliberately focused on the unitary group $\U(N)$, the representation-theoretic architecture is not intrinsically unitary. The corresponding tensor models and commutant algebras suggest analogues for $\SU(N)$, $\SO(N)$ and $\Sp(N)$. Appendix~\ref{sec:appendix} records the precise ingredients that would be needed for such extensions.

\subsection{First application: the heat-kernel master field}

The universal structure developed here already has a concrete non-Wilson application. In the companion paper~\cite{Lem26prep}, the state-sum splitting, the projection-supported local-channel representation and the coefficientwise master loop equation are combined with large-$N$ heat-kernel estimates from~\cite{LemMai25,LM2} to construct the heat-kernel master field on $\Z^d$, for every $d\ge2$, at strong coupling. That work proves a uniform all-order $1/N$-expansion with exponentially local infinite-volume coefficients, leading-order factorization, and an area-law upper bound for the resulting master field. Thus the separation achieved in the present paper is not only formal: it isolates an action-independent combinatorial layer that can be coupled to action-specific spectral coercivity in order to solve a genuinely different lattice action.

\subsection{Organization of the paper}

First, we will develop in Section~\ref{sec:surface-sum} the refined Weingarten construction of spanning surfaces from integrals of products of rational characters of $\U(N)$, leading to the proof of Theorem~\ref{thm:surface-sum-Wilsonloops}. Second, we will present in Section~\ref{sec:local-channel} the model of local channels based on the dual incidence graph, and prove Theorems~\ref{thm:spin-foam} and~\ref{thm:spin-foam2}. Third, we will prove the master loop equation (Theorem~\ref{thm:universal-coefficientwise-master-equation}) in Section~\ref{sec:master-loop}. Finally, we will recover existing results about the Wilson action in Section~\ref{sec:specializations}, and record in Appendix~\ref{sec:appendix} the group-specific inputs required to extend the framework to other compact classical groups.

\subsection*{Acknowledgements}

The author would like to thank Nguyen Viet Dang, Thierry L\'evy, Elias Nohra for various discussions, and in particular Antoine Dahlqvist for discussions about mixed Schur--Weyl duality that unlocked many ideas leading to the present paper.

\section{Decorated-surface expansion}\label{sec:surface-sum}

In this section, we introduce the combinatorial toolbox required to prove Theorem~\ref{thm:surface-sum-Wilsonloops}. We will start with a general procedure, consisting in the integration of products of characters of rational irreducible representations of $\U(N)$, and which is a natural generalization of the Weingarten calculus that computes integrals of products of traces of tensor powers of the fundamental and contragredient representation. An important part of the arguments presented here have been initiated in the recent works \cite{Mag,Mag2,Dah26}, and we will extend them to a fairly general setting, before applying them to Wilson loop expectations. More precisely, Magee showed how some particular integrals of products of rational representations of $\U(N)$ can be rewritten as $N^\chi$-weighted sums over surfaces with boundary, whose topology can then be controlled using geometric arguments on words in surface groups. On the other hand, Dahlqvist recast this strategy in the language of mixed Schur--Weyl duality, traceless mixed tensors, and walled Brauer algebras, obtaining a surface-sum formalism similar to Magee's, but with a slightly different control.

Our goal is to show that the arguments of \cite{Mag2,Dah26} can be naturally extended to \emph{any product} of irreducible representations of $\U(N)$ and \emph{any words}, thereby yielding Theorem~\ref{prop:character_weingarten} and Proposition~\ref{prop:character_surface}. Then, we will apply these techniques to the particular setting of the present paper, which is the study of topological coefficients~\eqref{eq:topological-coefficient}. It will lead us to Theorem~\ref{thm:fixed-alpha-punctured-sum}, and finally to the proof of Theorem~\ref{thm:surface-sum-Wilsonloops}.

\subsection{Irreducible representations of $\mathrm{U}(N)$}

The irreducible representations of $\U(N)$ are in bijection with \emph{highest weights} (or \emph{signatures})
\[
\lambda=(\lambda_1,\ldots,\lambda_N)\in\Z^N,\quad \text{such that}\quad \lambda_1\geq\ldots\lambda_N.
\]
We will therefore identify canonically $\widehat{\U(N)}$ with the set of highest weights. Given any highest weight $\lambda\in\widehat{\U}(N)$, we will denote by $(\rho_\lambda,V_\lambda)$ the corresponding irreducible representation $\rho_\lambda\in\Hom(\U(N),\End(V_\lambda))$, which has finite dimension $d_\lambda=\dim V_\lambda$. We will also denote by $\chi_\lambda:G\to\C$ the corresponding \emph{character}:
\begin{equation}
\chi_\lambda(U)=\Tr_{V_\lambda}(\rho_\lambda(U)),\qquad \forall U\in\U(N).
\end{equation}

Let $n\geq 0$ be an integer. A \emph{partition of size $n$} is a finite nonincreasing family of positive integers $\alpha=(\alpha_1,\ldots,\alpha_r)$ such that $\vert\alpha\vert:=\sum_i \alpha_i=n$ (by convention, the empty partition $\varnothing$ is the only partition of size 0). We call $\vert\alpha\vert$ its \emph{size} and $\ell(\alpha):=r$ its \emph{length}. As much as highest weights label irreducible representations of $\U(N)$, the partitions of size $n$ label irreducible representations of the symmetric group $S_n$. We shall denote by $\chi^\alpha$ the irreducible characters of $S_n$, with an exponent instead of a subscript to distinguish with the characters of $\U(N)$. Given two integer partitions $\lambda^+,\lambda^-$, with respective sizes $n^+$ and $n^-$ (not necessarily equal), one can always construct a highest weight
\begin{equation}
[\lambda^+,\lambda^-]_N=(\lambda_1^+,\ldots,\lambda_{\ell(\lambda^+)}^+,0,\ldots,0,-\lambda_{\ell(\lambda^-)}^-,\ldots,-\lambda_1^-),
\end{equation}
for any $N\geq \ell(\lambda^+)+\ell(\lambda^-)$. All such highest weights are particular cases of irreducible representations of $\U(N)$, but also correspond to \emph{rational representations} (or \emph{stable representations}) of $\GL(N,\C)$ \cite{Lit1,Lit2,Sta84}. The interplay between rational representations and general irreducible representations of $\U(N)$ has been investigated in the large-$N$ regime by the author \cite{Lem22,Lem25} together with Dahlqvist \cite{DL23} and Ma\"ida \cite{LemMai25,LM2}, but also by Magee \cite{Mag,Mag2,MdlS26} and Dahlqvist \cite{Dah26} alone. We will not elaborate on the large-$N$ implication in the present paper, but mention that every irreducible representation of $\U(N)$ can be seen as a rational representation, by splitting the highest weight $\lambda\in\widehat{\U(N)}$ into a positive part $\lambda^+$ and a negative part $\lambda^-$ and discarding the possible zero coefficients. By construction, in this case, we always have $\ell(\lambda^+)+\ell(\lambda^-)\leq N$. We will see in the next subsection that the study of rational representations of $\U(N)$ is actually a manifestation of Schur--Weyl duality.

\subsection{Mixed Schur--Weyl duality and walled Brauer diagrams}\label{sec:mixed-schur-weyl}

Let $V=\C^N$ be the standard/fundamental representation of $\U(N)$ and $V^*=(\C^N)^*$ its dual. For any $n,m\geq 1$, we define the space of \emph{mixed tensors} $T_{n,m}=V^{\otimes n}\otimes {(V^*)}^{\otimes m}$. There are two actions, one of the unitary group $\U(N)$ and one of the product $S_{n,m}=S_n\times S_m$ of symmetric groups, which are given by
\[
U\cdot(v_1\otimes \ldots\otimes v_n\otimes v_1^*\otimes\ldots\otimes v_m^*)=(Uv_1\otimes\ldots\otimes Uv_n\otimes v_1^*(U^{-1})\otimes\ldots\otimes v_m^*(U^{-1})),
\]
\[
(\alpha,\beta)\cdot(v_1\otimes \ldots\otimes v_n\otimes v_1^*\otimes\ldots\otimes v_m^*)=(v_{\alpha^{-1}(1)}\otimes\ldots\otimes v_{\alpha^{-1}(n)}\otimes v_{\beta^{-1}(1)}\otimes\ldots\otimes v_{\beta^{-1}(m)}).
\]
We will denote respectively by $\rho_{n,m}:\U(N)\to\End(T_{n,m})$ and $\rho_N:S_{n,m}\to\End(T_{n,m})$ the representations of $\U(N)$ and $S_{n,m}$ on $T_{n,m}$ corresponding to these actions.

A convenient way to describe the decomposition of mixed tensors is through the \emph{walled Brauer algebra}, which is a subalgebra of the standard Brauer algebra. We briefly recall them, keeping things short because we only use them here as finite diagrammatic bases for equivariant endomorphisms.

\begin{definition}
Let $r\ge 1$. A \emph{Brauer diagram on $r$ strands} is a perfect matching of the $2r$ vertices $\{1,\dots,r\}\times\{+,-\},$ which we picture as $r$ \emph{top} vertices $(1,+),\dots,(r,+)$ and $r$ \emph{bottom} vertices $(1,-),\dots,(r,-)$, joined by $r$ pairwise disjoint edges. Equivalently, a Brauer diagram is an involution without fixed point on $\{1,\dots,r\}\times\{+,-\}$.

We denote by $\mathcal B_r$ the finite set of Brauer diagrams on $r$ strands. An edge joining a top vertex to a bottom vertex is called \emph{vertical}, while an edge joining two top vertices or two bottom vertices is called \emph{horizontal}.
\end{definition}

\begin{definition}
Fix a parameter $\delta\in\C$. The \emph{Brauer algebra} $\Br_r(\delta)$ is the $\C$-vector space with basis $\mathcal B_r$, endowed with the multiplication defined as follows.

For $D_1,D_2\in\mathcal B_r$, place $D_1$ above $D_2$ and identify the bottom row of $D_1$ with the top row of $D_2$. This produces:
\begin{itemize}
\item a new Brauer diagram $D_1\circ D_2\in\mathcal B_r$, obtained by reading the connections between the top row of $D_1$ and the bottom row of $D_2$;
\item a certain number $\ell(D_1,D_2)\ge 0$ of closed loops created in the middle row.
\end{itemize}
One then sets
\[
D_1D_2:=\delta^{\ell(D_1,D_2)}(D_1\circ D_2),
\]
and extends this multiplication bilinearly to all of $\Br_r(\delta)$.
\end{definition}

\begin{figure}[h!]
    \centering
    \includegraphics[width=0.5\linewidth]{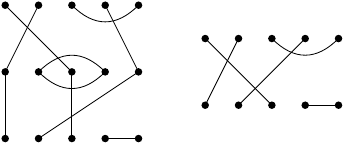}
    \caption{\small A concatenation of two Brauer diagrams (on the left), and the resulting Brauer diagram after simplification (on the right). Since one loop is erased, the concatenation produces a factor $\delta$.}
    \label{fig:brauer-1}
\end{figure}

\begin{remark}
The symmetric group $S_r$ embeds into $\Br_r(\delta)$ as the subset of Brauer diagrams having only vertical strands: a permutation $\sigma\in\mathcal{S}_r$ is identified with the diagram joining $(i,+)$ to $(\sigma(i),-)$.
\end{remark}

\begin{definition}
Let $n,m\ge 0$, not both zero. A \emph{walled Brauer diagram of type $(n,m)$} is a perfect matching of the $2(n+m)$ vertices
\[
\{-m,\dots,-1,1,\dots,n\}\times\{-1,1\},
\]
subject to the rule that vertical strands preserve the side of the wall, whereas horizontal strands cross the wall.

Equivalently, it is an involution $\pi$ without fixed point on
\[
\{-m,\dots,-1,1,\dots,n\}\times\{-1,1\}
\]
such that, whenever $\pi(v,\eta)=(v',\eta')$, one has
\[
\operatorname{sign}(v')\eta'=-\operatorname{sign}(v)\eta.
\]
We denote by $\mathcal \B_{n,m}$ the set of such diagrams.

\begin{figure}[h!]
    \centering
    \includegraphics[width=0.3\linewidth]{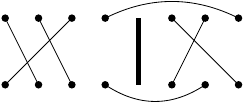}
    \caption{\small A walled Brauer diagram of type $(4,3)$.}
    \label{fig:brauer-2}
\end{figure}

As before, strands joining one top vertex to one bottom vertex are called \emph{vertical}, and strands joining two top vertices or two bottom vertices are called \emph{horizontal}.
\end{definition}

\begin{remark}
The condition
\[
\operatorname{sign}(v')\eta'=-\operatorname{sign}(v)\eta
\]
means exactly that vertical strands stay on the same side of the wall, whereas horizontal strands necessarily connect one positive label and one negative label. This is the mixed-tensor diagrammatics appropriate to $V^{\otimes n}\otimes (V^*)^{\otimes m}$. 
\end{remark}

\begin{definition}
Fix $\delta\in\C$. The \emph{walled Brauer algebra of type $(n,m)$} is the $\C$-vector space
\[
\Br_{n,m}(\delta):=\Span_\C\big(\mathcal \B_{n,m}\big),
\]
with multiplication defined by the same concatenation rule as for the ordinary Brauer algebra.

More precisely, if $D_1,D_2\in\mathcal \B_{n,m}$, stacking $D_1$ above $D_2$ and identifying the intermediate row produces a new walled Brauer diagram $D_1\circ D_2$ together with a number $\ell(D_1,D_2)$ of closed loops in
the middle row, and one sets
\[
D_1D_2:=\delta^{\ell(D_1,D_2)}(D_1\circ D_2).
\]
This extends bilinearly to $\Br_{n,m}(\delta)$.
\end{definition}

By default, we shall work with $\delta:=N$, and we will denote by $\Br_{n,m}(N)$ the corresponding Brauer algebra.

\begin{remark}
The subgroup $S_{n,m}=S_n\times S_m$ identifies with the subset of walled Brauer diagrams with no horizontal strands: the factor $S_n$ permutes the positive labels $\{1,\dots,n\}$, and the factor $S_m$ permutes the negative labels
$\{-m,\dots,-1\}$. In particular, $\Br_{n,m}(N)$ acts on the mixed tensor space $T_{n,m}$ and provides the standard diagrammatic model for the commutant of the $\U(N)$-action.
\end{remark}

There are two important operations that are not really endomorphisms of $T_{n,m}$. For any $1\leq i\leq n$ and $1\leq j\leq m$, the \emph{contraction} $c_{i,j}:T_{n,m}\to T_{n-1,m-1}$ is the pairing of the $i$th $V$-slot with the $j$th $V^*$-slot, and corresponds to an edge between the top vertices $(i,+)$ and $(j,+)$, and the \emph{coevaluation} $\iota_{i,j}:T_{n-1,m-1}\to T_{n,m}$ inserts the canonical tensor $\sum e_a\otimes e_a^*$ into the $i$th $V$-slot and the $j$th $V^*$-slot, and corresponds to an edge between the bottom vertices $(i,-)$ and $(j,-)$. Combining these yields a walled Brauer diagram $e_{i,j}\in\B_{n,m}$ that we will call \emph{Weyl contraction}. It is known (see e.g. \cite{Hal96}) that $\Br_{n,m}(N)$ is generated by transpositions and Weyl contractions. We shall also call \emph{covariant permutations} the permutations $(\sigma,\Id)\in S_{n,m}\subset\B_{n,m}$ and \emph{contravariant permutations} the permutations $(\Id,\tau)\in S_{n,m}\subset\B_{n,m}$. From all Weyl contractions, we can build an element $A_{n,m}\in\Br_{n,m}(N)$ of the algebra, defined by
\begin{equation}
A_{n,m}=\sum_{i=1}^n\sum_{j=1}^me_{i,j}.
\end{equation}
It is the image of the \emph{total contraction-coevaluation operator} $A=\sum_{i=1}^n\sum_{j=1}^m\iota_{i,j}c_{i,j}$ by the natural representation of $\Br_{n,m}(N)$ on $T_{n,m}$. This operator is studied by Goncharov \cite{Gon25}, who proved that $\Spec(A_{n,m})\subset\R^+$ and that $A_{n,m}$ is diagonalizable in $T_{n,m}$.

The \emph{mixed Schur--Weyl duality}, discovered by Koike \cite{Koi89}, states that the subspace
\[
\TT_{n,m}=\bigcap_{\substack{1\leq i\leq n\\ 1\leq j\leq m}}\ker(c_{i,j}),
\]
of \emph{traceless tensors} decomposes in terms of rational representations of $\U(N)$ and mixed representations of $S_{n,m}$:
\begin{equation}\label{eq:Koike}
\TT_{n,m}=\bigoplus_{\substack{\lambda\vdash n,\mu\vdash m\\
\ell(\lambda)+\ell(\mu)\leq N}}V^{[\lambda,\mu]}\otimes V_{[\lambda,\mu]_N}.
\end{equation}
We will consider three types of projectors:
\begin{itemize}
\item The $\U(N)\times S_{n,m}$-equivariant projector $P_N^{n,m}$ onto the space $\TT_{n,m}$ of traceless mixed tensors.
\item The isotypic projector $P_N^{[\lambda,\mu]}$ onto the summand associated to $(\lambda,\mu)$ in the decomposition~\eqref{eq:Koike}.
\item The symmetric group central idempotent $\Pi_{\lambda,\mu}\in\End(T_{n,m})$ defined by
\[
\Pi_{\lambda,\mu}=\frac{\chi^\lambda(1)\chi^\mu(1)}{n!m!}\sum_{(\sigma,\tau)\in S_{n,m}}\chi^\lambda(\sigma^{-1})\chi^\mu(\tau^{-1})\rho_N(\sigma,\tau),
\]
which is the projector on $T_{n,m}$ onto the $(\lambda,\mu)$-isotypic part inside the symmetric group side.
\end{itemize}
Because of the $S_{n,m}$-equivariance of $P_N^{n,m}$, we have in particular
\begin{equation}
P_N^{[\lambda,\mu]}=P_N^{n,m}\Pi_{\lambda,\mu}=\Pi_{\lambda,\mu}P_N^{n,m}.
\end{equation}
It follows that the character of the rational representation of $\U(N)$ associated to $(\lambda,\mu)$ is given by
\begin{equation}\label{eq:Koike-character}
\chi_{[\lambda,\mu]_N}(U)=\Tr_{T_{n,m}}\left(P_N^{[\lambda,\mu]}\rho_{n,m}(U)\right)=\Tr_{T_{n,m}}\left(P_N^{n,m}\Pi_{\lambda,\mu}\rho_{n,m}(U)\right),\quad \forall U\in\U(N).
\end{equation}

The last result we need before going further is the following, which is due to Goncharov \cite{Gon25}.

\begin{proposition}[\cite{Gon25}]\label{prop:goncharov}
For any $\lambda\vdash n,\mu\vdash m$, write $W_{\lambda,\mu}=\Pi_{\lambda,\mu}T_{n,m}$, and let
\[
I_{\lambda,\mu}=\Spec(A_{n,m}\vert_{W_{\lambda,\mu}})\setminus\{0\}.
\]
Then
\begin{equation}
P_N^{[\lambda,\mu]}=\Pi_{\lambda,\mu}\prod_{a\in I_{\lambda,\mu}}\left(1-\frac1a A_{n,m}\right).
\end{equation}
\end{proposition}
In fact, Goncharov proves a stronger statement, where $I_{\lambda,\mu}$ can be refined into a smaller and more precise subset of the spectrum of $A_{n,m}$, but we will not need it here. What matters to us is that it leads to an expansion of the projector $P_N^{[\lambda,\mu]}$.

\begin{corollary}\label{cor:PN-nm}
For any $\lambda\vdash n,\mu\vdash m$, there is a family $(c_N^{[\lambda,\mu]}(\tau))_{\tau\in\B_{n,m}}$ such that
\begin{equation}\label{eq:PN-nm}
P_N^{[\lambda,\mu]}=\sum_{\tau\in\B_{n,m}} c_N^{[\lambda,\mu]}(\tau)\rho_N(\tau).
\end{equation}
\end{corollary}

\begin{proof}
By Proposition~\ref{prop:goncharov}, $P_N^{[\lambda,\mu]}$ belongs to the image of $\Br_{n,m}(N)$ in $T_{n,m}$, which is spanned by the set $\B_{n,m}$ of Brauer diagrams. In the stable range $N\geq n+m$, $\B_{n,m}$ is a basis therefore there is a unique decomposition of the projector. Otherwise, $\B_{n,m}$ gives only a spanning family, so that there exist several decompositions, and we just need to choose one.
\end{proof}

A similar statement, which holds directly for $P_N^{n,m}$ can be found in \cite[Lemma 4.1]{Dah26}, but only in the stable range $N\geq n+m$, which is enough for a large-$N$ analysis but not a finite-$N$ expansion. The advantage of Dahlqvist's version is to provide a precise bound of the coefficients of the decomposition, which we do not have here.

\subsection{Weingarten expansion of the integral}

We will now perform a refined Weingarten expansion of all integrals of products of rational characters of $\U(N)$. The basic point, already emphasized by Magee \cite{Mag2}, is that a coarse expansion performed using the classical Weingarten calculus may hide large cancellations and may even exhibit apparently ``catastrophic'' contributions whose size is much too large to be compatible with the expected large-$N$ behavior. In Magee's approach, this is resolved by refining the Weingarten calculus through a projector onto the relevant traceless mixed-tensor component. Although we remain in the finite-$N$ regime here, this remains the appropriate strategy to follow. The first combinatorial object, popularized by Collins~\cite{Col03} and later by Collins--Sniady~\cite{ColSni06}, is the Weingarten function, whose definition is recalled below.

\begin{definition}
For $n\geq 0$, let $G_N^{(n)}$ be the matrix indexed by $S_n\times S_n$ defined by
\[
G_N^{(n)}(\sigma,\tau):=N^{\#(\sigma^{-1}\tau)},\qquad \sigma,\tau\in S_n,
\]
where $\#(\pi)$ denotes the number of cycles of the permutation $\pi$. The \emph{unitary Weingarten function} at rank $n$ is the central function
\[
\mathrm{Wg}_N^{(n)}:S_n\to \mathbb C
\]
whose matrix
\[
\bigl(\mathrm{Wg}_N^{(n)}(\sigma^{-1}\tau)\bigr)_{\sigma,\tau\in S_n}
\]
is the Moore--Penrose pseudoinverse of $G_N^{(n)}$. When there is no ambiguity on $n$, we simply write $\mathrm{Wg}_N(\pi)$ for $\mathrm{Wg}_N^{(n)}(\pi)$, $\pi\in S_n$.
In the stable range $n\leq N$, the matrix $G_N^{(n)}$ is invertible, so $\mathrm{Wg}_N^{(n)}$ is characterized by
\[
\sum_{\tau\in S_n} N^{\#(\sigma^{-1}\tau)}\,\mathrm{Wg}_N^{(n)}(\tau^{-1}\rho)=\delta_{\sigma,\rho}.
\]
\end{definition}

We will now introduce the main objects involved in the procedure:
\begin{itemize}
\item A finite family of words $w_1,\ldots,w_k\in\Fbb_m$ in the free group on generators $x_1,\ldots,x_m$. We always assume them in their \emph{reduced form}, so that their combinatorial length $\vert w_i\vert$ is minimal.
\item A finite family $(\lambda_i^+,\lambda_i^-)_{1\leq i\leq k}$ of pairs of partitions, denoting by $n_i^\pm=\vert \lambda_i^\pm\vert$ their size.
\end{itemize}
Given this data, we define $s=\sum_{i=1}^k(n_i^++n_i^-)\vert w_i\vert$, and three maps on $[s]=\{1,\ldots,s\}$:
\begin{enumerate}
\item An \emph{orientation map} $\varepsilon:[s]\to\{+,-\}$
\item A \emph{letter map} $a:[s]\to\{1,\ldots,m\}$,
\item A \emph{sign map} $\delta:[s]\to\{-1,1\}$.
\end{enumerate}
For each $i$, let $T_i:=T_{n_i^+,n_i^-}$ be the space of mixed tensors, $\rho_i=\rho_{n_i^+,n_i^-}:\U(N)\to\End(T_{n_i^+,n_i^-})$ be the natural mixed tensor representation, and
\[
P_i:=P_N^{[\lambda_i^+,\lambda_i^-]}=P_N^{n_i^+,n_i^-}\Pi_{\lambda_i^+,\lambda_i^-}
\]
be the isotypic projector involved in the mixed Schur--Weyl decomposition~\eqref{eq:Koike}. We have
\begin{equation}\label{eq:Koike_character}
\chi_{[\lambda_i^+,\lambda_i^-]_N}(g)=\Tr_{T_i}(P_i\rho_i(g)).
\end{equation}
The maps $\varepsilon,a,\delta$ record the elementary matrix factors arising from the mixed-tensor actions: for each slot they determine the letter and whether the corresponding factor is an entry of $U_a$ or of $\overline{U_a}$.  The precise row/column order depends additionally on whether the slot is covariant or contravariant and will be fixed explicitly below.  At the operator level we simply have
\begin{equation}\label{eq:character_tensor}
\prod_{i=1}^k \chi_{[\lambda_i^+,\lambda_i^-]_N}(w_i(U))
=\Tr_{\bigotimes_i T_i}\!\left(\bigotimes_{i=1}^k P_i\rho_i(w_i(U))\right).
\end{equation}

For each $i\in\{1,\dots,k\}$, write each word $w_i$ in an explicit form
\[
w_i=x_{a_{i,1}}^{\varepsilon_{i,1}}\cdots x_{a_{i,\ell_i}}^{\varepsilon_{i,\ell_i}},\qquad \ell_i:=|w_i|, \qquad r_i:=n_i^+ + n_i^-,
\]
and let
\[
T_i:=T_{n_i^+,n_i^-}=V^{\otimes n_i^+}\otimes (V^*)^{\otimes n_i^-}, \qquad \chi_{[\lambda_i^+,\lambda_i^-]_N}(g)=\Tr_{T_i}\bigl(P_i\rho_i(g)\bigr).
\]
Recall that by Corollary~\ref{cor:PN-nm} we have the projector decomposition
\[
P_i=\sum_{\tau\in \B_{n_i^+,n_i^-}} c_{i,N}(\tau)\rho_N(\tau),
\]
where $c_{i,N}(\tau):=c_N^{[\lambda_i^+,\lambda_i^-]}(\tau)$. Define the local Brauer index set $\mathcal I :=\prod_{i=1}^k \B_{n_i^+,n_i^-}^{\ell_i}.$ For $\tau=(\tau_{i,r})_{1\le i\le k,\ 1\le r\le \ell_i}\in \mathcal I$, set
\[
c_N(\tau):=\prod_{i=1}^k\prod_{r=1}^{\ell_i} c_{i,N}(\tau_{i,r}).
\]
Although it does not appear in the notation, note that $c_N(\tau)$ depends on all $\lambda_i^\pm$. Let $\mathcal S:=\{(i,r,u): 1\le i\le k,\ 1\le r\le \ell_i,\ 1\le u\le r_i\}.$ For $x=(i,r,u)\in\mathcal S$, define its letter by $a(x):=a_{i,r},$ and its matrix-entry sign by
\[
\delta(x):=
\begin{cases}
+1,&\text{if }(u\le n_i^+\text{ and }\varepsilon_{i,r}=+1)\text{ or }(u>n_i^+\text{ and }\varepsilon_{i,r}=-1),\\
-1,&\text{if }(u\le n_i^+\text{ and }\varepsilon_{i,r}=-1)\text{ or }(u>n_i^+\text{ and }\varepsilon_{i,r}=+1).
\end{cases}
\]
For each letter $a\in\{1,\dots,m\}$, set
\[
\mathcal{S}_a^\pm:=\{x\in\mathcal S:\ a(x)=a,\ \delta(x)=\pm1\}.
\]

\begin{theorem}\label{prop:character_weingarten}
If for some $a$ one has $|\mathcal{S}_a^+|\neq |\mathcal{S}_a^-|$, then
\[
\int_{\U(N)^m}\prod_{i=1}^k \chi_{[\lambda_i^+,\lambda_i^-]_N}(w_i(U))dU = 0.
\]
Otherwise, define $|\mathcal{S}_a^+|=|\mathcal{S}_a^-|=:p_a$ for every $a$. Choose orderings $\mathcal{S}_a^+=(x_{a,1}^+,\dots,x_{a,p_a}^+)$ and $\mathcal{S}_a^-=(x_{a,1}^-,\dots,x_{a,p_a}^-)$, and define $\mathcal J:=\prod_{a=1}^m (S_{p_a}\times S_{p_a})$. For fixed $\tau\in\mathcal I$, introduce cyclic multi-indices
\[
I_{i,r},J_{i,r}\in [N]^{r_i},\qquad 1\le i\le k,\ 1\le r\le \ell_i,
\]
with the convention $I_{i,\ell_i+1}=I_{i,1}$. For $x=(i,r,u)\in\mathcal S$, set
\[
j_x := (J_{i,r})_u,\qquad i_x := (I_{i,r+1})_u.
\]
Define the row and column indices of the matrix entry attached to $x$ by
\[
(\mathrm r(x),\mathrm c(x)):=
\begin{cases}
(j_x,i_x),&u\le n_i^+,\ \varepsilon_{i,r}=+1,\\
(i_x,j_x),&u\le n_i^+,\ \varepsilon_{i,r}=-1,\\
(j_x,i_x),&u>n_i^+,\ \varepsilon_{i,r}=+1,\\
(i_x,j_x),&u>n_i^+,\ \varepsilon_{i,r}=-1.
\end{cases}
\]
The distinction between the four cases is that the first and fourth contribute an entry of $U_{a(x)}$, whereas the second and third contribute an entry of $\overline{U_{a(x)}}$. Thus $\delta(x)=+1$ means that the factor is $U_{a(x),\mathrm r(x),\mathrm c(x)}$, while $\delta(x)=-1$ means that it is $\overline{U_{a(x),\mathrm r(x),\mathrm c(x)}}$.

For $\sigma=((\alpha_a,\beta_a))_{a=1}^m\in\mathcal J$, define
\[
\mathcal K_N(\tau,\sigma):=\sum_{\{I_{i,r},J_{i,r}\}}\left(\prod_{i=1}^k\prod_{r=1}^{\ell_i}(\rho_N(\tau_{i,r}))_{I_{i,r},J_{i,r}}\right)
\prod_{a=1}^m\prod_{t=1}^{p_a}
\delta^{\mathrm K}_{\mathrm r(x_{a,t}^+),\mathrm r(x_{a,\alpha_a(t)}^-)}
\delta^{\mathrm K}_{\mathrm c(x_{a,t}^+),\mathrm c(x_{a,\beta_a(t)}^-)},
\]
where $\delta^{\mathrm K}$ denotes the Kronecker delta. Then
\begin{equation}\label{eq:refined-character-weingarten}
\int_{\U(N)^m}\prod_{i=1}^k \chi_{[\lambda_i^+,\lambda_i^-]_N}(w_i(U))dU
=\sum_{\tau\in\mathcal I}c_N(\tau)\sum_{\sigma\in\mathcal J}\left(\prod_{a=1}^m {\mathrm{Wg}}_N(\alpha_a^{-1}\beta_a)\right)\mathcal K_N(\tau,\sigma).
\end{equation}
\end{theorem}

\begin{proof}
For each $1\leq i\leq k$ and $1\leq r\leq \ell_i$, set $R_{i,r}(U):=\rho_i(U_{a_{i,r}}^{\varepsilon_{i,r}})$. Since $P_i$ is $\U(N)$-equivariant, it commutes with every $R_{i,r}(U)$, and since $P_i$ is a projector, $P_i^2=P_i$. Hence
\[
\Tr_{T_i}\bigl(P_i R_{i,1}(U)\cdots R_{i,\ell_i}(U)\bigr)=\Tr_{T_i}\bigl(P_iR_{i,1}(U)P_iR_{i,2}(U)\cdots P_iR_{i,\ell_i}(U)\bigr).
\]
Expanding each copy of $P_i$ using~\eqref{eq:PN-nm}, we obtain
\[
\prod_{i=1}^k \chi_{[\lambda_i^+,\lambda_i^-]_N}(w_i(U))=\sum_{\tau\in\mathcal I}c_N(\tau)\prod_{i=1}^k\Tr_{T_i}\Bigl(\prod_{r=1}^{\ell_i}\rho_N(\tau_{i,r})R_{i,r}(U)\Bigr).
\]
Fix $\tau$. Expanding each trace in the standard basis of $T_i$ gives
\[
\Tr_{T_i}\Bigl(\prod_{r=1}^{\ell_i}\rho_N(\tau_{i,r})R_{i,r}(U)\Bigr)
=\sum_{\{I_{i,r},J_{i,r}\}}\prod_{r=1}^{\ell_i}(\rho_N(\tau_{i,r}))_{I_{i,r},J_{i,r}}(R_{i,r}(U))_{J_{i,r},I_{i,r+1}},
\]
with $I_{i,\ell_i+1}=I_{i,1}$. For a fixed slot $x=(i,r,u)$, the corresponding elementary factor is one of the following four matrix entries:
\[
\begin{array}{c|c|c}
\text{slot type}&\varepsilon_{i,r}&\text{matrix entry}\\ \hline
V&+1&(U_{a(x)})_{j_x,i_x}\\
V&-1&\overline{(U_{a(x)})_{i_x,j_x}}\\
V^*&+1&\overline{(U_{a(x)})_{j_x,i_x}}\\
V^*&-1&(U_{a(x)})_{i_x,j_x}.
\end{array}
\]
This is exactly the row/column convention in the statement. Grouping by the letter $a$, the fixed-$\tau$ integrand becomes a finite sum whose $U_a$-dependent factor is
\[
\prod_{t=1}^{|\mathcal S_a^+|}(U_a)_{\mathrm r(x_{a,t}^+),\mathrm c(x_{a,t}^+)}
\prod_{t=1}^{|\mathcal S_a^-|}\overline{(U_a)_{\mathrm r(x_{a,t}^-),\mathrm c(x_{a,t}^-)}}.
\]
If $|\mathcal S_a^+|\ne|\mathcal S_a^-|$ for some $a$, Haar invariance under multiplication by scalar matrices makes the $U_a$-integral vanish. Otherwise, the ordinary unitary Weingarten formula yields
\[
\int_{\U(N)}\prod_{t=1}^{p_a}(U_a)_{\mathrm r(x_{a,t}^+),\mathrm c(x_{a,t}^+)}
\prod_{t=1}^{p_a}\overline{(U_a)_{\mathrm r(x_{a,t}^-),\mathrm c(x_{a,t}^-)}}\,dU_a
\]
\[
=\sum_{\alpha_a,\beta_a\in S_{p_a}}\left(\prod_{t=1}^{p_a}
\delta^{\mathrm K}_{\mathrm r(x_{a,t}^+),\mathrm r(x_{a,\alpha_a(t)}^-)}
\delta^{\mathrm K}_{\mathrm c(x_{a,t}^+),\mathrm c(x_{a,\beta_a(t)}^-)}\right)\mathrm{Wg}_N(\alpha_a^{-1}\beta_a).
\]
Multiplying over $a=1,\dots,m$ and interchanging the finite sums gives~\eqref{eq:refined-character-weingarten}.
\end{proof}

One expects Theorem~\ref{prop:character_weingarten} to admit further vanishing criteria beyond the ordinary balance condition
\[
|\mathcal{S}_a^+|=|\mathcal{S}_a^-|\qquad (a=1,\dots,m).
\]
For instance, in a similar setting, Magee \cite{Mag2} proves that in his refined calculus certain classical Weingarten matchings vanish identically because the relevant projector acts on the traceless mixed-tensor space: when a local pairing forces a mixed contraction incompatible with tracelessness, the corresponding contribution is zero. This is summarized in his ``forbidden matching property.'' The same mechanism holds here.

\begin{proposition}\label{prop:elementary-forbidden}
Keep the notation of Theorem~\ref{prop:character_weingarten}. For $\sigma=((\alpha_a,\beta_a))_{a=1}^m\in\mathcal J$,
define
\[
\widetilde K_N(\sigma):=\sum_{\tau\in \mathcal I} c_N(\tau)\mathcal K_N(\tau,\sigma).
\]
Suppose that there exist a letter $a\in\{1,\dots,m\}$ and an index $t\in\{1,\dots,p_a\}$ such that, writing
\[
x_{a,t}^+=(i,r,u)\in\mathcal S_a^+,\qquad x_{a,\alpha_a(t)}^-=x_{a,\beta_a(t)}^-=(i,r,v)\in\mathcal S_a^-,
\]
the positive occurrence $x_{a,t}^+$ and the two negative occurrences paired with it by $\alpha_a$ and $\beta_a$ all come from the same local occurrence $(i,r)$. Then $\widetilde K_N(\sigma)=0.$
\end{proposition}

\begin{proof}
For each $i\in\{1,\dots,k\}$,
the image of $P_i=\Pi_i P_N^{n_i^+,n_i^-}$ is contained in $\TT_{n_i^+,n_i^-}$ because $P_N^{n_i^+,n_i^-}$ is the projector onto $\TT_{n_i^+,n_i^-}$. Hence, for every covariant slot $p$ and contravariant slot $q$, $c_{p,q}P_i=0.$ By self-adjointness, this also gives $P_i\iota_{p,q}=0,$ where $\iota_{p,q}$ denotes the corresponding coevaluation operator.

Undoing the Brauer expansion in the proof of Theorem~\ref{prop:character_weingarten}, the quantity $\widetilde K_N(\sigma)$ is obtained by keeping the factors $P_i$ unexpanded. If the configuration above occurs, then the two row/column Kronecker identifications corresponding to the chosen pair $(a,t)$ identify the local covariant and contravariant indices in such a way that, at the occurrence $(i,r)$, one inserts either a mixed contraction $c_{u,v}$ or its adjoint coevaluation $\iota_{u,v}$ between two consecutive copies of $P_i$. After summing over the intermediate indices, the local contraction therefore contains either the factor $c_{u,v}P_i$ or the factor $P_i\iota_{u,v}$.  Both vanish: the first because the image of $P_i$ is contained in the traceless mixed-tensor space, and the second by adjunction. Hence every summand is zero, so $\widetilde K_N(\sigma)=0$.
\end{proof}

This gives a first rigorous forbidden-matching criterion. A systematic exploitation of these cancellation mechanisms would be interesting and perhaps useful for a large-$N$ analysis.

\begin{remark}
As pointed out in \cite[Section 3]{Mag2}, Theorem~\ref{prop:character_weingarten} contains the ordinary unitary Weingarten calculus as a special case. Indeed, if there is no nontrivial mixed-tensor projector insertion, then the Brauer side becomes trivial: only the identity contribution survives in each local Brauer factor, the sum over $\tau$ collapses, and the formula reduces to a pure sum over the permutation-pair data $\sigma=((\alpha_a,\beta_a))_{a=1}^m$ with weights $\prod_{a=1}^m \mathrm{Wg}_N(\alpha_a^{-1}\beta_a),$ which is exactly the classical unitary Weingarten expansion for expectations of traces of words in independent Haar unitaries.
\end{remark}

\subsection{Construction of spanning surfaces}\label{sec:construction-spanning-surfaces}

In this subsection we reinterpret the Weingarten expansion of Theorem~\ref{prop:character_weingarten} in geometric terms. Let us recall a few notations: for each $i\in\{1,\dots,k\}$, write
\[
w_i=x_{a_{i,1}}^{\varepsilon_{i,1}}\cdots x_{a_{i,\ell_i}}^{\varepsilon_{i,\ell_i}},\qquad\ell_i:=|w_i|,\qquad r_i:=n_i^+ + n_i^-,
\]
and $\mathcal I=\prod_{i=1}^k \B_{n_i^+,n_i^-}^{\ell_i},$ $\mathcal J=\prod_{a=1}^m S_{p_a,p_a}.$ We write
\[
\tau=(\tau_{i,r})_{1\le i\le k,\ 1\le r\le \ell_i}\in \mathcal I,\qquad\sigma=((\alpha_a,\beta_a))_{a=1}^m\in \mathcal J.
\]
From this data we will construct polygons associated to each word $w_i$ and each tensor slot $u\in\{1,\ldots,r_i\}$, as well as ``Brauer gadgets'' associated to $\tau$ that attach ribbons to the polygons according to rules dictated by Theorem~\ref{prop:character_weingarten}. Then, these combinations of polygons and Brauer gadgets are glued together by means of $\sigma$, following Haar integration over all remaining edge variables.

In order to illustrate the procedure, we will describe a simple example in full details: 

\begin{example}\label{ex:surface}
Consider the integral
\[
\int_{\U(N)^2}\chi_{[(1),(1)]_N}(xyx^{-1}y^{-1})dxdy.
\]
We will describe, for each step of the construction, how it works for this specific example.
\end{example}

\medskip

{\bf Boundary polygons.} For each $i\in\{1,\dots,k\}$ and each slot $u\in\{1,\dots,r_i\}$, let $P_{i,u}$ be an
oriented polygon with $\ell_i$ sides
\[
e_{i,u,1},\dots,e_{i,u,\ell_i},
\]
read in cyclic order according to the word $w_i$. The side $e_{i,u,r}$ is labelled by the $r$-th letter $x_{a_{i,r}}^{\varepsilon_{i,r}}$ of $w_i$, together with the slot index $u$. We denote by $P:=\bigsqcup_{i=1}^k \bigsqcup_{u=1}^{r_i} P_{i,u}$ the resulting disjoint union of slot-polygons. Thus the set of boundary edges of $P$ is naturally indexed by
\[
\mathcal S=\{(i,r,u):1\le i\le k,\ 1\le r\le \ell_i,\ 1\le u\le r_i\},
\]
and the cyclic order around the polygons defines the implicit boundary-order datum $\rho$. The polygons corresponding to Example~\ref{ex:surface} are displayed in Figure~\ref{fig:ikea}.

\begin{figure}[h!]
    \centering
    \includegraphics[width=0.5\linewidth]{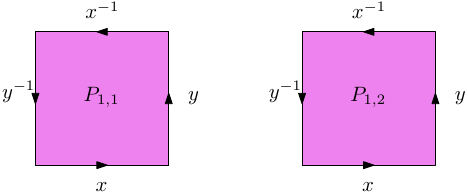}
    \caption{\small Boundary polygons from Example~\ref{ex:surface}. There are two polygons because the representation carries a one-dimensional covariant and a one-dimensional contravariant tensor.}
    \label{fig:ikea}
\end{figure}

\medskip

{\bf Brauer gadgets.} Fix $\tau\in\mathcal I$. For each pair $(i,r)$, let $G_{i,r}(\tau_{i,r})$ be a CW complex obtained as follows: given the walled Brauer diagram associated with $\tau_{i,r}$, we thicken all strands of the diagram into rectangles. It is a compact band complex with $2r_i$ distinguished boundary intervals
\[
p^{\mathrm{in}}_{i,r,1},\dots,p^{\mathrm{in}}_{i,r,r_i},\qquad p^{\mathrm{out}}_{i,r,1},\dots,p^{\mathrm{out}}_{i,r,r_i},
\]
one incoming and one outgoing interval for each slot $u\in\{1,\dots,r_i\}$, obtained by thickening every strand of the Brauer diagram into a rectangle. The horizontal strands give rise to the contraction/coevaluation bands.

We write $h(\tau_{i,r})$ for the number of horizontal bands in $G_{i,r}(\tau_{i,r})$, and we define the \emph{total defect}
\[
h(\tau):=\sum_{i=1}^k\sum_{r=1}^{\ell_i} h(\tau_{i,r}).
\]
In the case of Example~\ref{ex:surface}, there are only two possible Brauer diagrams, resulting in two Brauer gadgets, see Figure~\ref{fig:gadget}.

\begin{figure}[h!]
    \centering
    \includegraphics[width=0.5\linewidth]{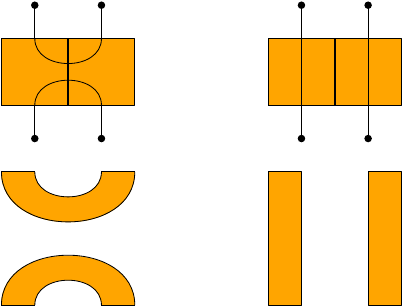}
    \caption{\small The two possible types of Brauer diagrams from Example~\ref{ex:surface} (above), and the corresponding Brauer gadgets (below).}
    \label{fig:gadget}
\end{figure}

\medskip

{\bf Construction of the surface.} Let us now detail the construction of the surface. It consists in three steps: we first glue the Brauer gadgets to boundary polygons, then we glue edges by performing Haar integration, and finally we cap the surface into a closed surface.

\medskip

\emph{Step 1. Brauer gluings.} Fix $\tau \in\mathcal I$. For every $i,r,u$, glue the incoming boundary interval $p^{\mathrm{in}}_{i,r,u}\subset \partial G_{i,r}(\tau_{i,r})$ to the boundary edge $e_{i,u,r}\subset \partial P_{i,u}$ by an orientation-reversing homeomorphism. Denote by $P(\tau)$ the resulting oriented surface with boundary. The opposite interval $p^{\mathrm{out}}_{i,r,u}$ remains exposed after the gluing; we denote the corresponding exposed boundary edge by
\[
\widetilde e_{i,u,r}\subset \partial P(\tau).
\]

\medskip

\emph{Step 2. Haar patches.} Fix $\sigma=((\alpha_a,\beta_a))_{a=1}^m\in \mathcal J$. For each letter $a\in\{1,\dots,m\}$, let
\[
\mathcal{S}_a^+=\{x_{a,1}^+,\dots,x_{a,p_a}^+\},\qquad \mathcal{S}_a^-=\{x_{a,1}^-,\dots,x_{a,p_a}^-\}
\]
be the ordered positive and negative occurrences of $a$ introduced in Theorem~\ref{prop:character_weingarten}. We define $H_a(\sigma_a)$, where $\sigma_a=(\alpha_a,\beta_a)$, to be the usual unitary Weingarten ribbon patch with boundary ports indexed by
\[
\widetilde e_{x_{a,1}^+},\dots,\widetilde e_{x_{a,p_a}^+}, \qquad \widetilde e_{x_{a,1}^-},\dots,\widetilde e_{x_{a,p_a}^-},
\]
and whose internal identifications encode the row and column Kronecker constraints
\[
\delta^{\mathrm K}_{\mathrm r(x_{a,t}^+),\mathrm r(x_{a,\alpha_a(t)}^-)},\qquad
\delta^{\mathrm K}_{\mathrm c(x_{a,t}^+),\mathrm c(x_{a,\beta_a(t)}^-)}
\qquad (1\le t\le p_a)
\]
appearing in Theorem~\ref{prop:character_weingarten}. We glue the boundary ports of $H_a(\sigma_a)$ to the corresponding exposed edges $\widetilde e_x\subset \partial P(\tau)$ by orientation-reversing homeomorphisms, compatible with the chosen orderings of $\mathcal{S}_a^\pm$. Performing this independently for every letter $a=1,\dots,m$ yields an oriented compact surface with boundary $\Sigma(\tau,\sigma).$ The gluing cell structure induced by the slot-polygons, the local Brauer gadgets and the Haar patches defines a topological map $\M(\tau,\sigma)\subset \Sigma(\tau,\sigma),$ whose set of faces decomposes canonically as $F(\M(\tau,\sigma)) = P \sqcup G \sqcup H,$ where:
\begin{itemize}
\item $P$ is the set of slot-polygons $P_{i,u}$,
\item $G$ is the set of rectangular faces coming from the local Brauer gadgets
$G_{i,r}(\tau_{i,r})$,
\item $H$ is the set of faces carried by the Haar patches $H_a(\sigma_a)$.
\end{itemize}
In the case of Example~\ref{ex:surface}, the general procedure is given in Figure~\ref{fig:pre-uncapped}.

\begin{figure}[h!]
    \centering
    \includegraphics[width=0.6\linewidth]{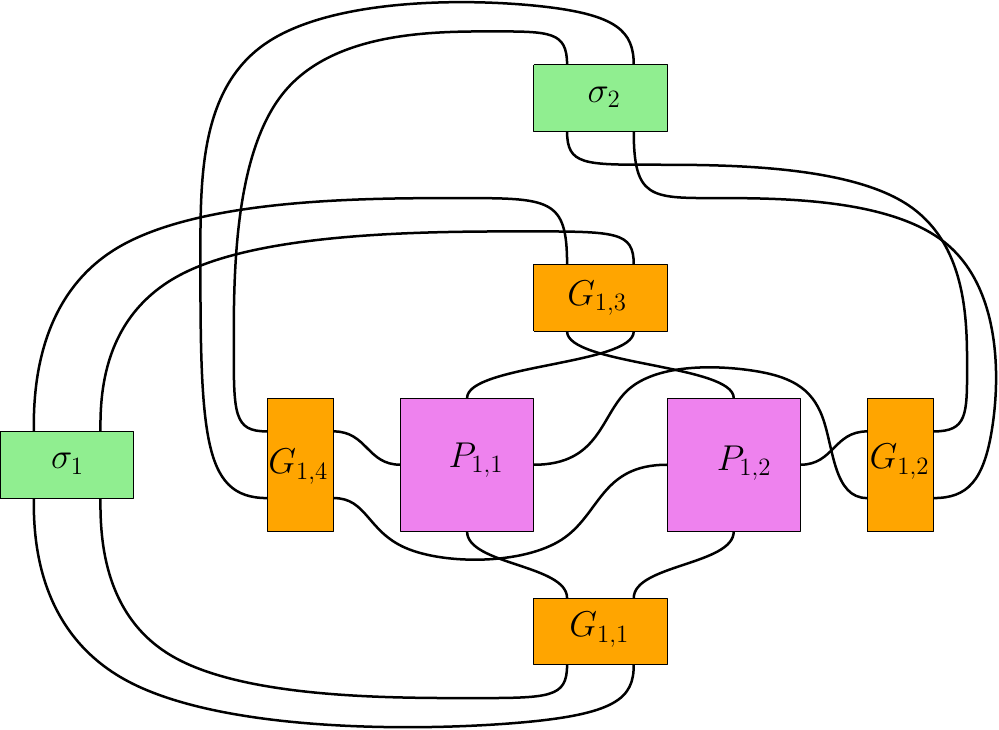}
    \caption{\small Schematic construction of an uncapped surface for the example~\ref{ex:surface}. Starting from each side of each boundary polygon (in purple, as in Figure~\ref{fig:ikea}), a strand goes through a Brauer projector (in orange, like those from Figure~\ref{fig:gadget}). The outgoing parts of strands associated to one edge and its inverse are then glued pairwise via a Haar patch (in green).}
    \label{fig:pre-uncapped}
\end{figure}

The construction is not finished there because the Brauer gadget and the Haar patches are not specified: there are as many possible surfaces as choices of each gadget and each patch. Each time, there are two choices of gadgets (as displayed in Figure~\ref{fig:gadget}) and two choices of patches (because they live in $S_2=\{\Id,(1\ 2)\}$). Figure~\ref{fig:uncapped} shows the construction of the uncapped surface from a given specification.

\begin{figure}[h!]
    \centering
    1. \includegraphics[width=0.52\linewidth]{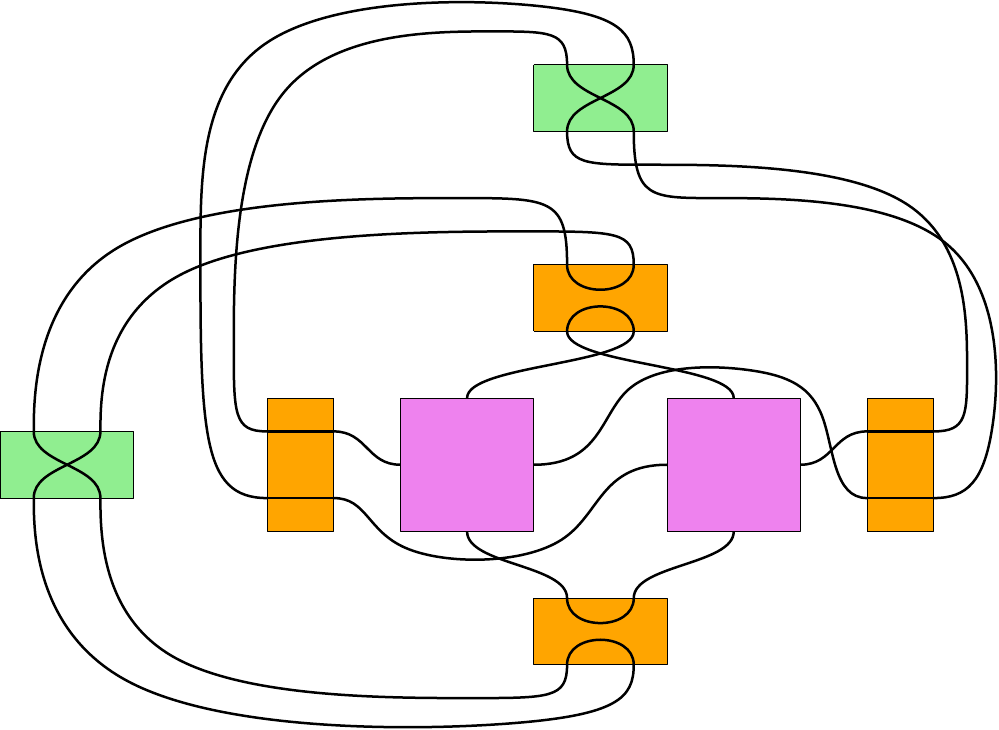}

    \smallskip

    2. \includegraphics[width=0.35\linewidth]{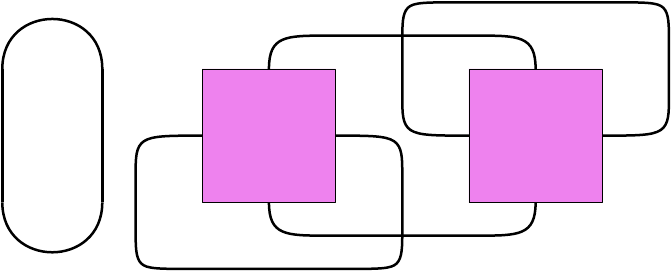}\hfill
    3. \includegraphics[width=0.35\linewidth]{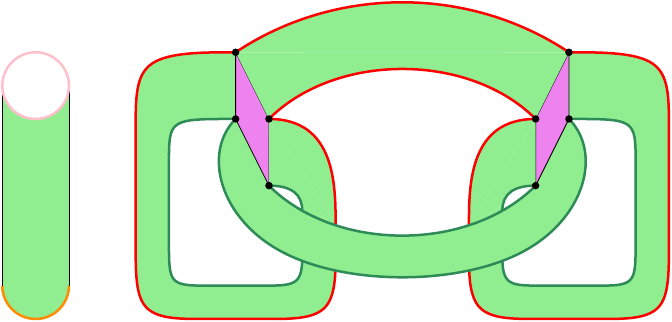}
    \caption{\small Construction of an uncapped surface. 1. For a given specification of Brauer projectors and Haar patches, the construction of Figure~\ref{fig:pre-uncapped} reveals complete strands. 2. Simplification of the diagram by shortening the strands. 3. Replacing the strands with bands (in green) yields the uncapped surface, which has 4 boundary components.}
    \label{fig:uncapped}
\end{figure}

\medskip

\emph{Step 3. Capping.} Let $b(\tau,\sigma):=\#\pi_0(\partial \Sigma(\tau,\sigma))$ be the number of boundary components of the refined surface. By attaching one capping disk to each boundary component, we obtain a closed oriented surface $\widehat\Sigma(\tau,\sigma).$ The cell structure extends to a completed topological map
\[
\widehat \M(\tau,\sigma)\subset \widehat\Sigma(\tau,\sigma),
\]
whose set of faces decomposes as
\[
F(\widehat \M(\tau,\sigma)) = P \sqcup G \sqcup H \sqcup C,
\]
where $C$ is the set of capping disks. For Example~\ref{ex:surface}, the uncapped surface obtained in Figure~\ref{fig:uncapped} is completed into a closed oriented surface with two connected components, of genera $0$ and $1$, see Figure~\ref{fig:capped}.

\begin{figure}[h!]
    \centering
    \includegraphics[width=0.7\linewidth]{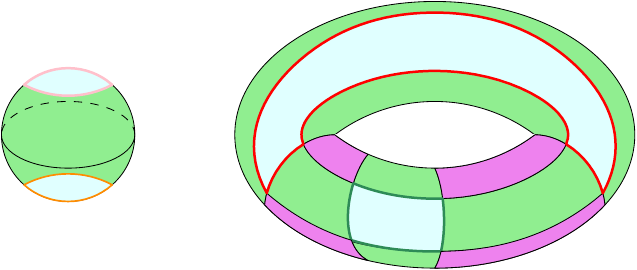}
    \caption{\small The capped surface obtained from the uncapped one considered in Figure~\ref{fig:uncapped}: it has two connected components, respectively of genus 0 and 1.}
    \label{fig:capped}
\end{figure}

\medskip

We call \emph{uncapped pair} a pair $(\M(\tau,\sigma),\Sigma(\tau,\sigma))$ where $\M(\tau,\sigma)$ is a topological map embedded in $\Sigma(\tau,\sigma)$. Two uncapped pairs $(\M(\tau,\sigma), \Sigma(\tau,\sigma)),\ (\M(\tau',\sigma'), \Sigma(\tau',\sigma'))$, are called \emph{decorated-isomorphic} if there exists an orientation-preserving homeomorphism
\[
f:\Sigma(\tau,\sigma)\to \Sigma(\tau',\sigma')
\]
carrying $\M(\tau,\sigma)$ onto $\M(\tau',\sigma')$, preserving the decomposition of faces into $P$, $G$, $H$, and preserving all labels induced by $(i,r,u)$, by the letters $a_{i,r}$, and by the chosen orderings of the sets $\mathcal{S}_a^\pm$. Similarly, two capped pairs $(\widehat \M(\tau,\sigma), \widehat\Sigma(\tau,\sigma)),\ (\widehat \M(\tau',\sigma'), \widehat\Sigma(\tau',\sigma'))$, are called \emph{capped-decorated-isomorphic} if there exists an orientation-preserving homeomorphism between the capped surfaces carrying the completed maps onto each other, preserving the decomposition of faces into $P$, $G$, $H$, $C$, and in particular preserving the set $C$ of capping faces. We denote respectively by $S$ and $\widehat S$ the sets of uncapped and capped decorated-isomorphism classes arising from pairs $(\tau,\sigma)\in \mathcal I\times\mathcal J$.

\begin{remark}\label{rmk:capping-bijection}
The capping construction induces a canonical bijection $\mathrm{cap}:S\longrightarrow \widehat S,$ because the capped-decorated-isomorphism class remembers the capping faces $C$, and removing their interiors recovers the uncapped decorated surface and map.
\end{remark}

For $(\tau,\sigma)\in \mathcal I\times\mathcal J$, define
\[
\Theta_N(\tau,\sigma):=N^{-\chi(\Sigma(\tau,\sigma))+h(\tau)}\mathcal K_N(\tau,\sigma),
\]
where $\mathcal K_N(\tau,\sigma)$ is the index-contraction factor of Theorem~\ref{prop:character_weingarten}. Equivalently, since
\begin{equation}\label{eq:chi_capped_uncapped}
\chi(\widehat\Sigma(\tau,\sigma))=\chi(\Sigma(\tau,\sigma)) + b(\tau,\sigma),
\end{equation}
we also have the tautological identity
\[
\mathcal K_N(\tau,\sigma)=\Theta_N(\tau,\sigma) N^{\chi(\widehat\Sigma(\tau,\sigma))-b(\tau,\sigma)-h(\tau)}.
\]
For any $\Xi=[(\M(\tau,\sigma), \Sigma(\tau,\sigma))]$, we define
\[
\chi(\Xi):=\chi(\Sigma(\tau,\sigma)),\qquad h(\Xi):=h(\tau),
\]
which are invariants of the decorated-isomorphism class, and the local coefficient
\[
\Omega_N(\Xi):=\sum_{(\tau,\sigma)}c_N(\tau) \left(\prod_{a=1}^m \mathrm{Wg}_N(\alpha_a^{-1}\beta_a)\right) \Theta_N(\tau,\sigma),
\]
where the sum is performed over all $(\tau,\sigma)\in \mathcal I\times\mathcal J$ such that $\Xi=[(\M(\tau,\sigma), \Sigma(\tau,\sigma))]$.

\begin{proposition}\label{prop:character_surface}
With the notation of Theorem~\ref{prop:character_weingarten},
\[
\int_{\U(N)^m}\prod_{i=1}^k \chi_{[\lambda_i^+,\lambda_i^-]_N}(w_i(U))dU=\sum_{\Xi\in S}\Omega_N(\Xi)N^{\chi(\Xi)-h(\Xi)}.
\]
\end{proposition}

\begin{proof}
By Theorem~\ref{prop:character_weingarten},
\[
\int_{\U(N)^m}\prod_{i=1}^k \chi_{[\lambda_i^+,\lambda_i^-]_N}(w_i(U))dU=\sum_{\tau\in\mathcal I} \sum_{\sigma\in \mathcal J} c_N(\tau) \left(\prod_{a=1}^m \mathrm{Wg}_N(\alpha_a^{-1}\beta_a)\right) \mathcal K_N(\tau,\sigma).
\]
Using the definition of $\Theta_N(\tau,\sigma)$, each summand can be rewritten as
\[
c_N(\tau)\left(\prod_{a=1}^m \mathrm{Wg}_N(\alpha_a^{-1}\beta_a)\right)\Theta_N(\tau,\sigma)N^{\chi(\Sigma(\tau,\sigma))-h(\tau)}.
\]
Grouping together all pairs $(\tau,\sigma)$ that produce the same uncapped decorated-isomorphism class $\Xi\in S$ yields the formula.
\end{proof}

\begin{remark}\label{rem:topological-grading-not-asymptotic}
The factor $N^{\chi(\Xi)-h(\Xi)}$ records a canonical topological grading of each decorated surface class. At this fully finite-$N$ level the coefficient $\Omega_N(\Xi)$ still contains the $N$-dependence of the projector and Weingarten factors. Thus Proposition~\ref{prop:character_surface} is an exact surface organization of the integral; by itself it does not assert that $\chi(\Xi)-h(\Xi)$ is the asymptotic order of the contribution as $N\to\infty$. In the trace-only Wilson specialization of Section~\ref{sec:specializations}, the standard renormalized Weingarten weights recover the genuine Euler-characteristic grading familiar from the usual topological $1/N$ expansion.
\end{remark}

At this stage, the construction seems to carry too much data to be seen as purely topological. It is nonetheless useful to get a better control in the large-$N$, as shown in \cite{Mag2,Dah26}. We can turn Proposition~\ref{prop:character_surface} into a coarser, more natural topological organization as follows. Let $\pi:S\longrightarrow S_{\mathrm{coarse}}$ be the forgetful map defined as follows. Starting from a refined pair $(\M(\tau,\sigma), \Sigma(\tau,\sigma)),$ forget:
\begin{itemize}
\item the decomposition into local Brauer gadgets $G_{i,r}(\tau_{i,r})$,
\item the positionwise labels $r$ on the Brauer layer,
\item the internal subdivision data that distinguish different local realizations of the same surface,
\end{itemize}
while retaining:
\begin{itemize}
\item the underlying oriented surface with boundary,
\item the associated topological map,
\item the decomposition of the faces into the images of $P$ and $H$,
\item the total defect
\[
h(\tau)=\sum_{i=1}^k\sum_{r=1}^{\ell_i} h(\tau_{i,r}).
\]
\end{itemize}
The image is denoted by $S_{\mathrm{coarse}}$.

\begin{corollary}\label{cor:coarse-capped-quotient}
With the notation above,
\begin{equation}
\int_{\U(N)^m}\prod_{i=1}^k \chi_{[\lambda_i^+,\lambda_i^-]_N}(w_i(U))dU=\sum_{\Sigma\in S_{\mathrm{coarse}}}\Omega_N(\Sigma)N^{\chi(\Sigma)-h(\Sigma)},
\end{equation}
where
\[
\Omega_N(\Sigma):=\sum_{\Xi\in \pi^{-1}(\Sigma)}\Omega_N(\Xi).
\]
\end{corollary}

\begin{proof}
Apply Proposition~\ref{prop:character_surface} and regroup the sum according to the forgetful map $\pi:S\to S_{\mathrm{coarse}}$. By construction, $\pi$ forgets only the local realization data and preserves the underlying surface, the topological map, and the total number of horizontal Brauer bands. Hence $\chi$, $b$, and $h$ are constant on the fibres of $\pi$, which yields the stated coarse formula.
\end{proof}

\subsection{Proof of the decorated-surface expansion}

We now apply the general surface expansion of Proposition~\ref{prop:character_surface} to the topological coefficients $\widehat{W}_{\Lambda,\Lc}(\alpha)$. The point is that the expansion applies to the family of plaquette words $(\partial p)_{p\in P(\Lambda)}$ together with the loop words $\ell_1,\dots,\ell_k$. For the plaquette factors, the highest weights are $\alpha_p=[\lambda_p^+,\lambda_p^-]_N,$ whereas for the loop factors we use the fundamental representation $\Tr(U_{\ell_i})=\chi_{[(1),\varnothing]_N}(U_{\ell_i}).$ Accordingly, each plaquette $p$ contributes $|\lambda_p^+|+|\lambda_p^-|$ slot-polygons in the slot-expanded construction of Section~\ref{sec:construction-spanning-surfaces}, while each loop $\ell_i$ contributes exactly one distinguished loop-polygon $P_{\ell_i}$.

\begin{theorem}\label{thm:fixed-alpha-punctured-sum}
Fix a plaquette-decoration $\alpha:P(\Lambda)\to \widehat{\U(N)}$, and note $\alpha_p=[\lambda_p^+,\lambda_p^-]_N$ for all $p\in P(\Lambda)$. There exists a finite set $S_{\Lambda,\Lc}(\alpha)$ of decorated oriented spanning surfaces with boundary, such that each $\Sigma\in S_{\Lambda,\Lc}(\alpha)$ is equipped with a coefficient $\Omega_N(\Sigma;\alpha)\in\mathbb C$ and an integer $h(\Sigma)\geq 0$, and such that
\begin{equation}
\widehat{W}_{\Lambda,\Lc}(\alpha)=\sum_{\Sigma\in S_{\Lambda,\Lc}(\alpha)}\Omega_N(\Sigma;\alpha) N^{\chi(\Sigma)-h(\Sigma)}.
\end{equation}
Moreover, every $\Sigma\in S_{\Lambda,\Lc}(\alpha)$ comes equipped with:
\begin{itemize}
\item a cellular map $\varphi_\Sigma:\Sigma\to \Lambda^{(2)}$;
\item $k$ ordered marked boundary components $C_1,\dots,C_k$;
\item for each $i$, the restriction $\varphi_\Sigma|_{C_i}$ has cyclic boundary word $\ell_i$;
\item a plaquette-sheet labeling compatible with $\alpha$;
\item a record of the local walled Brauer and Haar-pairing data.
\end{itemize}
\end{theorem}

\begin{proof}
Apply Proposition~\ref{prop:character_surface} to the family of words
\[
(\partial p)_{p\in P(\Lambda)}\cup (\ell_1,\dots,\ell_k),
\]
with highest weights $\alpha_p=[\lambda_p^+,\lambda_p^-]_N$ on the plaquette factors and
$((1),\varnothing)$ on the loop factors. Since $\Tr(U_{\ell_i})=\chi_{[(1),\varnothing]_N}(U_{\ell_i}),$ this yields a finite expansion indexed by pairs $(\tau,\sigma)$, where $\tau$ is the local walled Brauer datum $\tau=(\tau_{i,r})_{i,r},$ and $\sigma$ is the letterwise Haar-pairing datum of Theorem~\ref{prop:character_weingarten}. For each such pair, Section~\ref{sec:construction-spanning-surfaces} and Remark~\ref{rmk:capping-bijection} produce a capped decorated pair $(\M(\tau,\sigma), \widehat\Sigma(\tau,\sigma)),$ where $\Sigma(\tau,\sigma)$ is a compact oriented surface with boundary and $\M(\tau,\sigma)$ is a topological map embedded in $\Sigma(\tau,\sigma)$. Because the loop factors are fundamental, each loop $\ell_i$ contributes exactly one distinguished
loop-polygon $P_{\ell_i}$ among the faces of $\M(\tau,\sigma)$.

We now define the canonical cellular map
\[
\widehat\varphi_{\tau,\sigma}:\widehat\Sigma(\tau,\sigma)\longrightarrow \Lambda^{(2)}.
\]
It is specified on the faces of the completed map $\M(\tau,\sigma)$ as follows:
\begin{itemize}
\item each plaquette slot-polygon corresponding to a plaquette word $\partial p$ is sent homeomorphically onto the plaquette $p$, respecting the cyclic boundary word $\partial p$;
\item each distinguished loop-polygon $P_{\ell_i}$ is sent to the loop $\ell_i$ in the $1$-skeleton of $\Lambda$, with boundary word $\ell_i$;
\item each edge of the map labelled by an oriented lattice edge $e^{\pm1}$ is sent homeomorphically onto the corresponding oriented edge of $\Lambda$;
\item each Brauer face and each Haar face is collapsed onto the corresponding lower-dimensional cell determined by its edge labels.
\end{itemize}
This defines a cellular map from the completed surface to $\Lambda^{(2)}$.

We then define the punctured spanning surface associated with $(\tau,\sigma)$ by removing the interiors of the distinguished loop-faces:
\[
\Sigma^\partial(\tau,\sigma):=\widehat\Sigma(\tau,\sigma)\setminus \bigsqcup_{i=1}^k \operatorname{int}(P_{\ell_i}).
\]
Its boundary contains $k$ distinguished components
\[
C_i:=\partial P_{\ell_i},\qquad i=1,\dots,k.
\]
Restricting $\widehat\varphi_{\tau,\sigma}$ gives a cellular map
\[
\varphi_{\tau,\sigma}:\Sigma^\partial(\tau,\sigma)\to \Lambda^{(2)}.
\]
By construction, each plaquette slot-polygon associated with a plaquette word $\partial p$ is mapped homeomorphically onto the plaquette $p$, while the boundary word of the marked boundary component $C_i$ is exactly the loop $\ell_i$. The restriction $\varphi_{\tau,\sigma}|_{C_i}$ need not be injective, so self-intersections and mutual intersections of the loops are allowed.

The contribution of $(\tau,\sigma)$ in Proposition~\ref{prop:character_surface} is
\[
c_N(\tau)\Bigl(\prod_a \mathrm{Wg}_N(\alpha_a^{-1}\beta_a)\Bigr)\Theta_N(\tau,\sigma) N^{\chi(\Sigma(\tau,\sigma))-h(\tau)},
\]
where $\Sigma(\tau,\sigma)$ is the uncapped surface of Section~\ref{sec:construction-spanning-surfaces}. Recall that, by~\eqref{eq:chi_capped_uncapped},
\[
\chi(\widehat\Sigma(\tau,\sigma))=\chi(\Sigma(\tau,\sigma))+b(\tau,\sigma).
\]
Moreover, since $\Sigma^\partial(\tau,\sigma)$ is obtained from $\widehat\Sigma(\tau,\sigma)$ by removing exactly the $k$ distinguished loop-faces, we also have
\[
\chi(\Sigma^\partial(\tau,\sigma))= \chi(\widehat\Sigma(\tau,\sigma)) - k.
\]
Hence
\[
\chi(\Sigma(\tau,\sigma)) = \chi(\Sigma^\partial(\tau,\sigma)) - b(\tau,\sigma)+k.
\]
Therefore
\[
c_N(\tau)\Bigl(\prod_a \mathrm{Wg}_N(\alpha_a^{-1}\beta_a)\Bigr)\Theta_N(\tau,\sigma) N^{\chi(\Sigma(\tau,\sigma))-h(\tau)}= \widetilde\Omega_N(\tau,\sigma;\alpha) N^{\chi(\Sigma^\partial(\tau,\sigma))-h(\tau)},
\]
where
\[
\widetilde\Omega_N(\tau,\sigma;\alpha):= c_N(\tau)\Bigl(\prod_a \mathrm{Wg}_N(\alpha_a^{-1}\beta_a)\Bigr)\Theta_N(\tau,\sigma) N^{k-b(\tau,\sigma)}.
\]

We now define $S_{\Lambda,\Lc}(\alpha)$ to be the finite set of all decorated punctured spanning surfaces $\Sigma=\Sigma^\partial(\tau,\sigma)$ arising in this way, modulo orientation-preserving homeomorphisms preserving the map to $\Lambda^{(2)}$, the ordered marked boundary components, the plaquette-sheet labels, and the local Brauer/Haar decorations. For $\Sigma\in S_{\Lambda,\Lc}(\alpha)$, define
\[
\Omega_N(\Sigma;\alpha) := \sum_{(\tau,\sigma):\Sigma^\partial(\tau,\sigma)=\Sigma} \widetilde\Omega_N(\tau,\sigma;\alpha),
\]
and let $h(\Sigma):=h(\tau),$ which is well-defined because the decoration remembers the local walled Brauer data. Grouping the finite expansion of Proposition~\ref{prop:character_surface} according to the punctured spanning surface produced yields
\[
\widehat{W}_{\Lambda,\Lc}(\alpha) = \sum_{\Sigma\in S_{\Lambda,\Lc}(\alpha)} \Omega_N(\Sigma;\alpha) N^{\chi(\Sigma)-h(\Sigma)}.
\]
This proves the result.
\end{proof}

We can now prove the decorated-surface expansion of Wilson loop expectations.

\begin{proof}[Proof of Theorem~\ref{thm:surface-sum-Wilsonloops}]
We combine Theorem~\ref{thm:state_sum_Wilson} that gives an expression of the Wilson loop expectation as a sum over plaquette decorations, and Theorem~\ref{thm:fixed-alpha-punctured-sum} that gives a finite expansion of topological coefficients, which yields
\[
\mathbb E[W_{\Lambda,\Lc}(U)]=\frac1{Z_{\Lambda,Q}} \sum_\alpha \kappa_{\Lambda,Q}(\alpha) \sum_{\Sigma'\in S_{\Lambda,\Lc}(\alpha)} \Omega_N(\Sigma';\alpha) N^{\chi(\Sigma')-h(\Sigma')}.
\]
This is exactly~\eqref{eq:surface_sum} by the definition of $\kappa_{\Lambda,Q}$.
\end{proof}

To conclude this section, let us mention that Theorem~\ref{thm:surface-sum-Wilsonloops} should be understood as an exact decorated-surface expansion at the level of moments. As such, it naturally mixes the connected components carrying the marked loop boundaries with closed vacuum components. Indeed, for $\Sigma\in S_{\Lambda,\Lc}(\alpha)$, one may decompose
\[
\Sigma=\Sigma_{\mathrm{obs}}\sqcup \Sigma_{\mathrm{vac}},
\]
where $\Sigma_{\mathrm{obs}}$ is the union of the connected components meeting at least one of the marked boundary components $C_1,\dots,C_k$, and $\Sigma_{\mathrm{vac}}$ is the union of the remaining connected components. The latter are closed connected components, namely vacuum bubbles. Thus the exact expansion of Theorem~\ref{thm:surface-sum-Wilsonloops} naturally mixes two kinds of contributions: the components carrying the loop boundaries, and the closed vacuum components independent of the loops. This is why the coarse decorated-surface expansion is exact but not yet the optimal form for large-$N$ analysis. A connected version could be obtained by refining the decorated surface classes so as to remember the connected-component decomposition prior to quotienting, and then applying the standard exponential formula for vacuum components together with M\"obius inversion on set partitions for Wilson loop cumulants. We do not develop this refinement here, since for the applications of the present paper the more useful reorganization is the local defect-ratio formalism of the next section, which absorbs the vacuum bubbles into the background partition function.

\section{Local channel model and defect partition functions}\label{sec:local-channel}

Section~\ref{sec:surface-sum} produced two complementary geometric descriptions of the coefficients $\widehat{W}_{\Lambda,\Lc}(\alpha).$ However, they are not yet the most convenient form for the later developments of the paper. Even in the refined formulation of Theorem~\ref{thm:fixed-alpha-punctured-sum}, the local building blocks are assembled into a global spanning surface, so that formulas still depend on global connectedness, boundary components, capping operations, and possible vacuum bubbles. For the purposes of defect partition functions and of the later master-loop analysis, one would instead like a genuinely finite-range description in which the coefficients are written directly in terms of local data attached to plaquettes, non-tree edges, and their incidences.

The goal of the present section is to extract exactly this local content from the refined geometric picture of Section~\ref{sec:surface-sum}. After gauge fixing along a spanning tree, each plaquette insertion can be resolved into finitely many local mixed Schur--Weyl channels. The contractions created by the corresponding projector insertions are not, in general, products of independent edge factors; they may connect tensor slots belonging to several incidences of the same plaquette. We therefore keep these contractions as finite-dimensional operators and only afterwards expand them in tensor-product coordinates indexed by the incidences. Once this scalarization has been performed, the remaining Haar integrations are genuinely edgewise and are represented by orthogonal Haar projections. The result is a finite-range channel model on the dual incidence graph.

In contrast with many previous dual formulations of lattice gauge theory in terms of spin foams~\cite{OeckPfei01,Oec03,Con05,CCK07}, we do not dualize the full Yang--Mills measure at once. Instead, we work after the universal spectral/topological splitting of Theorem~\ref{thm:state_sum_Wilson}: the dependence on the plaquette action is entirely carried by the spectral weights $\kappa_{\Lambda,Q}(\alpha)$, whereas the present section reconstructs the topological coefficients $\widehat{W}_{\Lambda,\Lc}(\alpha)$ as a local channel partition function. The construction is completely finite-$N$ and action-independent. After summing against the spectral coefficients $\kappa_{\Lambda,Q}$, it yields the exact defect-ratio formula of Theorem~\ref{thm:spin-foam2}.

\subsection{The dual incidence graph}

Fix a finite connected lattice $\Lambda$, and let $T\subset E(\Lambda)$ be a spanning tree. We define the \emph{dual incidence graph} $D_T(\Lambda)$ to be the bipartite graph whose vertex set is
\[
V\left(D_T(\Lambda)\right)=P(\Lambda)\sqcup \bigl(E(\Lambda)\setminus T\bigr),
\]
and with one incidence edge $(p,e)$ between a plaquette vertex $p\in P(\Lambda)$ and an edge vertex $e\in E(\Lambda)\setminus T$ whenever the non-tree edge $e$ occurs in the gauge-fixed boundary word of $p$. Equivalently, before gauge fixing one may say that the underlying unoriented edge of $e$ lies in $\partial p$; repeated occurrences, when present, will be remembered in the local tensor data below. Hence the plaquette vertices of $D_T(\Lambda)$ record the face variables, whereas the edge vertices record the non-tree holonomies that remain after gauge fixing. An example in two dimensions is given in Figure~\ref{fig:dual-incidence}.

\begin{figure}[h!]
    \centering
    \includegraphics[width=0.3\linewidth]{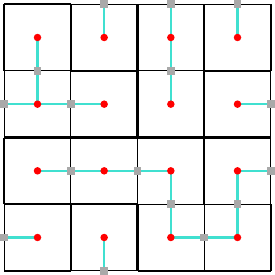}
    \caption{\small The dual incidence graph of a finite square lattice in $\Z^2$ for a given spanning tree (in bold). Plaquette vertices are red disks, edge vertices are grey squares, incidence edges are in blue.}
    \label{fig:dual-incidence}
\end{figure}

This graph is the natural local carrier for the coefficient $\widehat W_{\Lambda,\Lc}(\alpha)$. Indeed, after gauge fixing along $T$, every plaquette holonomy $U_{\partial p}$ is a word in the non-tree variables. The tensor associated with the plaquette label $\alpha_p$ therefore has external matrix-entry legs only at the incidences $(p,e)$ with $e\notin T$. Conversely, once the plaquette tensors have been expressed in incidence coordinates, the Haar variable $U_e$ occurs only in the factors attached to the edge vertex $e$, together with the Wilson strands traversing $e$. Thus:
\begin{itemize}
\item the \emph{plaquette side} of the model is localized at the plaquette vertices $p$ of $D_T(\Lambda)$;
\item the \emph{Haar constraints} are localized at the non-tree edge vertices $e$;
\item the \emph{channel variables} live on the incidence edges $(p,e)$;
\item the Wilson observable contributes an additional finite tensor-network defect along the edge vertices traversed by the loop family.
\end{itemize}

From this viewpoint, the global surface expansion of Section~\ref{sec:surface-sum} and the local model of the present section retain complementary information. The surface expansion remembers global topology---connected components, Euler characteristic and spanning surfaces---whereas the channel model forgets this global organization and keeps the local representation-theoretic compatibility data. The latter are the data relevant for locality questions and for the coefficientwise master-loop equation.

It is worth emphasizing that $D_T(\Lambda)$ does not arise ad hoc. Once a spanning tree is fixed, the independent group variables are exactly the non-tree edge holonomies, and a plaquette tensor can interact with another plaquette tensor only through a common non-tree variable. Thus the incidence relation
\[
\{\text{plaquettes}\}\longleftrightarrow\{\text{non-tree edges}\}
\]
is the minimal factor-graph structure describing the gauge-fixed integral. In particular, it is the natural higher-dimensional analogue of the local graphs appearing in tensor-network and spin-foam formulations.

A loop family $\Lc$ enters the gauge-fixed integral only through the non-tree edges that it traverses. We therefore define its \emph{defect support} by
\begin{equation}\label{eq:defect-support-sec3}
D(\Lc):=\{e\in E(\Lambda)\setminus T:\text{ some loop of }\Lc\text{ traverses }e\}.
\end{equation}
All plaquette-resolution data constructed below are independent of $\Lc$; only the Wilson scalarization and the resulting local edge factors are modified by the observable, and this modification is supported on $D(\Lc)$.

\subsection{Gauge-fixed local resolutions}\label{sec:local-resolution}

We first resolve a single plaquette character into incidence-local tensor coordinates. The argument starts from the mixed Schur--Weyl realization used in Section~\ref{sec:surface-sum}, but we keep the tensor contractions produced by the projectors as operators until the final tensor-product expansion. This point is essential: a contraction joining two different occurrences of a plaquette boundary can connect two different incidences, so it should not be split edgewise before scalarization.

We begin with the standard tree gauge fixing.

\begin{lemma}\label{lem:gauge-fixing-top-coeff}
Let $\Lambda$ be a finite connected lattice, let $T\subset E(\Lambda)$ be a spanning tree, and let $\alpha:P(\Lambda)\to\widehat{\U(N)}$ be a plaquette decoration. Then the coefficient
\[
\widehat W_{\Lambda,\Lc}(\alpha)
=\int_{\U(N)^{E(\Lambda)}}W_{\Lambda,\Lc}(U)
\prod_{p\in P(\Lambda)}\chi_{\alpha_p}(U_{\partial p})\,dU
\]
is unchanged by a rooted-tree gauge fixing sending every tree-edge holonomy to the identity. Consequently, after gauge fixing the only integration variables are $(U_e)_{e\in E(\Lambda)\setminus T}$.
\end{lemma}

\begin{proof}
Choose a root vertex $v_0$ of $T$. For every vertex $v$, let $\gamma_v$ be the unique oriented path in $T$ joining $v_0$ to $v$, and let $g_v$ be the ordered product of the tree-edge variables along $\gamma_v$, with inverse matrices when an edge is traversed against its chosen orientation. Apply the gauge transformation
\[
U_e\longmapsto g_{e^-}U_eg_{e^+}^{-1}
\]
to every edge variable. By construction, every transformed tree-edge variable is $I_N$. The product Haar measure is invariant under the successive left and right translations used in this change of variables. Each plaquette holonomy is conjugated, and each loop holonomy is conjugated at its base point; hence every character and every trace in the integrand is unchanged. Integrating out the now-trivial tree variables gives the announced gauge-fixed integral.
\end{proof}

Fix a plaquette $p\in P(\Lambda)$. In the rooted-tree gauge, write its boundary holonomy as the cyclic word
\begin{equation}\label{eq:gauge-fixed-plaquette-word}
\partial^Tp=e_{p,1}^{\varepsilon_{p,1}}\cdots e_{p,b_p}^{\varepsilon_{p,b_p}},
\qquad e_{p,r}\in E(\Lambda)\setminus T,
\quad \varepsilon_{p,r}\in\{+1,-1\}.
\end{equation}
The integer $b_p$ is the number of non-tree occurrences remaining in the gauge-fixed boundary word. We also write $e\in\partial^Tp$ when $e$ occurs at least once in~\eqref{eq:gauge-fixed-plaquette-word}; multiplicities are kept in the occurrence sets defined below.

Write
\[
\alpha_p=[\lambda_p^+,\lambda_p^-]_N,
\qquad n_p^\pm=|\lambda_p^\pm|,
\qquad r_p=n_p^++n_p^-,
\]
and
\[
T_{\alpha_p}=V^{\otimes n_p^+}\otimes(V^*)^{\otimes n_p^-},
\qquad P_{\alpha_p}=P_N^{[\lambda_p^+,\lambda_p^-]}.
\]
By~\eqref{eq:Koike_character},
\begin{equation}\label{eq:plaquette-character-projector}
\chi_{\alpha_p}(U_{\partial p})
=\Tr_{T_{\alpha_p}}\!\left(
P_{\alpha_p}\rho_{n_p^+,n_p^-}(U_{e_{p,1}}^{\varepsilon_{p,1}})\cdots
\rho_{n_p^+,n_p^-}(U_{e_{p,b_p}}^{\varepsilon_{p,b_p}})
\right).
\end{equation}
Since $P_{\alpha_p}$ is a $\U(N)$-equivariant projection, it commutes with every factor in the group representation and satisfies $P_{\alpha_p}^2=P_{\alpha_p}$. We may therefore insert a copy of $P_{\alpha_p}$ between every pair of consecutive letters in~\eqref{eq:plaquette-character-projector}. Expanding each copy with Corollary~\ref{cor:PN-nm}, or equivalently~\eqref{eq:PN-nm}, gives
\begin{equation}\label{eq:local-brauer-expansion-sec3}
\chi_{\alpha_p}(U_{\partial p})
=\sum_{\tau_p\in\mathcal B_{n_p^+,n_p^-}^{b_p}}
 c_N^{\alpha_p}(\tau_p)
\Tr_{T_{\alpha_p}}\!\left(
\prod_{r=1}^{b_p}
\rho_N(\tau_{p,r})\rho_{n_p^+,n_p^-}(U_{e_{p,r}}^{\varepsilon_{p,r}})
\right),
\end{equation}
where $\tau_p=(\tau_{p,1},\ldots,\tau_{p,b_p})$ and
\[
c_N^{\alpha_p}(\tau_p)=\prod_{r=1}^{b_p}c_N^{\alpha_p}(\tau_{p,r}).
\]

We next reorganize the tensor factors by incidences. Put
\[
\kappa_p(u)=
\begin{cases}
+,&1\le u\le n_p^+,\\
-,&n_p^+<u\le r_p,
\end{cases}
\qquad V^+=V,\quad V^-=V^*.
\]
For a boundary occurrence $(r,u)$, define its effective matrix-entry type by
\[
\eta_{p,r,u}=
\begin{cases}
\kappa_p(u),&\varepsilon_{p,r}=+1,\\
-\kappa_p(u),&\varepsilon_{p,r}=-1,
\end{cases}
\qquad -(+)=-,\quad -(-)=+.
\]
This convention simply records whether the corresponding matrix coefficient is a coefficient of $U_{e_{p,r}}$ or of $\overline{U_{e_{p,r}}}$; any static transposition of the input/output indices caused by an inverse occurrence is absorbed into the deterministic contraction operator below.

For every incidence $e\in\partial^Tp$, define the ordered occurrence set
\[
S_{p,e}:=\{(r,u):1\le r\le b_p,\ 1\le u\le r_p,\ e_{p,r}=e\}
\]
and the incidence space
\begin{equation}\label{eq:incidence-space-sec3}
\mathcal H_{p,e}(\tau_p;\alpha_p):=
\bigotimes_{(r,u)\in S_{p,e}}V^{\eta_{p,r,u}}.
\end{equation}
The notation keeps $\tau_p$ because the coordinate operator placed on this space will depend on the Brauer tuple, although the underlying list of $V/V^*$ factors is determined by the boundary word and $\alpha_p$. Set
\begin{equation}\label{eq:plaquette-incidence-product}
\mathcal K_p(\tau_p;\alpha_p)
:=\bigotimes_{e\in\partial^Tp}\mathcal H_{p,e}(\tau_p;\alpha_p).
\end{equation}

Let us make explicit why the expression~\eqref{eq:local-brauer-expansion-sec3} has the required form on~\eqref{eq:plaquette-incidence-product}. If $H$ is a finite-dimensional vector space and $m\ge1$, let $\operatorname{Cyc}^{(m)}_H$ be the cyclic permutation on $H^{\otimes m}$ characterized by
\[
\Tr_{H^{\otimes m}}\!\left(\operatorname{Cyc}^{(m)}_H(A_1\otimes\cdots\otimes A_m)\right)
=\Tr_H(A_1\cdots A_m).
\]
Applying this identity with $H=T_{\alpha_p}$ turns the trace in~\eqref{eq:local-brauer-expansion-sec3} into a trace on $T_{\alpha_p}^{\otimes b_p}$. Reorder the $r_pb_p$ elementary tensor slots so that all slots carrying the same edge label are adjacent. The tensor product of the group-dependent factors then becomes
\[
\bigotimes_{e\in\partial^Tp}\pi^N_{p,e,\tau_p}(U_e),
\]
where $\pi^N_{p,e,\tau_p}$ is the tensor product of the fundamental and dual actions on the factors of~\eqref{eq:incidence-space-sec3}. The cyclic permutation, the Brauer diagrams and the static index transpositions are independent of the gauge variables; after the same reordering, they form a single operator
\[
\Phi^N_{p,\tau_p}\in\End\bigl(\mathcal K_p(\tau_p;\alpha_p)\bigr).
\]
Consequently,
\begin{equation}\label{eq:plaquette-operator-form}
\chi_{\alpha_p}(U_{\partial p})
=\sum_{\tau_p}c_N^{\alpha_p}(\tau_p)
\Tr_{\mathcal K_p(\tau_p;\alpha_p)}\!\left[
\Phi^N_{p,\tau_p}
\left(\bigotimes_{e\in\partial^Tp}\pi^N_{p,e,\tau_p}(U_e)\right)
\right].
\end{equation}

Formula~\eqref{eq:plaquette-operator-form} is the point at which we differ from a naive edgewise factorization. The operator $\Phi^N_{p,\tau_p}$ may couple several incidence spaces, and in general there is no reason for it to be a single tensor product over $e$. It is, however, an element of the finite-dimensional tensor product
\[
\End\bigl(\mathcal K_p(\tau_p;\alpha_p)\bigr)
\simeq
\bigotimes_{e\in\partial^Tp}\End\bigl(\mathcal H_{p,e}(\tau_p;\alpha_p)\bigr).
\]
Choose once and for all, for every incidence space appearing here, a finite basis of its endomorphism algebra. Expanding $c_N^{\alpha_p}(\tau_p)\Phi^N_{p,\tau_p}$ in the induced tensor-product basis gives the finite identity
\begin{equation}\label{eq:finiteN-decomp-cphi}
c_N^{\alpha_p}(\tau_p)\Phi^N_{p,\tau_p}
=\sum_{\iota_{\partial p}}
 a_p^N(\tau_p,\iota_{\partial p};\alpha_p)
\bigotimes_{e\in\partial^Tp}
C^N_{p,e}(\tau_p,\iota_{p,e};\alpha_p).
\end{equation}
Here $\iota_{\partial p}=(\iota_{p,e})_{e\in\partial^Tp}$ ranges over a finite set of coordinate indices. This finite tensor-product expansion, rather than any factorization of $\Phi^N_{p,\tau_p}$ itself, is what produces the local channels.

\begin{definition}\label{def:local-resolution}
A \emph{local plaquette resolution} of the label $\alpha_p$ is a pair
\[
r_p=(\tau_p,\iota_{\partial p})
\]
for which the coefficient $a_p^N(\tau_p,\iota_{\partial p};\alpha_p)$ in~\eqref{eq:finiteN-decomp-cphi} is non-zero. We denote the finite set of such resolutions by $\mathcal R_p(\alpha_p)$ and write
\[
a_p^N(r_p;\alpha_p):=a_p^N(\tau_p,\iota_{\partial p};\alpha_p),
\qquad
C^N_{p,e}(r_p;\alpha_p):=C^N_{p,e}(\tau_p,\iota_{p,e};\alpha_p),
\]
\[
\mathcal H_{p,e}(r_p;\alpha_p):=\mathcal H_{p,e}(\tau_p;\alpha_p),
\qquad
\pi^N_{p,e,r_p}:=\pi^N_{p,e,\tau_p}.
\]
The \emph{incidence channel} associated with $r_p$ at $(p,e)$ is the finite record
\[
\Gamma_{p,e}(r_p;\alpha_p)
:=\bigl(\mathcal H_{p,e}(r_p;\alpha_p),C^N_{p,e}(r_p;\alpha_p)\bigr),
\]
together with the ordered list of $V/V^*$ types in the incidence space. The \emph{boundary-channel tuple} is
\[
\Gamma_{\partial p}(r_p;\alpha_p)
:=\bigl(\Gamma_{p,e}(r_p;\alpha_p)\bigr)_{e\in\partial^Tp}.
\]
\end{definition}

\begin{figure}[h!]
    \centering
    \includegraphics[width=0.7\linewidth]{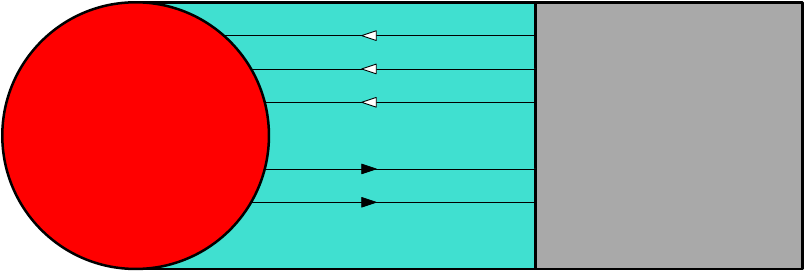}
    \caption{\small A boundary channel (in blue) between a non-tree edge (in grey) and an adjacent plaquette (in red). The oriented lines represent the incidence tensor legs: $V$ for black arrows and $V^*$ for white arrows. The deterministic operator $C^N_{p,e}$ records the corresponding incidence coordinate of the plaquette contraction.}
    \label{fig:local-channel}
\end{figure}

The preceding construction can now be packaged as a local resolution statement.

\begin{proposition}\label{prop:local-gauge-fixed-resolutions}
For every plaquette $p$ and every label $\alpha_p\in\widehat{\U(N)}$, there exists a finite set $\mathcal R_p(\alpha_p)$ of local plaquette resolutions such that
\begin{align}\label{eq:local-resolution-character}
\chi_{\alpha_p}(U_{\partial p})
=\sum_{r_p\in\mathcal R_p(\alpha_p)}a_p^N(r_p;\alpha_p)
\Tr_{\mathcal K_p(r_p;\alpha_p)}\!\left[
\left(\bigotimes_{e\in\partial^Tp}C^N_{p,e}(r_p;\alpha_p)\right)
\left(\bigotimes_{e\in\partial^Tp}\pi^N_{p,e,r_p}(U_e)\right)
\right],
\end{align}
where $\mathcal K_p(r_p;\alpha_p)=\bigotimes_{e\in\partial^Tp}\mathcal H_{p,e}(r_p;\alpha_p)$. In particular, all dependence on a non-tree variable $U_e$ is carried by the representation on the corresponding incidence space, while all contractions linking different incidences have been transferred to the finite sum over tensor-product coordinate operators.
\end{proposition}

\begin{proof}
Starting from~\eqref{eq:plaquette-character-projector}, insert the mixed Schur--Weyl projector between consecutive boundary letters and expand it in walled Brauer diagrams. This gives the finite sum~\eqref{eq:local-brauer-expansion-sec3}. The cyclic-trace identity and the canonical reordering by edge labels yield~\eqref{eq:plaquette-operator-form}. Finally expand the deterministic operator $c_N^{\alpha_p}(\tau_p)\Phi^N_{p,\tau_p}$ in the tensor-product basis~\eqref{eq:finiteN-decomp-cphi}. Substituting this expansion into~\eqref{eq:plaquette-operator-form} and grouping the pair $(\tau_p,\iota_{\partial p})$ into the single resolution variable $r_p$ gives~\eqref{eq:local-resolution-character}.
\end{proof}

The resolution is auxiliary rather than canonical term by term. Outside the stable range, the coefficients in the Brauer expansion~\eqref{eq:PN-nm} are not uniquely determined by the spanning family of diagrams, and the tensor-product bases used in~\eqref{eq:finiteN-decomp-cphi} are also choices. Different choices therefore produce different local channel sets and individual amplitudes. What is canonical is their resummation: Proposition~\ref{prop:local-gauge-fixed-resolutions} always reconstructs the same character, hence the subsequent full channel sum always reconstructs the same topological coefficient.

For later regrouping, define the finite set of admissible boundary-channel tuples
\[
\mathcal C_p(\alpha_p):=
\{\Gamma_{\partial p}(r_p;\alpha_p):r_p\in\mathcal R_p(\alpha_p)\}.
\]
For $\eta_{\partial p}\in\mathcal C_p(\alpha_p)$ set
\begin{equation}\label{eq:compressed-plaquette-amplitude}
A_p^N(\eta_{\partial p};\alpha_p)
:=\sum_{\substack{r_p\in\mathcal R_p(\alpha_p)\\
\Gamma_{\partial p}(r_p;\alpha_p)=\eta_{\partial p}}}
a_p^N(r_p;\alpha_p).
\end{equation}
Equality of boundary-channel tuples means equality of the incidence spaces, their ordered $V/V^*$ types and the coordinate operators $C^N_{p,e}$; hence all edge data appearing below are unchanged by this regrouping.

\subsection{Edgewise Haar projections and compression}\label{subsec:edgewise-haar}

We now add the Wilson insertion and perform the edge integrations. The role of this subsection is to justify carefully the locality statement in Theorem~\ref{thm:spin-foam}. Before scalarization, neither a plaquette trace nor a Wilson trace need factor as a product of functions of individual edge variables. After the finite tensor-product expansions of the preceding subsection and the analogous expansion of the cyclic Wilson contraction, however, the full integrand becomes a tensor product over non-tree edges. At that stage Fubini's theorem applies edge by edge, and each Haar integral is simply the orthogonal projection onto an invariant subspace.

Let the gauge-fixed word of the $i$-th loop be
\[
\ell_i=f_{i,1}^{\delta_{i,1}}\cdots f_{i,n_i}^{\delta_{i,n_i}},
\qquad f_{i,r}\in E(\Lambda)\setminus T,
\quad \delta_{i,r}\in\{+1,-1\}.
\]
As in the plaquette construction, an occurrence with exponent $+1$ contributes a fundamental matrix-entry leg and an occurrence with exponent $-1$ contributes a dual leg. For every non-tree edge $e$ define the Wilson incidence space
\begin{equation}\label{eq:wilson-edge-space}
\mathcal W_e(\Lc):=
\bigotimes_{(i,r):f_{i,r}=e}V^{\delta_{i,r}},
\end{equation}
with the convention $\mathcal W_e(\Lc)=\C$ if $e$ is not traversed. The corresponding representation is
\begin{equation}\label{eq:wilson-edge-rep}
\pi^W_{e,N}(g):=
\bigotimes_{(i,r):f_{i,r}=e}g^{\delta_{i,r}},
\end{equation}
where the notation means the fundamental action on $V$ and the dual action on $V^*$; the static reversal of matrix indices at inverse occurrences is again included in the contraction tensor.

The product of loop traces is obtained from these matrix-entry legs by a fixed cyclic contraction. More precisely, the cyclic identity used above for the plaquette traces, applied separately to the loops and followed by the reordering of all loop occurrences according to their non-tree edge label, gives a deterministic operator
\[
C^W_{\Lambda,\Lc,N}
\in
\End\!\left(\bigotimes_{e\in D(\Lc)}\mathcal W_e(\Lc)\right)
\]
such that the gauge-fixed Wilson functional is the trace of $C^W_{\Lambda,\Lc,N}$ against the tensor product of the representations~\eqref{eq:wilson-edge-rep}. Loop components which reduce to the trivial word after tree gauge fixing contribute the scalar $N$ and may be removed at this stage; we keep these explicit scalar factors as part of the Wilson contraction coefficient.

Since the defect tensor product is finite-dimensional, choose a decomposable expansion
\begin{equation}\label{eq:wilson-decomposable-expansion}
C^W_{\Lambda,\Lc,N}
=\sum_{\omega\in\Omega^W_{\Lambda,\Lc}}
\bigotimes_{e\in E(\Lambda)\setminus T}C^W_{e,\omega,N},
\end{equation}
where $C^W_{e,\omega,N}=1$ on the one-dimensional space $\mathcal W_e(\Lc)=\C$ whenever $e\notin D(\Lc)$. Scalar coefficients of a decomposable term are absorbed into one of its non-trivial defect factors. If $D(\Lc)=\varnothing$, the Wilson functional is a constant and the conclusions below are immediate after multiplying the vacuum coefficient by this constant. For $\Lc=\varnothing$,~\eqref{eq:wilson-decomposable-expansion} consists of the single scalar term $1$.

\begin{lemma}\label{lem:wilson-scalarization}
After tree gauge fixing, the Wilson functional admits a finite expansion of the form
\begin{equation}\label{eq:wilson-scalarized-trace}
W^T_{\Lambda,\Lc}(U)
=\sum_{\omega\in\Omega^W_{\Lambda,\Lc}}
\prod_{e\in E(\Lambda)\setminus T}
\Tr_{\mathcal W_e(\Lc)}\!\left(
C^W_{e,\omega,N}\pi^W_{e,N}(U_e)
\right),
\end{equation}
where the scalar factors $N$ contributed by loop components which become trivial in the tree gauge are included in the coefficients of the decomposable expansion~\eqref{eq:wilson-decomposable-expansion}. Moreover, the factors in~\eqref{eq:wilson-scalarized-trace} are equal to $1$ away from $D(\Lc)$.
\end{lemma}

\begin{proof}
For a non-trivial gauge-fixed loop $\ell_i=f_{i,1}^{\delta_{i,1}}\cdots f_{i,n_i}^{\delta_{i,n_i}}$, use the cyclic-trace identity to write $\Tr(U_{\ell_i})$ as a trace on the tensor product of the $n_i$ elementary $V/V^*$ occurrence spaces. Taking the tensor product over all loop components turns the product of traces into one trace on the tensor product over all loop occurrences. Reorder these factors according to their edge labels. The group-dependent part is then exactly $\bigotimes_e\pi^W_{e,N}(U_e)$, while all cyclic permutations and static index reversals form the deterministic operator $C^W_{\Lambda,\Lc,N}$. The scalar factors $N$ coming from gauge-fixed trivial loop components have already been incorporated into $C^W_{\Lambda,\Lc,N}$ and hence into the coefficients of~\eqref{eq:wilson-decomposable-expansion}. Expanding this operator and using multiplicativity of the trace on tensor products gives the exact identity~\eqref{eq:wilson-scalarized-trace}. If an edge is not traversed, its Wilson space is $\C$, its representation and coordinate operator are both the identity, and the corresponding trace equals $1$.
\end{proof}

We now combine one plaquette resolution at every plaquette with one Wilson scalarization term. Let
\[
r=(r_p)_{p\in P(\Lambda)}\in
\mathcal R(\alpha):=\prod_{p\in P(\Lambda)}\mathcal R_p(\alpha_p),
\qquad \omega\in\Omega^W_{\Lambda,\Lc}.
\]
For a non-tree edge $e$, collect all incidence spaces adjacent to $e$ and the Wilson space at $e$:
\begin{equation}\label{eq:edge-hilbert-space}
\mathcal H_e(r,\omega;\alpha,\Lc)
:=\left(\bigotimes_{p:e\in\partial^Tp}\mathcal H_{p,e}(r_p;\alpha_p)\right)
\otimes\mathcal W_e(\Lc).
\end{equation}
The group acts through
\begin{equation}\label{eq:edge-representation}
\pi^{r,\omega}_{e,N}(g)
:=\left(\bigotimes_{p:e\in\partial^Tp}\pi^N_{p,e,r_p}(g)\right)
\otimes\pi^W_{e,N}(g),
\end{equation}
and the deterministic contraction operator is
\begin{equation}\label{eq:edge-contraction-operator}
C_{e,N}(r,\omega;\alpha,\Lc)
:=\left(\bigotimes_{p:e\in\partial^Tp}C^N_{p,e}(r_p;\alpha_p)\right)
\otimes C^W_{e,\omega,N}.
\end{equation}
Finally define the edge Haar projection
\begin{equation}\label{eq:edge-haar-projection}
\Pi_{e,N}(r,\omega;\alpha,\Lc)
:=\int_{\U(N)}\pi^{r,\omega}_{e,N}(g)\,dg
=\operatorname{Proj}_{\mathcal H_e(r,\omega;\alpha,\Lc)^{\U(N)}}.
\end{equation}
The second equality follows from the standard averaging argument for unitary representations of compact groups.

The diagrams in Figures~\ref{fig:local-channel-2} and~\ref{fig:local-channel-3} give a strand picture of the two situations. The individual pairings drawn there may be obtained by expanding the Haar projection with the ordinary Weingarten formula; for the present local model we keep the entire Haar average as the single operator~\eqref{eq:edge-haar-projection}.

\begin{figure}[h!]
    \centering
    \includegraphics[width=0.9\linewidth]{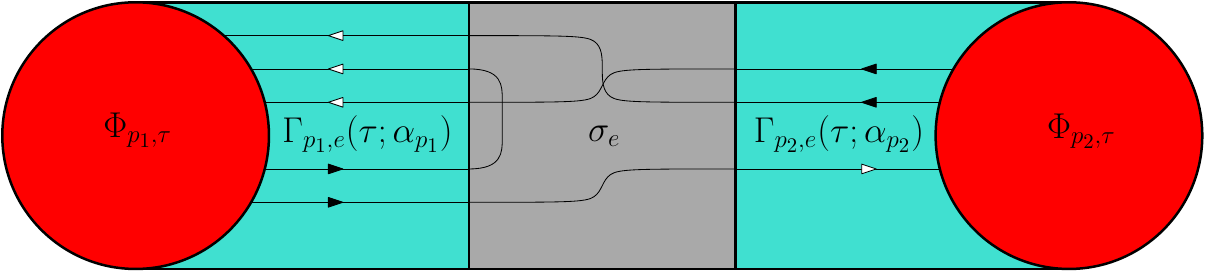}
    \caption{\small Schematic local channel at a non-tree edge $e$ without Wilson insertion. The incidence spaces and deterministic contraction coordinates are attached to the adjacent plaquettes. Haar integration couples all legs adjacent to $e$ through the projection onto the invariant subspace; expanding this projection in matrix entries recovers the usual finite sum over Weingarten pairings represented by the strands.}
    \label{fig:local-channel-2}
\end{figure}

\begin{figure}[h!]
    \centering
    \includegraphics[width=0.9\linewidth]{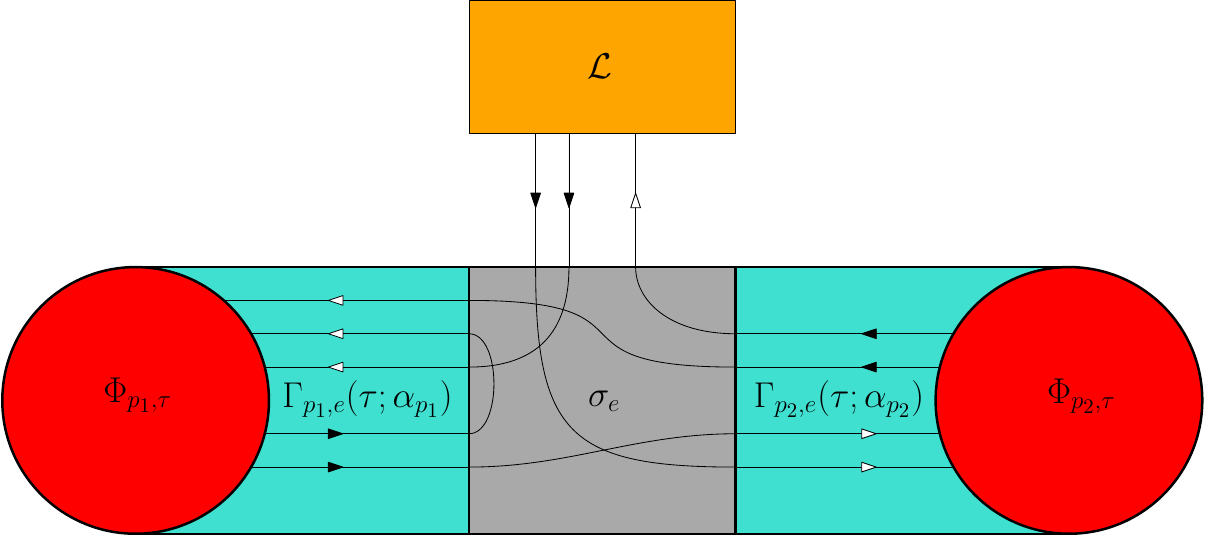}
    \caption{\small The same local picture with a Wilson insertion. Each traversal of $e$ or $e^{-1}$ contributes an additional $V$- or $V^*$-leg to the edge space. The Haar projection is therefore modified only at edge vertices belonging to the defect support $D(\Lc)$.}
    \label{fig:local-channel-3}
\end{figure}

The crucial factorization is now valid at the operator level.

\begin{lemma}\label{lem:edgewise-projection-factorization}
Fix $\alpha$, a global plaquette-resolution field $r\in\mathcal R(\alpha)$ and a Wilson scalarization index $\omega\in\Omega^W_{\Lambda,\Lc}$. The corresponding summand of the gauge-fixed integrand is
\begin{equation}\label{eq:correct-edgewise-factorization}
\prod_{e\in E(\Lambda)\setminus T}
\Tr_{\mathcal H_e(r,\omega;\alpha,\Lc)}\!\left(
C_{e,N}(r,\omega;\alpha,\Lc)\,
\pi^{r,\omega}_{e,N}(U_e)
\right).
\end{equation}
Consequently, its integral over the non-tree variables is
\begin{equation}\label{eq:correct-edgewise-integral}
\prod_{e\in E(\Lambda)\setminus T}
\Tr_{\mathcal H_e(r,\omega;\alpha,\Lc)}\!\left(
C_{e,N}(r,\omega;\alpha,\Lc)\,
\Pi_{e,N}(r,\omega;\alpha,\Lc)
\right).
\end{equation}
In particular, after the finite plaquette and Wilson scalarizations, each remaining Haar integration depends only on the data adjacent to the single edge vertex $e$.
\end{lemma}

\begin{proof}
Substitute the plaquette expansion~\eqref{eq:local-resolution-character} at every plaquette and the Wilson expansion~\eqref{eq:wilson-scalarized-trace}. Fix one choice of all resolution indices and of $\omega$. Before reordering by edges, the corresponding term is the trace of the tensor product of all deterministic plaquette-coordinate operators and Wilson-coordinate operators against the tensor product of all group representations. Reorder the tensor factors so that all incidence and Wilson spaces carrying the same non-tree variable $U_e$ are adjacent. By~\eqref{eq:edge-hilbert-space}--\eqref{eq:edge-contraction-operator}, the term becomes the trace on $\bigotimes_e\mathcal H_e$ of
\[
\left(\bigotimes_e C_{e,N}\right)
\left(\bigotimes_e\pi^{r,\omega}_{e,N}(U_e)\right).
\]
The trace of a tensor product is the product of the traces, which gives~\eqref{eq:correct-edgewise-factorization}. This is precisely the step that would fail before the tensor-product scalarizations: the unsplit plaquette and Wilson contraction operators need not themselves be decomposable over the edges.

The Haar measure on the remaining variables is the product $\prod_{e\notin T}dU_e$. Applying Fubini to~\eqref{eq:correct-edgewise-factorization} gives a product of one-edge integrals. Linearity of the trace and~\eqref{eq:edge-haar-projection} yield
\[
\int_{\U(N)}
\Tr_{\mathcal H_e}\bigl(C_{e,N}\pi^{r,\omega}_{e,N}(U_e)\bigr)\,dU_e
=
\Tr_{\mathcal H_e}\bigl(C_{e,N}\Pi_{e,N}\bigr),
\]
which proves~\eqref{eq:correct-edgewise-integral}.
\end{proof}

We can now compress the auxiliary resolution labels. A \emph{scalar local-channel field} consists of a choice of a boundary-channel tuple $\Gamma_{\partial p}\in\mathcal C_p(\alpha_p)$ for every plaquette $p$, together with a Wilson scalarization index $\omega\in\Omega^W_{\Lambda,\Lc}$. We denote the finite set of all such compatible fields by $\mathcal C^\Lc_\Lambda(\alpha)$. At a non-tree edge $e$, the local datum $\Gamma_e$ collects the incidence channels $\Gamma_{p,e}$ of the adjacent plaquettes and the local Wilson factor $C^W_{e,\omega,N}$. Since equality of boundary-channel tuples includes the incidence spaces, their ordered types and their coordinate operators, the following definitions do not depend on the choice of representatives $r_p$ realizing the prescribed boundary tuples:
\[
\mathcal H_e(\Gamma_e;\alpha,\Lc)
:=\left(\bigotimes_{p:e\in\partial^Tp}\mathcal H_{p,e}(r_p;\alpha_p)\right)
\otimes\mathcal W_e(\Lc),
\]
\[
\Pi^\Gamma_{e,N}(\alpha,\Lc)
:=\int_{\U(N)}\pi^\Gamma_{e,N}(g)\,dg,
\qquad
C_{e,N}(\Gamma_e;\alpha,\Lc)
:=\left(\bigotimes_{p:e\in\partial^Tp}C^N_{p,e}(r_p;\alpha_p)\right)
\otimes C^W_{e,\omega,N}.
\]
We define the \emph{scalar edge kernel} by
\begin{equation}\label{eq:projection-edge-kernel}
K_e^N(\Gamma_e;\alpha,\Lc)
:=\Tr_{\mathcal H_e(\Gamma_e;\alpha,\Lc)}\!\left(
C_{e,N}(\Gamma_e;\alpha,\Lc)\Pi^\Gamma_{e,N}(\alpha,\Lc)
\right).
\end{equation}

We are now ready to prove the local-channel theorem in full detail.

\begin{proof}[Proof of Theorem~\ref{thm:spin-foam}]
Fix a plaquette decoration $\alpha:P(\Lambda)\to\widehat{\U(N)}$. By Lemma~\ref{lem:gauge-fixing-top-coeff},
\begin{equation}\label{eq:gauge-fixed-topological-coefficient-sec3}
\widehat W_{\Lambda,\Lc}(\alpha)
=\int_{\U(N)^{E(\Lambda)\setminus T}}
W^T_{\Lambda,\Lc}(U)
\prod_{p\in P(\Lambda)}\chi_{\alpha_p}(U_{\partial p})
\prod_{e\in E(\Lambda)\setminus T}dU_e.
\end{equation}
Insert Proposition~\ref{prop:local-gauge-fixed-resolutions} for each plaquette and Lemma~\ref{lem:wilson-scalarization} for the Wilson functional. Since every indexing set is finite, the sums may be moved outside the integral. We obtain a finite sum over
\[
r=(r_p)_{p\in P(\Lambda)}\in\mathcal R(\alpha)
\qquad\text{and}\qquad
\omega\in\Omega^W_{\Lambda,\Lc}
\]
whose coefficient contains the plaquette factor $\prod_{p\in P(\Lambda)}a_p^N(r_p;\alpha_p).$ For fixed $(r,\omega)$, Lemma~\ref{lem:edgewise-projection-factorization} evaluates the remaining integral as
\[
\prod_{e\in E(\Lambda)\setminus T}
\Tr_{\mathcal H_e(r,\omega;\alpha,\Lc)}\!\left(
C_{e,N}(r,\omega;\alpha,\Lc)
\Pi_{e,N}(r,\omega;\alpha,\Lc)
\right).
\]
Thus
\begin{align}\label{eq:precompressed-channel-sum}
\widehat W_{\Lambda,\Lc}(\alpha)
=\sum_{r\in\mathcal R(\alpha)}\sum_{\omega\in\Omega^W_{\Lambda,\Lc}}
\left(\prod_{p\in P(\Lambda)}a_p^N(r_p;\alpha_p)\right)
\left(\prod_{e\in E(\Lambda)\setminus T}
K_e^N(\Gamma_e(r,\omega);\alpha,\Lc)\right).
\end{align}
Here $\Gamma_e(r,\omega)$ denotes the edge datum determined by the incidence channels of the local resolutions adjacent to $e$ and by the Wilson index $\omega$.

It remains only to forget the redundant plaquette-resolution labels. Fix a scalar local-channel field $\Gamma\in\mathcal C^\Lc_\Lambda(\alpha)$. The condition that a resolution field $r$ induce the prescribed boundary-channel tuples $\Gamma_{\partial p}$ is plaquettewise. Moreover, by construction of the tuples, the edge kernels in~\eqref{eq:precompressed-channel-sum} are constant when $r_p$ is replaced by another local resolution with the same $\Gamma_{\partial p}$. Hence the sum over resolutions inducing $\Gamma$ factorizes:
\[
\sum_{\substack{r\in\mathcal R(\alpha)\\
\Gamma_{\partial p}(r_p;\alpha_p)=\Gamma_{\partial p}\ \forall p}}
\prod_{p\in P(\Lambda)}a_p^N(r_p;\alpha_p)
=
\prod_{p\in P(\Lambda)}
\sum_{\substack{r_p\in\mathcal R_p(\alpha_p)\\
\Gamma_{\partial p}(r_p;\alpha_p)=\Gamma_{\partial p}}}
a_p^N(r_p;\alpha_p).
\]
By~\eqref{eq:compressed-plaquette-amplitude}, the last expression is
\[
\prod_{p\in P(\Lambda)}A_p^N(\Gamma_{\partial p};\alpha_p).
\]
Regrouping~\eqref{eq:precompressed-channel-sum} according to the scalar channel field therefore gives
\[
\widehat W_{\Lambda,\Lc}(\alpha)
=\sum_{\Gamma\in\mathcal C^\Lc_\Lambda(\alpha)}
\left(\prod_{p\in P(\Lambda)}A_p^N(\Gamma_{\partial p};\alpha_p)\right)
\left(\prod_{e\in E(\Lambda)\setminus T}K_e^N(\Gamma_e;\alpha,\Lc)\right),
\]
which is the required formula.

The three locality assertions follow directly from the construction. The coefficient $A_p^N$ is defined entirely from the single plaquette character $\chi_{\alpha_p}$ and its incidence coordinates. At a fixed edge $e$, the kernel~\eqref{eq:projection-edge-kernel} sees only the incidence channels of plaquettes adjacent to $e$ and the Wilson factor at $e$. Finally, the entire plaquette resolution is independent of $\Lc$; the loop family appears only through the finite scalarization~\eqref{eq:wilson-decomposable-expansion}, whose non-trivial factors occur on $D(\Lc)$. This proves Theorem~\ref{thm:spin-foam}.
\end{proof}

\subsection{Defect partition functions}

We now sum the local coefficient formula against the action-dependent spectral weights. Recall the defect partition function from~\eqref{eq:defect-pf}:
\[
\mathcal Z_{\Lambda,Q}^{(\Lc)}
=\sum_{\alpha:P(\Lambda)\to\widehat{\U(N)}}\kappa_{\Lambda,Q}(\alpha)
\sum_{\Gamma\in\mathcal C^\Lc_\Lambda(\alpha)}
\left(\prod_{p\in P(\Lambda)}A_p^N(\Gamma_{\partial p};\alpha_p)\right)
\left(\prod_{e\in E(\Lambda)\setminus T}K_e^N(\Gamma_e;\alpha,\Lc)\right).
\]
The point of this notation is that the numerator and denominator of a Wilson expectation are built from the same bulk plaquette-channel system. The observable adds a finite defect tensor, not a new bulk model.

\begin{proposition}\label{prop:finite-range-defect-property}
Fix a plaquette decoration $\alpha$. The plaquette-resolution sets $\mathcal R_p(\alpha_p)$, the admissible boundary-channel tuples $\mathcal C_p(\alpha_p)$ and the coefficients $A_p^N$ are independent of the loop family. The only additional choices in $\mathcal C^\Lc_\Lambda(\alpha)$ are the Wilson scalarization data. Moreover, if $e\notin D(\Lc)$, then
\[
\mathcal W_e(\Lc)=\C,
\qquad \pi^W_{e,N}=1,
\qquad C^W_{e,\omega,N}=1,
\]
and the edge kernel~\eqref{eq:projection-edge-kernel} is exactly the vacuum edge kernel associated with the same incident plaquette channels.
\end{proposition}

\begin{proof}
Everything up to and including the compressed plaquette coefficient~\eqref{eq:compressed-plaquette-amplitude} was constructed from a single plaquette character and does not involve the Wilson observable. The loop family enters for the first time through the spaces~\eqref{eq:wilson-edge-space}, the representations~\eqref{eq:wilson-edge-rep} and the decomposable Wilson factors~\eqref{eq:wilson-decomposable-expansion}. If $e\notin D(\Lc)$, there is no Wilson occurrence at $e$, so the Wilson tensor factor at this edge is the one-dimensional trivial representation and its contraction operator is the scalar identity. Therefore the edge Hilbert space~\eqref{eq:edge-hilbert-space}, the Haar projection~\eqref{eq:edge-haar-projection} and the deterministic contraction operator~\eqref{eq:edge-contraction-operator} reduce exactly to their vacuum counterparts. Taking the trace gives the vacuum kernel.
\end{proof}

\begin{proof}[Proof of Theorem~\ref{thm:spin-foam2}]
By Theorem~\ref{thm:state_sum_Wilson} and Theorem~\ref{thm:spin-foam},
\[
\E[W_{\Lambda,\Lc}(U)]
=\frac{\mathcal Z_{\Lambda,Q}^{(\Lc)}}{Z_{\Lambda,Q}}.
\]
For the empty loop family the Wilson observable is identically $1$. Applying the same identity with $\Lc=\varnothing$ therefore gives
\[
\mathcal Z_{\Lambda,Q}^{(0)}=Z_{\Lambda,Q}.
\]
Consequently,
\[
\E[W_{\Lambda,\Lc}(U)]
=\frac{\mathcal Z_{\Lambda,Q}^{(\Lc)}}{\mathcal Z_{\Lambda,Q}^{(0)}}.
\]
Proposition~\ref{prop:finite-range-defect-property} shows that numerator and denominator use the same bulk plaquette-channel data and that every edge factor outside $D(\Lc)$ is the corresponding vacuum factor. The only additional Wilson-dependent data are therefore supported on the finite defect. This proves the theorem.
\end{proof}

\section{The master loop equation}\label{sec:master-loop}

In this section we derive the universal coefficientwise master loop equation and make its closure on the topological coefficients completely explicit. Starting from a Haar integration-by-parts identity, we will identify its two main contributions: the classical cut-and-join loop term, and a new Fourier-side local recoupling term. We will then prove Theorem~\ref{thm:universal-coefficientwise-master-equation}.

\subsection{Edge derivatives and integration by parts}

Let $\Lambda$ be a finite connected oriented lattice, with oriented edge set $E(\Lambda)$ and oriented plaquette set $P(\Lambda)$. For a configuration $U=(U_e)_{e\in E(\Lambda)}\in\U(N)^{E(\Lambda)},$ and an oriented path $\gamma=e_1^{\varepsilon_1}\cdots e_m^{\varepsilon_m}$, $\varepsilon_j\in\{\pm1\}$, write $U_\gamma:=U_{e_1}^{\varepsilon_1}\cdots U_{e_m}^{\varepsilon_m}.$ For a loop family $\Lc=(\ell_1,\dots,\ell_r)$, set $W_{\Lc}(U):=\prod_{i=1}^r \Tr(U_{\ell_i}).$ Let us endow the Lie algebra $\mathfrak u(N)=\{X\in M_N(\C):\ X^*=-X\}$ with the invariant Hilbert--Schmidt inner product $\langle X,Y\rangle=-\Tr(XY)=\Tr(XY^*)$ for anti-Hermitian matrices. Let $(X_a)_{a=1}^{N^2}$ be an orthonormal basis of $\mathfrak u(N)$. For this inner product, one has the so-called \emph{magic formulas}
\begin{equation}\label{eq:magic}
\sum_{a=1}^{N^2}X_a^2=-NI_N,\qquad \sum_{a=1}^{N^2} X_aAX_a=-\Tr(A)I_N,\qquad \sum_{a=1}^{N^2} \Tr(X_aA)\Tr(X_aB)=-\Tr(AB).
\end{equation}
The name ``magic formulas'' is due to Driver--Hall--Kemp \cite{DHK13}, but they were already known before by Sengupta \cite{Sen08}. Note that we use a different normalization for the inner product on $\mathfrak u(N)$, which yields a different scaling in $N$.

For $e\in E(\Lambda)$ and $F\in \mathscr{C}^1(\U(N)^{E(\Lambda)})$, define
\[
(\partial_{e,a}F)(U):=\left.\frac{d}{dt}\right|_{t=0} F\bigl((U_{e'})_{e'\neq e}, e^{tX_a}U_e\bigr).
\]
For $F\in \mathscr{C}^2(\U(N)^{E(\Lambda)})$, define by $\Delta_e F:=\sum_{a=1}^{N^2}\partial_{e,a}^2 F$ the local Laplacian acting on $U_e$. For a plaquette decoration $\alpha:P(\Lambda)\to \widehat{\U(N)}$, define
\[
F_{\Lc,\alpha}(U):=W_{\Lc}(U)\prod_{p\in P(\Lambda)} \chi_{\alpha_p}(U_{\partial p}).
\]

\begin{proposition}
\label{prop:first-order-haar-ibp}
For every oriented edge $e\in E(\Lambda)$, every loop family $\Lc$, and every plaquette decoration
$\alpha:P(\Lambda)\to \widehat{\U(N)}$, one has
\begin{align}
0 &= \int_{\U(N)^{E(\Lambda)}}\sum_{a=1}^{N^2}\partial_{e,a}\Bigl((\partial_{e,a}W_{\Lc})\prod_{p\in P(\Lambda)}\chi_{\alpha_p}(U_{\partial p})\Bigr) dU \notag \\
&=\int_{\U(N)^{E(\Lambda)}}(\Delta_e W_{\Lc})\prod_{p\in P(\Lambda)}\chi_{\alpha_p}(U_{\partial p}) dU\notag \\
&\qquad+\sum_{p\ni e}\int_{\U(N)^{E(\Lambda)}}\sum_{a=1}^{N^2}(\partial_{e,a}W_{\Lc})\partial_{e,a}\chi_{\alpha_p}(U_{\partial p})\prod_{q\neq p}\chi_{\alpha_q}(U_{\partial q}) dU.\label{eq:first-order-coefficientwise-ibp}
\end{align}
\end{proposition}

\begin{proof}
Fix $e\in E(\Lambda)$ and $a\in\{1,\dots,N^2\}$. For any smooth function $G\in C^1(\U(N)^{E(\Lambda)})$, Haar invariance on the copy of $\U(N)$ attached to the edge $e$ implies
\[
\int_{\U(N)^{E(\Lambda)}} G\bigl((U_{e'})_{e'\neq e},e^{tX_a}U_e\bigr)dU = \int_{\U(N)^{E(\Lambda)}} G(U)dU \qquad (t\in\mathbb R),
\]
because the product Haar measure is invariant under left translation in each coordinate. Differentiating at $t=0$ gives
\[
\int_{\U(N)^{E(\Lambda)}} \partial_{e,a}G(U)dU=0.
\]
Apply this identity to
\[
G(U):=(\partial_{e,a}W_{\Lc})(U)\prod_{p\in P(\Lambda)}\chi_{\alpha_p}(U_{\partial p}).
\]
We obtain
\[
0=
\int_{\U(N)^{E(\Lambda)}}\partial_{e,a}\left((\partial_{e,a}W_{\Lc})(U)\prod_{p\in P(\Lambda)}\chi_{\alpha_p}(U_{\partial p})\right)dU.
\]
Summing over $a=1,\dots,N^2$ yields the first displayed identity in the statement.

We now expand the derivative by the Leibniz rule. Since
\[
\sum_{a=1}^{N^2}\partial_{e,a}^2W_{\Lc}=\Delta_e W_{\Lc},
\]
the contribution in which the derivative falls twice on $W_{\Lc}$ is
\[
\int_{\U(N)^{E(\Lambda)}} (\Delta_eW_{\Lc})(U)\prod_{p\in P(\Lambda)}\chi_{\alpha_p}(U_{\partial p})dU.
\]
For the derivative of the plaquette part, note that $\chi_{\alpha_p}(U_{\partial p})$ depends on
$U_e$ if and only if the edge $e$ lies on the boundary of $p$. Therefore
\[
\partial_{e,a}\left(\prod_{p\in P(\Lambda)}\chi_{\alpha_p}(U_{\partial p})\right) =\sum_{p\ni e}\bigl(\partial_{e,a}\chi_{\alpha_p}(U_{\partial p})\bigr) \prod_{q\neq p}\chi_{\alpha_q}(U_{\partial q}).
\]
Substituting this into the previous identity gives
\begin{align*}
0={}&\int_{\U(N)^{E(\Lambda)}}
(\Delta_eW_{\Lc})(U)
\prod_{p\in P(\Lambda)}\chi_{\alpha_p}(U_{\partial p})\,dU\\
&+\sum_{p\ni e}\int_{\U(N)^{E(\Lambda)}}
\sum_{a=1}^{N^2}(\partial_{e,a}W_{\Lc})
\partial_{e,a}\chi_{\alpha_p}(U_{\partial p})
\prod_{q\neq p}\chi_{\alpha_q}(U_{\partial q})\,dU,
\end{align*}
which is exactly the second displayed identity.
\end{proof}

Proposition~\ref{prop:first-order-haar-ibp} contains the essential information for the master loop equation, and it remains to find the right interpretation. Equation~\eqref{eq:first-order-coefficientwise-ibp} can be rewritten as a relationship between two terms:
\begin{equation}
I_{\Delta}+ I_{\mathrm{mix}}=0
\end{equation}
where
\[
I_\Delta=\int_{\U(N)^{E(\Lambda)}}(\Delta_e W_{\Lc})\prod_{p\in P(\Lambda)}\chi_{\alpha_p}(U_{\partial p}) dU,
\]
and
\[
I_\mathrm{mix}=\sum_{p\ni e}\int_{\U(N)^{E(\Lambda)}}\sum_{a=1}^{N^2}(\partial_{e,a}W_{\Lc})\partial_{e,a}\chi_{\alpha_p}(U_{\partial p})\prod_{q\neq p}\chi_{\alpha_q}(U_{\partial q}) dU.
\]
We will provide interpretations of these terms in separate subsections: the first one will be related to surgery of loop families, and the second one involves simultaneously deformation of loops and branching rules on plaquette decorations.

\subsection{Cut-and-join term}\label{sec:cut-and-join}

Let us start with $I_\Delta$. Its interpretation will involve standard local loop surgeries, that can be found in \cite{Cha19,Jaf16,SSZ24,SSZ25b} for instance. In particular, all surgeries we present here are conveniently illustrated in the aforementioned papers and we advise the reader to look at them for a better understanding. The main result will be Proposition~\ref{prop:universal-loop-laplacian-term}. Before we state it, recall a bit of terminology.

Let $\Lc=(\ell_1,\ldots,\ell_k)$ be a finite family of loops. In order to describe local surgeries at an active edge, we choose for each loop $\ell_i$ a cyclic word representative
\[
w_{\ell_i}=e_{i,1}^{\varepsilon_{i,1}}\cdots e_{i,m_i}^{\varepsilon_{i,m_i}}, \qquad \varepsilon_{i,r}\in\{\pm1\},
\]
in oriented edges of $\Lambda$. We write $|\ell_i|:=m_i$ and $\ell(\Lc):=\sum_{i=1}^k |\ell_i|$ for the total loop length. All constructions below are independent of the choice of cyclic representative, up to cyclic rotation of the resulting loops. When needed, we denote by $[w]$ the equivalence class of the word $w$ for backtracking equivalence. We denote by $\widehat{D}_i(\Lc)$ the set of all edges $e\in E(\Lambda)$ such that $e$ or $e^{-1}$ is contained in $\ell_i$, and by $D_i(\Lc)$ a fixed subset of $\widehat{D}_i(\Lc)$ keeping only one oriented edge among $e$ and $e^{-1}$. For $e\in D_i(\Lc)$, define
\[
A_i(e;\Lc):=\{r\in\{1,\dots,m_i\}: e_{i,r}^{\varepsilon_{i,r}}=e\},\qquad B_i(e;\Lc):=\{r\in\{1,\dots,m_i\}: e_{i,r}^{\varepsilon_{i,r}}=e^{-1}\},
\]
and
\[
C_i(e;\Lc):=A_i(e;\Lc)\cup B_i(e;\Lc),\qquad m_i(e;\Lc):=|C_i(e;\Lc)|.
\]
For $j\neq i$, define similarly $A_j(e;\Lc)$, $B_j(e;\Lc)$ and $C_j(e;\Lc)$.

\begin{definition}[Splittings]
Let us fix $e\in D_i(\Lc)$, and
$x\neq y\in C_i(e;\Lc)$.

\begin{itemize}
\item If $x,y\in A_i(e;\Lc)$, write $w_{\ell_i}=aebec.$ The associated positive splitting is the ordered pair
\[
\times^1_{x,y}\ell_i:=[aec], \qquad \times^2_{x,y}\ell_i:=[be].
\]

\item If $x,y\in B_i(e;\Lc)$, write $w_{\ell_i}=ae^{-1}be^{-1}c.$ The associated positive splitting is again
\[
\times^1_{x,y}\ell_i:=[aec],\qquad \times^2_{x,y}\ell_i:=[be^{-1}].
\]

\item If $x\in A_i(e;\Lc)$ and $y\in B_i(e;\Lc)$, write $w_{\ell_i}=aebe^{-1}c.$ The associated negative splitting is the ordered pair
\[
\times^1_{x,y}\ell_i:=[ac], \qquad \times^2_{x,y}\ell_i:=[b].
\]

\item If $x\in B_i(e;\Lc)$ and $y\in A_i(e;\Lc)$, write $w_{\ell_i}=ae^{-1}bec.$ The associated negative splitting is the ordered pair
\[
\times^1_{x,y}\ell_i:=[ac],\qquad \times^2_{x,y}\ell_i:=[b].
\]
\end{itemize}
\end{definition}

\begin{definition}[Mergers]
Let us fix $e\in D_i(\Lc)$, and
$x\in C_i(e;\Lc)$ and $y\in C_j(e;\Lc)$.

\begin{itemize}
\item If $x\in A_i(e;\Lc)$ and $y\in A_j(e;\Lc)$, write $w_{\ell_i}=aeb,$ $w_{\ell_j}=ced.$
Define
\[
\ell_i\oplus_{x,y}\ell_j:=[aedceb],\qquad \ell_i\ominus_{x,y}\ell_j:=[ac^{-1}d^{-1}b].
\]
\item If $x\in B_i(e;\Lc)$ and $y\in B_j(e;\Lc)$, write $w_{\ell_i}=ae^{-1}b,$  $w_{\ell_j}=ce^{-1}d.$
Define
\[
\ell_i\oplus_{x,y}\ell_j:=[aedceb],\qquad \ell_i\ominus_{x,y}\ell_j:=[ac^{-1}d^{-1}b].
\]

\item If $x\in A_i(e;\Lc)$ and $y\in B_j(e;\Lc)$, write $w_{\ell_i}=aeb,$ $w_{\ell_j}=ce^{-1}d.$
Define
\[
\ell_i\oplus_{x,y}\ell_j:=[aec^{-1}d^{-1}eb], \qquad \ell_i\ominus_{x,y}\ell_j:=[adcb].
\]

\item If $x\in B_i(e;\Lc)$ and $y\in A_j(e;\Lc)$, write $w_{\ell_i}=ae^{-1}b,$ $w_{\ell_j}=ced.$
Define
\[
\ell_i\oplus_{x,y}\ell_j:=[aec^{-1}d^{-1}eb],\qquad \ell_i\ominus_{x,y}\ell_j:=[adcb].
\]
\end{itemize}
\end{definition}

\begin{definition}[Plaquette deformations]
Let us fix $e\in D_i(\Lc)$, $i\in\{1,\dots,k\}$, and let $x\in C_i(e;\Lc)$ be an active occurrence of an oriented edge $e$ or $e^{-1}$ in $w_{\ell_i}$. Let $p$ be a positively oriented plaquette whose boundary $\partial p$ contains the same oriented edge occurrence (we denote by $P^+(e)$ the set of such plaquettes). We define
\[
\ell_i\oplus_x p,\qquad \ell_i\ominus_x p,
\]
to be the positive and negative merger of $\ell_i$ with the oriented boundary loop $\partial p$ at the active incidence.
\end{definition}

\begin{definition}[Local surgery multisets on ordered loop families]
Let $\Lc=(\ell_1,\dots,\ell_k)$, let $i\in\{1,\dots,k\}$, and let $e\in D_i(\Lc)$.

\begin{itemize}
\item The positive and negative splitting multisets are
\[
S^+_{i,e}(\Lc):=\Bigl\{(\ell_1,\dots,\ell_{i-1},\times^1_{x,y}\ell_i,\times^2_{x,y}\ell_i,\ell_{i+1},\dots,\ell_k):x,y\in A_i(e;\Lc),\ x\neq y\Bigr\}
\]
\[
\sqcup\ \Bigl\{(\ell_1,\dots,\ell_{i-1},\times^1_{x,y}\ell_i,\times^2_{x,y}\ell_i,\ell_{i+1},\dots,\ell_k):x,y\in B_i(e;\Lc),\ x\neq y\Bigr\},
\]
and
\[
S^-_{i,e}(\Lc):=\Bigl\{(\ell_1,\dots,\ell_{i-1},\times^1_{x,y}\ell_i,\times^2_{x,y}\ell_i,\ell_{i+1},\dots,\ell_k):x\in A_i(e;\Lc),\ y\in B_i(e;\Lc)\Bigr\}
\]
\[
\sqcup\ \Bigl\{(\ell_1,\dots,\ell_{i-1},\times^1_{x,y}\ell_i,\times^2_{x,y}\ell_i,\ell_{i+1},\dots,\ell_k):x\in B_i(e;\Lc),\ y\in A_i(e;\Lc)\Bigr\}.
\]

\item The positive and negative $\U(N)$-merger multisets are
\[
M^+_{U,i,e}(\Lc):=\bigsqcup_{j\neq i}\Bigl\{\text{replace }\ell_i,\ell_j\text{ by }\ell_i\oplus_{x,y}\ell_j:x\in A_i(e;\Lc),\ y\in A_j(e;\Lc)\Bigr\}
\]
\[
\sqcup\ \bigsqcup_{j\neq i}\Bigl\{\text{replace }\ell_i,\ell_j\text{ by }\ell_i\oplus_{x,y}\ell_j:x\in B_i(e;\Lc),\ y\in B_j(e;\Lc)\Bigr\},
\]
and
\[
M^-_{U,i,e}(\Lc):=\bigsqcup_{j\neq i}\Bigl\{\text{replace }\ell_i,\ell_j\text{ by }\ell_i\ominus_{x,y}\ell_j:x\in A_i(e;\Lc),\ y\in B_j(e;\Lc)\Bigr\}
\]
\[
\sqcup\ \bigsqcup_{j\neq i}\Bigl\{\text{replace }\ell_i,\ell_j\text{ by }\ell_i\ominus_{x,y}\ell_j:x\in B_i(e;\Lc),\ y\in A_j(e;\Lc)\Bigr\}.
\]

\item The positive and negative deformation multisets are
\[
D^+_{i,e}(\Lc):=\Bigl\{(\ell_1,\dots,\ell_{i-1},\ell_i\oplus_x p,\ell_{i+1},\dots,\ell_k):p\in P^+(e),\ x\in C_i(e;\Lc)\Bigr\},
\]
and
\[
D^-_{i,e}(\Lc):=\Bigl\{(\ell_1,\dots,\ell_{i-1},\ell_i\ominus_x p,\ell_{i+1},\dots,\ell_k):p\in P^+(e),\ x\in C_i(e;\Lc)\Bigr\}.
\]
\end{itemize}
\end{definition}

We now return to the coefficient
\[
\widehat{W}_{\Lambda,\Lc}(\alpha)=\int_{\U(N)^{E(\Lambda)}} F_{\Lc,\alpha}(U) dU.
\]
Let $e\in E(\Lambda)$ be an oriented edge. For a loop family $\Lc$ and a plaquette decoration $\alpha$, define
\begin{align*}
\mathscr L_e \widehat{W}_{\Lambda,\Lc}(\alpha) &:=-N\sum_{i:e\in D_i(\Lc)}m_i(e;\Lc) \widehat{W}_{\Lambda,\Lc}(\alpha) \\
&\quad+ \sum_{i:e\in D_i(\Lc)}\ \sum_{\Lc'\in S^-_{i,e}(\Lc)} \widehat{W}_{\Lambda,\Lc'}(\alpha)- \sum_{i:e\in D_i(\Lc)}\ \sum_{\Lc'\in S^+_{i,e}(\Lc)} \widehat{W}_{\Lambda,\Lc'}(\alpha) \\
&\quad+ \sum_{i:e\in D_i(\Lc)}\ \sum_{\Lc'\in M^-_{U,i,e}(\Lc)} \widehat{W}_{\Lambda,\Lc'}(\alpha)- \sum_{i:e\in D_i(\Lc)}\ \sum_{\Lc'\in M^+_{U,i,e}(\Lc)} \widehat{W}_{\Lambda,\Lc'}(\alpha).
\end{align*}

This operator $\mathscr{L}_e$ can be described as follows: it turns the topological coefficient associated to the loop family $\Lc$ into a linear combination, with coefficients $-N$ and $\pm1$, of topological coefficients associated to loop families $\Lc'$ obtained from $\Lc$ by splitting and merger. The next proposition shows that it encodes the behavior of the Laplacian acting on the loop part of the coefficient, namely $I_\Delta$.

\begin{proposition}
\label{prop:universal-loop-laplacian-term}
For every oriented edge $e$, loop sequence $\Lc$, and plaquette decoration $\alpha$, one has
\[
I_\Delta:=\int_{\U(N)^{E(\Lambda)}}(\Delta_e W_{\Lc}(U))\prod_{p\in P(\Lambda)}\chi_{\alpha_p}(U_{\partial p}) dU=(\mathscr L_e \widehat{W}_{\Lambda,\Lc})(\alpha).
\]
\end{proposition}

\begin{proof}
We prove first a pointwise identity for $\Delta_eW_{\Lc}$, and then integrate against the undifferentiated plaquette factors $\prod_{p\in P(\Lambda)}\chi_{\alpha_p}(U_{\partial p})$. Fix an oriented edge $e\in E(\Lambda)$. For each $i\in\{1,\dots,k\}$, choose the cyclic representative
\[
w_{\ell_i}=e_{i,1}^{\varepsilon_{i,1}}\cdots e_{i,m_i}^{\varepsilon_{i,m_i}},\qquad \varepsilon_{i,r}\in\{\pm1\}.
\]
If $x\in A_i(e;\Lc)$, write $w_{\ell_i}=a_xeb_x,$ and if $x\in B_i(e;\Lc)$, write $w_{\ell_i}=a_xe^{-1}b_x.$ By differentiation, we have
\[
\partial_{e,a}\Tr(U_{\ell_i})=\sum_{x\in A_i(e;\Lc)}\Tr\bigl(U_{a_x}X_aU_eU_{b_x}\bigr)-\sum_{x\in B_i(e;\Lc)}\Tr\bigl(U_{a_x}U_e^{-1}X_aU_{b_x}\bigr).
\]
Using cyclicity of the trace, this may be rewritten as
\begin{equation}\label{eq:local-first-derivative-loop}
\partial_{e,a}\Tr(U_{\ell_i})=\sum_{x\in A_i(e;\Lc)} \Tr\bigl(X_aU_{e b_x a_x}\bigr)-\sum_{x\in B_i(e;\Lc)} \Tr\bigl(X_aU_{b_x a_x e^{-1}}\bigr).
\end{equation}

Since $W_{\Lc}=\prod_{i=1}^k \Tr(U_{\ell_i})$, the first derivative is
\[
\partial_{e,a}W_{\Lc}=\sum_{i:e\in D_i(\Lc)}\left(\prod_{j\neq i}\Tr(U_{\ell_j})\right)\partial_{e,a}\Tr(U_{\ell_i}),
\]
and therefore, by differentiating once more and summing over $a$,
\[
\Delta_eW_{\Lc}=\sum_{i=1}^k\left(\prod_{j\neq i}\Tr(U_{\ell_j})\right)\Delta_e\Tr(U_{\ell_i})+\sum_{i\neq j}\sum_{a=1}^{N^2}\left(\prod_{m\neq i,j}\Tr(U_{\ell_m})\right)\bigl(\partial_{e,a}\Tr(U_{\ell_i})\bigr)\bigl(\partial_{e,a}\Tr(U_{\ell_j})\bigr).
\]
We analyze separately the two terms on the right-hand side and use the magic formulas~\eqref{eq:magic}.

\smallskip

\emph{Step 1: second derivatives inside a single loop.}
Fix $i$ with $e\in D_i(\Lc)$. Differentiating \eqref{eq:local-first-derivative-loop} once more, the two derivatives may hit either the same active occurrence or two distinct active occurrences.
\begin{itemize}
\item If they hit the same active occurrence $x\in C_i(e;\Lc)$, then one gets a factor $\sum_a X_a^2$, hence, with the normalization matching the definition of $\Lc_e$,
\[
\sum_{a=1}^{N^2}\partial_{e,a}^2 \Tr(U_{\ell_i})\Big|_{\text{same }x}=-N\Tr(U_{\ell_i}).
\]
Summing over the $m_i(e;\Lc)=|C_i(e;\Lc)|$ active occurrences gives the contribution
\[
-Nm_i(e;\Lc)\Tr(U_{\ell_i}).
\]
\item If they hit two distinct active occurrences $x\neq y\in C_i(e;\Lc)$, we have to distinguish between four cases.
\begin{enumerate}
\item If $x,y\in A_i(e;\Lc)$, write $w_{\ell_i}=aebec$. Then the contribution is
\[
\sum_{a=1}^{N^2}\Tr(U_aX_aU_eU_bX_aU_eU_c).
\]
By cyclicity of the trace and the standard $\U(N)$ contraction identity in the chosen
normalization,
\[
\sum_{a=1}^{N^2}\Tr(AX_aBX_a)=-\Tr(A)\Tr(B),
\]
this equals $-\Tr(U_{aec})\Tr(U_{be}),$ which is exactly the Wilson weight associated with the positive splitting $(\times^1_{x,y}\ell_i,\times^2_{x,y}\ell_i)$.
\item If $x,y\in B_i(e;\Lc)$, writing $w_{\ell_i}=ae^{-1}be^{-1}c$, the same computation gives $-\Tr(U_{aec})\Tr(U_{be^{-1}}),$
 i.e. the positive splitting contribution again.
\item If $x\in A_i(e;\Lc)$ and $y\in B_i(e;\Lc)$, write $w_{\ell_i}=aebe^{-1}c$. Then the second
derivative carries one minus sign from differentiating the inverse occurrence, and one gets
\[
-\sum_{a=1}^{N^2}\Tr(U_aX_aU_eU_bU_e^{-1}X_aU_c).
\]
Using cyclicity and the same contraction identity, this becomes $\Tr(U_{ac})\Tr(U_b),$ which is exactly the contribution of the negative splitting $(\times^1_{x,y}\ell_i,\times^2_{x,y}\ell_i)$.
\item If $x\in B_i(e;\Lc)$ and $y\in A_i(e;\Lc)$, the case is identical to (iii) and gives the other half of
$S^-_{i,e}(\Lc)$.
\end{enumerate}
\end{itemize}

Collecting all same-loop contributions, we obtain
\begin{align}
\Delta_e\Tr(U_{\ell_i}) &=-Nm_i(e;\Lc)\Tr(U_{\ell_i})+ \sum_{\Lc'\in S^-_{i,e}(\Lc)} W_{\Lc'}^{(i)}- \sum_{\Lc'\in S^+_{i,e}(\Lc)} W_{\Lc'}^{(i)},\label{eq:single-loop-Laplacian}
\end{align}
where $W_{\Lc'}^{(i)}$ means the product of traces obtained from $\Lc'$ after reinserting the undifferentiated loops $\ell_j$, $j\neq i$.

\smallskip

\emph{Step 2: cross terms between two different loops.}
Fix $i\neq j$ with $e\in D_i(\Lc)\cap D_j(\Lc)$. By \eqref{eq:local-first-derivative-loop}, the cross term is a sum over pairs of active occurrences $x\in C_i(e;\Lc)$, $y\in C_j(e;\Lc)$.

If $x,y$ have the same orientation, say $x\in A_i(e;\Lc)$, $y\in A_j(e;\Lc)$, write $w_{\ell_i}=aeb$ and $w_{\ell_j}=ced$. Then
\[
\sum_{a=1}^{N^2}\Tr(U_aX_aU_eU_b)\Tr(U_cX_aU_eU_d).
\]
Using the contraction identity
\[
\sum_{a=1}^{N^2}\Tr(X_aA)\Tr(X_aB)=-\Tr(AB)
\]
in the normalization matching $\Lc_e$, this becomes
\[
-\Tr(U_{aedceb}),
\]
namely the positive $\U(N)$-merger $\ell_i\oplus_{x,y}\ell_j$.

If $x,y\in B_i(e;\Lc)\times B_j(e;\Lc)$, one obtains similarly the other half of $M^+_{U,i,e}(\Lc)$.

If $x,y$ have opposite orientations, say $x\in A_i(e;\Lc)$, $y\in B_j(e;\Lc)$, there is one minus sign coming from differentiating the inverse occurrence in $\ell_j$, and the same contraction identity yields
\[
\Tr(U_{adcb}),
\]
which is the negative $\U(N)$-merger $\ell_i\ominus_{x,y}\ell_j$. The case $x\in B_i(e;\Lc)$, $y\in A_j(e;\Lc)$ is identical.

Therefore the total cross-term contribution is
\begin{align}
\sum_{a=1}^{N^2}
\bigl(\partial_{e,a}\Tr(U_{\ell_i})\bigr)\bigl(\partial_{e,a}\Tr(U_{\ell_j})\bigr)=\sum_{\Lc'\in M^-_{U,i,e}(\Lc;j)} W_{\Lc'}-\sum_{\Lc'\in M^+_{U,i,e}(\Lc;j)} W_{\Lc'},
\label{eq:cross-loop-Laplacian}
\end{align}
where $M^\pm_{U,i,e}(\Lc;j)$ denotes the part of the merger multiset involving the pair
$(i,j)$.

\smallskip

\emph{Step 3: conclusion.}
Substituting \eqref{eq:single-loop-Laplacian} and \eqref{eq:cross-loop-Laplacian} into the product-rule expansion of $\Delta_eW_{\Lc}$, and summing over all active loops and all active pairs of loops, we obtain the pointwise identity
\[
\Delta_eW_{\Lc}(U)=-N\sum_{i:e\in D_i(\Lc)} m_i(e;\Lc)W_{\Lc}(U)+\sum_{i:e\in D_i(\Lc)}\sum_{\Lc'\in S^-_{i,e}(\Lc)} W_{\Lc'}(U)-\sum_{i:e\in D_i(\Lc)}\sum_{\Lc'\in S^+_{i,e}(\Lc)} W_{\Lc'}(U)
\]
\[
\hspace{3cm}
+\sum_{i:e\in D_i(\Lc)}\sum_{\Lc'\in M^-_{U,i,e}(\Lc)} W_{\Lc'}(U)-\sum_{i:e\in D_i(\Lc)}\sum_{\Lc'\in M^+_{U,i,e}(\Lc)} W_{\Lc'}(U).
\]
Multiplying this identity by $\prod_{p\in P(\Lambda)}\chi_{\alpha_p}(U_{\partial p})$ and
integrating over $\U(N)^{E(\Lambda)}$ gives the expected formula.
\end{proof}

\subsection{Mixed terms}

The next proposition isolates the genuinely new part of the coefficientwise master loop equation. Indeed, Proposition~\ref{prop:universal-loop-laplacian-term} treats the loop-Laplacian term, and this is still governed by the classical cut-and-join calculus on loop families. Its proof is therefore very close to previous proofs obtained for the Wilson action \cite{Cha19,Jaf16,SSZ24}. By contrast, the mixed term
\[
\sum_{a=1}^{N^2} (\partial_{e,a}W_{\Lc})\partial_{e,a}\chi_{\alpha_p}(U_{\partial p})
\]
is the place where the Fourier expansion in plaquette variables interacts with the geometric loop sector. We will see in Proposition~\ref{prop:local-incidence-tensor} that this interaction remains local: after differentiating at a fixed incidence $(e,p)$, one does not generate any nonlocal reorganization of the whole plaquette decoration, but only a finite list of local reconnection types at the active incidence together with a finite local recoupling of the active plaquette label $\alpha_p$. This is the mechanism behind the closure of the master loop equation on the family of topological coefficients $\widehat{W}_{\Lambda,\Lc}(\alpha)$. The resulting operator acts on the topological coefficients $\widehat{W}_{\Lambda,\Lc}(\alpha)$ while the loop family is still kept in geometric position space.

\begin{definition}[loop-plaquette mergers]
Fix an oriented incidence $(e,p)$, and choose a cyclic representative of the plaquette boundary $\partial p = aeb.$ Let
\[
\gamma_{p,u}:=\begin{cases}
aeb=\partial p, & 1\le u\le n_p^+,\\[1mm]
b^{-1}e^{-1}a^{-1}=(\partial p)^{-1}, & n_p^+<u\le r_p,
\end{cases}
\]
so that covariant slots carry the positively oriented plaquette word and contravariant slots carry the oppositely oriented plaquette word.

Let $\Lc=(\ell_1,\dots,\ell_k)$ be a loop family, and let $x$ be an active occurrence of $e$ or $e^{-1}$ in the loop $\ell_i$. Write
\[
w_{\ell_i}=
\begin{cases}
ced, & x\in A_i(e;\Lc),\\
ce^{-1}d, & x\in B_i(e;\Lc),
\end{cases}
\]
where the displayed occurrence of $e$ or $e^{-1}$ is the active one corresponding to $x$. The \emph{loop-plaquette merger} $\tau_{x,u}(\Lc)$ is the loop family obtained from $\Lc$ by replacing $\ell_i$ with the loop produced by merging $\ell_i$ and $\gamma_{p,u}$ at their active occurrences. More precisely, if $y_u$ denotes the active occurrence in the slot-word $\gamma_{p,u}$, then
\[
\tau_{x,u}(\Lc)=\begin{cases}
\text{replace }\ell_i\text{ by }\ell_i\oplus_{x,y_u}\gamma_{p,u},
& \text{if }x\text{ and }y_u\text{ have the same orientation},\\
\text{replace }\ell_i\text{ by }\ell_i\ominus_{x,y_u}\gamma_{p,u},
& \text{if }x\text{ and }y_u\text{ have opposite orientations}.
\end{cases}
\]
\end{definition}

Given a loop-plaquette merger $\tau_{x,u}(\Lc)$ as above, we define the auxiliary mixed tensor space
\[
T^{\mathrm{aux}}_{x,u}:=
\begin{cases}
V^{\otimes(n_p^+-1)}\otimes (V^*)^{\otimes n_p^-}, & x\in A_i(e;\Lc),\ 1\le u\le n_p^+,\\[1mm]
V^{\otimes(n_p^++1)}\otimes (V^*)^{\otimes n_p^-}, & x\in B_i(e;\Lc),\ 1\le u\le n_p^+,\\[1mm]
V^{\otimes n_p^+}\otimes (V^*)^{\otimes(n_p^-+1)}, & x\in A_i(e;\Lc),\ n_p^+<u\le r_p,\\[1mm]
V^{\otimes n_p^+}\otimes (V^*)^{\otimes(n_p^--1)}, & x\in B_i(e;\Lc),\ n_p^+<u\le r_p.
\end{cases}
\]
These correspond to the following four local cases
\[
\begin{array}{c|c|c}
\text{local type of }(x,u) & \tau_{x,u}(\Lc) & \text{plaquette operation} \\ \hline
x\in A_i(e;\Lc),\ u\le n_p^+ & \ell_i\oplus_{x,y_u}\gamma_{p,u} & \text{covariant contraction} \\
x\in B_i(e;\Lc),\ u\le n_p^+ & \ell_i\ominus_{x,y_u}\gamma_{p,u} & \text{covariant coevaluation} \\
x\in A_i(e;\Lc),\ u>n_p^+ & \ell_i\ominus_{x,y_u}\gamma_{p,u} & \text{contravariant coevaluation} \\
x\in B_i(e;\Lc),\ u>n_p^+ & \ell_i\oplus_{x,y_u}\gamma_{p,u} & \text{contravariant contraction}
\end{array}
\]
where $i=i(x)$ denotes the index of the loop containing $x$, and $y_u$ denotes the active occurrence in the slot-word $\gamma_{p,u}$. The next result translates the mixed term of the master loop equation in terms of these loop-plaquette mergers.

\begin{proposition}\label{prop:local-incidence-tensor}
Fix an oriented incidence $(e,p)$, and choose a cyclic representative $\partial p=aeb$. Write $\alpha_p=[\lambda_p^+,\lambda_p^-]_N,\ n_p^+:=|\lambda_p^+|,\ n_p^-:=|\lambda_p^-|,\ r_p:=n_p^++n_p^-,$ and let
\[
T_p:=V^{\otimes n_p^+}\otimes (V^*)^{\otimes n_p^-},\qquad P_p:=P^{[\lambda_p^+,\lambda_p^-]}_N.
\]
Then, for each pair $(x,u)$, there exists an endomorphism $R_{x,u}(U)\in \mathrm{End}(T^{\mathrm{aux}}_{x,u})$ such that
\[
\sum_{a=1}^{N^2}(\partial_{e,a}W_\Lc)(U)\,\partial_{e,a}\chi_{\alpha_p}(U_{\partial p})=\sum_x\sum_{u=1}^{r_p} W_{\tau_{x,u}(\Lc)}(U)\,\mathrm{Tr}_{T^{\mathrm{aux}}_{x,u}}\!\bigl(R_{x,u}(U)\bigr),
\]
where the outer sum runs over all active occurrences $x$ of $e$ or $e^{-1}$ in $\Lc$.
\end{proposition}

\begin{proof}
We start from the mixed-tensor expression
\[
\chi_{\alpha_p}(U_{\partial p})=\mathrm{Tr}_{T_p}\!\bigl(P_p\,\rho_p(U_a)\rho_p(U_e)\rho_p(U_b)\bigr).
\]
For the loop part, if $x$ is an active occurrence of $e$ or $e^{-1}$ in the loop $\ell_i$, choose a cyclic representative of $\ell_i$ of the form
\[
w_{\ell_i}=
\begin{cases}
cxe\,d_x,& x\in A_i(e;\Lc),\\
cxe^{-1}d_x,& x\in B_i(e;\Lc),
\end{cases}
\]
and write
\[
\Xi_{x,a}(U):=
\begin{cases}
\mathrm{Tr}(U_{c_x}X_aU_eU_{d_x}),& x\in A_i(e;\Lc),\\
-\mathrm{Tr}(U_{c_x}U_e^{-1}X_aU_{d_x}),& x\in B_i(e;\Lc).
\end{cases}
\]
Then
\[
\partial_{e,a}W_{\Lc}(U)=\sum_x W_{\Lc\setminus x}(U)\,\Xi_{x,a}(U),
\]
where $W_{\Lc\setminus x}(U)$ denotes the product of the loop traces not differentiated at the occurrence $x$.

For the plaquette part, differentiating only the active factor $\rho_p(U_e)$ gives
\[
\partial_{e,a}\chi_{\alpha_p}(U_{\partial p})=\sum_{u=1}^{r_p}
\mathrm{Tr}_{T_p}\!\bigl(P_p\,\rho_p(U_a)\,D_{u,a}\,\rho_p(U_b)\bigr),
\]
where $D_{u,a}$ is the infinitesimal action on the $u$-th slot. Thus
\[
\sum_{a=1}^{N^2}(\partial_{e,a}W_{\Lc})\,\partial_{e,a}\chi_{\alpha_p}(U_{\partial p})=\sum_x\sum_{u=1}^{r_p}W_{\Lc\setminus x}(U)\sum_{a=1}^{N^2} \Xi_{x,a}(U)\,\mathrm{Tr}_{T_p}\!\bigl(P_p\,\rho_p(U_a)\,D_{u,a}\,\rho_p(U_b)\bigr).
\]
Thus, it is enough to analyze a fixed pair $(x,u)$.

For such a pair, the index $a$ appears exactly once in the loop factor and exactly once in the active plaquette slot. Summing over $a$ therefore performs the unique local Casimir contraction at the incidence $(e,p)$.  The resulting contraction or coevaluation is $\U(N)$-equivariant, but the remaining endomorphism still contains the plaquette holonomy.  We therefore call it simply an endomorphism here; Proposition~\ref{prop:explicit-incidence-recoupling} will factor its gauge-variable dependence as a fixed equivariant operator times the natural mixed-tensor representation of $U_{\partial p}$. There are four cases.

If $x\in A_i(e;\Lc)$ and $u\le n_p^+$, the active loop occurrence and the active slot occurrence $y_u$ in $\gamma_{p,u}=\partial p$ have the same orientation. The Casimir contraction reconnects the loop strand with the active plaquette strand into the positive loop-plaquette merger $\tau_{x,u}(\Lc)=\ell_i\oplus_{x,y_u}\gamma_{p,u}$, while on the plaquette side it removes the active covariant slot. Hence the remaining plaquette factor is the trace of an endomorphism of
\[
V^{\otimes(n_p^+-1)}\otimes (V^*)^{\otimes n_p^-}.
\]

If $x\in B_i(e;\Lc)$ and $u\le n_p^+$, the active loop occurrence and $y_u$ have opposite orientations. The Casimir contraction yields the negative loop-plaquette merger $\tau_{x,u}(\Lc)=\ell_i\ominus_{x,y_u}\gamma_{p,u}$, while on the plaquette side it inserts one additional covariant slot, i.e. a covariant coevaluation. The resulting plaquette factor is therefore the trace of an endomorphism of
\[
V^{\otimes(n_p^++1)}\otimes (V^*)^{\otimes n_p^-}.
\]

If $x\in A_i(e;\Lc)$ and $u>n_p^+$, then $\gamma_{p,u}=(\partial p)^{-1}$, so the active loop occurrence and $y_u$ have opposite orientations. The loop part is the negative loop-plaquette merger $\tau_{x,u}(\Lc)=\ell_i\ominus_{x,y_u}\gamma_{p,u}$, while the plaquette side acquires one additional contravariant slot, i.e. a contravariant coevaluation. Thus the remaining plaquette
factor is the trace of an endomorphism of
\[
V^{\otimes n_p^+}\otimes (V^*)^{\otimes(n_p^-+1)}.
\]

Finally, if $x\in B_i(e;\Lc)$ and $u>n_p^+$, the active loop occurrence and $y_u$ have the same orientation. The loop part is the positive loop-plaquette merger $\tau_{x,u}(\Lc)=\ell_i\oplus_{x,y_u}\gamma_{p,u}$, while the plaquette side loses one contravariant slot, i.e. undergoes a contravariant contraction. Hence the remaining plaquette
factor is the trace of an endomorphism of
\[
V^{\otimes n_p^+}\otimes (V^*)^{\otimes(n_p^--1)}.
\]

In each case, all non-active loop and plaquette letters remain unchanged, so the contribution of the fixed pair $(x,u)$ has the form
\[
W_{\tau_{x,u}(\Lc)}(U)\,\mathrm{Tr}_{T^{\mathrm{aux}}_{x,u}}\!\bigl(R_{x,u}(U)\bigr).
\]
Summing over all active occurrences $x$ and all slots $u$ gives the stated identity.
\end{proof}

\begin{proposition}\label{prop:explicit-incidence-recoupling}
In the setting of Proposition~\ref{prop:local-incidence-tensor}, for each pair $(x,u)$ there is a $\U(N)$-equivariant endomorphism $A_{x,u;N}\in\End(T^{\mathrm{aux}}_{x,u})$, independent of the gauge variables, such that
\[
R_{x,u}(U)=A_{x,u;N}\,\rho_{x,u}(U_{\partial p}),
\]
where $\rho_{x,u}$ is the natural mixed-tensor representation on $T^{\mathrm{aux}}_{x,u}$. Consequently,
\[
\Tr_{T^{\mathrm{aux}}_{x,u}}\!\bigl(R_{x,u}(U)\bigr)=\sum_{\beta\in \B_{x,u}(\alpha_p)}c^N_{x,u}(\alpha_p,\beta)\,\chi_\beta(U_{\partial p}),
\]
where $\B_{x,u}(\alpha_p)$ is a finite subset of the irreducible summands of $T^{\mathrm{aux}}_{x,u}$. If $P^{T^{\mathrm{aux}}}_{\beta,N}$ denotes the orthogonal projection onto the $\beta$-isotypic component, then the coefficient is explicitly
\begin{equation}\label{eq:incidence-recoupling-coefficient}
c^N_{x,u}(\alpha_p,\beta)
=d_\beta^{-1}\Tr_{T^{\mathrm{aux}}_{x,u}}\!\left(A_{x,u;N}P^{T^{\mathrm{aux}}}_{\beta,N}\right).
\end{equation}
\end{proposition}

\begin{proof}
Fix a pair \((x,u)\) as in Proposition~\ref{prop:local-incidence-tensor}, and write $T^{\mathrm{aux}}_{x,u}=T_{n',m'}:=V^{\otimes n'}\otimes (V^*)^{\otimes m'}.$ The pair \((n',m')\) is one of the following four possibilities:
\[
(n_p^+-1,n_p^-),\qquad (n_p^++1,n_p^-),\qquad
(n_p^+,n_p^-+1),\qquad (n_p^+,n_p^--1),
\]
according to the local type of \((x,u)\). By construction in Proposition~\ref{prop:local-incidence-tensor}, the contribution indexed by \((x,u)\) is obtained from the mixed-tensor formula for \(\chi_{\alpha_p}(U_{\partial p})\) by performing exactly one local contraction or one local coevaluation at the active incidence, while leaving all other plaquette letters unchanged. Therefore there exists a \(\U(N)\)-equivariant endomorphism $A_{x,u;N}\in \End(T_{n',m'})$ such that
\[
R_{x,u}(U)=A_{x,u;N}\,\rho_{n',m'}(U_{\partial p}),
\]
where \(\rho_{n',m'}\) denotes the natural mixed tensor representation on \(T_{n',m'}\). Indeed, all dependence on the gauge variables is still carried by the plaquette holonomy \(U_{\partial p}\), whereas the local contraction/coevaluation and the projector/idempotent data produce a \(\U(N)\)-equivariant operator independent of \(U\).

Since $T_{n',m'}$ is a finite-dimensional rational $\U(N)$-module, it is completely reducible. Write
\[
T_{n',m'}\simeq \bigoplus_{\beta\in\B_{n',m'}}M_\beta\otimes V_\beta,
\qquad
\rho_{n',m'}(g)=\bigoplus_{\beta\in\B_{n',m'}}\Id_{M_\beta}\otimes\rho_\beta(g).
\]
Because $A_{x,u;N}$ is equivariant, Schur's lemma gives
\[
A_{x,u;N}=\bigoplus_{\beta\in\B_{n',m'}}A_{x,u;\beta}(N)\otimes\Id_{V_\beta}.
\]
Hence
\[
\Tr_{T_{n',m'}}\!\bigl(A_{x,u;N}\rho_{n',m'}(g)\bigr)
=\sum_{\beta\in\B_{n',m'}}\Tr_{M_\beta}(A_{x,u;\beta}(N))\,\chi_\beta(g).
\]
The orthogonal $\beta$-isotypic projector is $P^{T^{\mathrm{aux}}}_{\beta,N}=\Id_{M_\beta}\otimes\Id_{V_\beta}$ on this summand and zero elsewhere, so
\[
\Tr_{T^{\mathrm{aux}}_{x,u}}\!\left(A_{x,u;N}P^{T^{\mathrm{aux}}}_{\beta,N}\right)
=d_\beta\Tr_{M_\beta}(A_{x,u;\beta}(N)).
\]
Thus the character coefficient is exactly~\eqref{eq:incidence-recoupling-coefficient}. Taking $\B_{x,u}(\alpha_p)$ to be the set of summands for which this coefficient is non-zero proves the result.
\end{proof}

\medskip

The finite recoupling set \(\B_{x,u}(\alpha_p)\) might seem abstract at first glance, but we can describe it quite explicitly: it is related to branching rules on rational highest weights already described in \cite[\S 2.2]{Lem25} for Wilson loops on compact surfaces: if \(\lambda,\nu\) are partitions, we write \(\lambda\nearrow\nu\) when \(\nu\) is obtained from \(\lambda\) by adding one
box. Let \(\alpha_p=[\lambda^+,\lambda^-]_N\), with
\(n^\pm=|\lambda^\pm|\). Combining Proposition~\ref{prop:local-incidence-tensor} with
the rational Pieri rules recalled in \cite{Lem25}, one obtains the following
support constraints:
\[
\B_{x,u}(\alpha_p)\subset
\begin{cases}
\bigl\{[\nu,\lambda^-]_N:\nu\nearrow\lambda^+\bigr\},
&
x\in A_i(e;\Lc),\ 1\leq u\leq n^+,
\\[0.5em]
\bigl\{[\nu,\lambda^-]_N:\lambda^+\nearrow\nu\bigr\}
\cup
\bigl\{[\lambda^+,\eta]_N:\eta\nearrow\lambda^-\bigr\},
&
x\in B_i(e;\Lc),\ 1\leq u\leq n^+,
\\[0.5em]
\bigl\{[\nu,\lambda^-]_N:\nu\nearrow\lambda^+\bigr\}
\cup
\bigl\{[\lambda^+,\eta]_N:\lambda^-\nearrow\eta\bigr\},
&
x\in A_i(e;\Lc),\ n^+<u\leq n^++n^-,
\\[0.5em]
\bigl\{[\lambda^+,\eta]_N:\eta\nearrow\lambda^-\bigr\},
&
x\in B_i(e;\Lc),\ n^+<u\leq n^++n^- .
\end{cases}
\]
Hence, the spectral part of the master loop equation translates into branching rules: at the incidence \((e,p)\), it changes only the active plaquette label \(\alpha_p\), and only by adding or removing a box on \(\lambda^+\) or \(\lambda^-\). The coefficients \(b^\tau_{e,p}(\alpha_p,\beta;N)\) contain the corresponding finite-\(N\) multiplicities and local contraction amplitudes.

A direct consequence of Proposition~\ref{prop:explicit-incidence-recoupling} is that
\[
\int_{\U(N)^{E(\Lambda)}} \sum_{a=1}^{N^2}(\partial_{e,a}W_{\Lc})\partial_{e,a}\chi_{\alpha_p}(U_{\partial p})\prod_{q\neq p}\chi_{\alpha_q}(U_{\partial q})dU=\sum_{x,u}\sum_{\beta\in \B_{x,u}(\alpha_p)}c^N_{x,u}(\alpha_p,\beta)\widehat W_{\Lambda,\tau_{x,u}(\Lc)}(\alpha^{p\to \beta}).
\]
Grouping the finitely many elementary contributions indexed by the pairs $(x,u)$ according to the resulting loop family and the resulting recoupled plaquette label, we obtain the following packaged form.

Let
\[
I_{e,p}(\Lc):=\{(x,u): x \text{ is an active occurrence of } e \text{ or } e^{-1}\text{ in }\Lc,\ 1\le u\le r_p\}.
\]
Let
\[
\mathcal T_{e,p}(\Lc):=\{\tau_{x,u}(\Lc):(x,u)\in I_{e,p}(\Lc)\}
\]
be the finite set of loop families obtained from the elementary local surgeries of Proposition~\ref{prop:local-incidence-tensor}.

For each $\tau\in \mathcal T_{e,p}(\Lc)$, define
\[
\B^\tau_{e,p}(\alpha_p):=\Bigl\{\beta\in \widehat{\U(N)}:\exists (x,u)\in I_{e,p}(\Lc)\text{ such that }\tau_{x,u}(\Lc)=\tau,\ \beta\in \B_{x,u}(\alpha_p)\Bigr\}.
\]
For $\beta\in \B^\tau_{e,p}(\alpha_p)$, define
\[
b^\tau_{e,p}(\alpha_p,\beta;N):=\sum_{\substack{(x,u)\in I_{e,p}(\Lc)\\ \tau_{x,u}(\Lc)=\tau\\ \beta\in \B_{x,u}(\alpha_p)}}c^N_{x,u}(\alpha_p,\beta).
\]
Then
\begin{align*}
&\int_{\U(N)^{E(\Lambda)}}
\sum_{a=1}^{N^2}(\partial_{e,a}W_\Lc)
\partial_{e,a}\chi_{\alpha_p}(U_{\partial p})
\prod_{q\neq p}\chi_{\alpha_q}(U_{\partial q})\,dU\\
&\qquad=
\sum_{\tau\in \mathcal T_{e,p}(\Lc)}
\sum_{\beta\in \B^\tau_{e,p}(\alpha_p)}
b^\tau_{e,p}(\alpha_p,\beta;N)
\widehat{W}_{\Lambda,\tau}(\alpha^{p\to\beta}).
\end{align*}
Before packaging the previous results into the operator $\mathscr{B}_{e,p}$, let us spell out its conceptual content. We have seen that the mixed incidence term is controlled by only three kinds of local data:
\begin{enumerate}
\item the choice of an active occurrence of $e$ or $e^{-1}$ in the loop family;
\item the choice of an active tensor slot in the mixed-tensor realization of
$\chi_{\alpha_p}(U_{\partial p})$;
\item the local contraction pattern determined by the invariant pairing
$\sum_a X_a \otimes X_a$ at the incidence $(e,p)$.
\end{enumerate}

Since each of these choices is finite for fixed $(\Lc,\alpha_p)$, the whole mixed term reduces to a finite linear combination of new decorated coefficients. The loop part changes only through a local surgery at $(e,p)$, while the plaquette part changes only through a finite recoupling of the single active label $\alpha_p$. In particular, no other plaquette labels are touched. This is the coefficientwise locality statement that replaces the classical Wilson-action deformation term.

From this point of view, Propositions~\ref{prop:local-incidence-tensor} and~\ref{prop:explicit-incidence-recoupling} are best understood as hybrid position/Fourier statements. On the loop side, it still looks like a local surgery rule in the one-skeleton of the lattice. On the plaquette side, however, it is already a representation-theoretic transfer rule: the active character $\chi_{\alpha_p}$ is replaced by a finite combination of characters $\chi_\beta$, with coefficients depending only on the local incidence type. This is why the operator $\mathscr{B}_{e,p}$ is generally not diagonal in the plaquette label, and also why the subsequent specializations behave differently: for the Wilson action, resummation collapses this local transfer rule back to the usual deformation terms, whereas for the heat-kernel action it remains visible as a genuinely spectral local operator on the active plaquette label.

This local-recoupling viewpoint is also closer in spirit to recent representation-theoretic formulas in two-dimensional Yang--Mills. For example, in L\'evy's recent combinatorial formula on compact surfaces \cite{Lev26}, Wilson loop expectations for a family of loops $\Lc$ in a compact surface $\Sigma$ are expressed as sums over highest-weight assignments to the connected components of $\Sigma\setminus\Lc$, with local factors attached to intersection points. Although L\'evy's results are more compact and elegant, thanks to special properties of the two-dimensional setting, the philosophy is still similar: once one has passed to the spectral side, the interaction of Wilson loops with the gauge field is encoded by local transfer rules between nearby representation labels.

We are now in position to define the second operator involved in the master loop equation.

\begin{definition}
For an incidence $(e,p)$, define the operator $\mathscr{B}_{e,p}$ on the decorated coefficients by
\[
(\mathscr{B}_{e,p}\widehat{W}_{\Lambda,\Lc})(\alpha):=\sum_{\tau\in\mathcal T_{e,p}(\Lc)}\ \sum_{\beta\in \B^\tau_{e,p}(\alpha_p)}b^\tau_{e,p}(\alpha_p,\beta;N)\widehat{W}_{\Lambda,\tau}(\alpha^{p\to \beta}),
\]
where the grouped local surgeries $\tau\in \mathcal T_{e,p}(\Lc)$, the finite sets $\B^\tau_{e,p}(\alpha_p)$, and the coefficients $b^\tau_{e,p}(\alpha_p,\beta;N)$ are those obtained at the end of Proposition~\ref{prop:explicit-incidence-recoupling}.
\end{definition}

\subsection{Proof of the master loop equation}

We can now prove the universal coefficientwise master loop equation.

\begin{proof}[Proof of Theorem~\ref{thm:universal-coefficientwise-master-equation}]
Start from Proposition~\ref{prop:first-order-haar-ibp}. The first term is identified with the universal loop operator by Proposition~\ref{prop:universal-loop-laplacian-term}. Each local incidence term is identified with the corresponding operator $\mathscr{B}_{e,p}$ by the grouped form obtained at the end of Proposition~\ref{prop:explicit-incidence-recoupling}. Summing over the incident plaquettes gives the result.
\end{proof}

\section{Wilson-action specializations}\label{sec:specializations}

The results of Sections~\ref{sec:surface-sum}--\ref{sec:master-loop} were formulated at fixed plaquette decoration $\alpha : P(\Lambda)\to \widehat{\U(N)}$. The purpose of this section is to show that the universal formalism recovers, with the correct normalizations, the known Wilson-action identities. Recall that we define the Wilson action by
\[
Q_t^{\mathrm W}(g)=\exp(t\Re\Tr(g))= \exp\left(\frac t2\bigl(\Tr(g)+\Tr(g^{-1})\bigr)\right),
\]
and in the general plaquette-weight notation of Section~\ref{sec:intro} we take $Q_p=Q^{\mathrm W}_{\beta_p}$, with plaquette couplings $(\beta_p)_{p\in P(\Lambda)}$. We still fix a loop family $\Lc=(\ell_1,\ldots,\ell_k)$.

\subsection{Surface-sum expansion}

The global surface expansion admits a more geometric specialization than the general coarse decorated-surface expansion of Theorem~\ref{thm:surface-sum-Wilsonloops}. The natural objects are no longer arbitrary decorated spanning surfaces, but edge-plaquette embeddings. However, we will recover the surface-sum from an intermediary step of Section~\ref{sec:surface-sum}.

Let us first prove a correspondence that recasts the coefficients in Proposition~\ref{prop:character_surface} in terms of Cao--Park--Sheffield's \emph{dual bipartite maps} (DBM). We refer to \cite{CPS25} for a precise description and definition.

\begin{lemma}\label{lem:trace-specialization-DBM}
Let $\Gamma=(\Gamma_1,\dots,\Gamma_k)$ be a balanced collection of words on the alphabet $\{\lambda_1,\dots,\lambda_L\}$. The following holds:
\begin{enumerate}
\item For every $i$, one has $r_i=1$, the Brauer index set $\mathcal I$ is a singleton, and $h(\tau)=0$. Consequently, the surface construction of Section~\ref{sec:construction-spanning-surfaces} involves no Brauer gadgets: it consists only of one polygon for each word $\Gamma_i$ and one Haar patch for each letter $\lambda_\ell$.

\item The data of a pair $(\tau,\sigma)$ occurring in Proposition~\ref{prop:character_surface} are therefore equivalent to the data of a collection of pairs of bijections
\[
(\alpha_\ell,\beta_\ell)\in S_{n_\ell,n_\ell},
\qquad \ell\in[L],
\]
where $n_\ell$ is the number of occurrences of $\lambda_\ell$ in $\Gamma$. The corresponding coarse surface class is canonically identified with an element $M\in DBM(\Gamma)$ in the sense of \cite{CPS25}.

\item Under this identification, the partition $\mu_\ell(M)$ is the cycle type of $\alpha_\ell^{-1}\beta_\ell$, and
\[
\Omega_N(\Xi)N^{\chi(\Xi)}=\left(\prod_{\ell=1}^L \mathrm{Wg}_N(\mu_\ell(M))\right)N^{V(M)},
\]
where $\Xi$ is the coarse surface class corresponding to $M$, and $V(M)$ denotes the number of vertices of $M$.
\end{enumerate}
\end{lemma}

\begin{proof}
Since $\chi_{[(1),\varnothing]_N}(g)=\Tr(g)$, the mixed tensor space attached to each word is just $V$, hence $r_i=1$ for every $i$. Therefore the local Brauer index set is trivial, there is no nontrivial projector expansion, and $h(\tau)=0$. The construction of Section~\ref{sec:construction-spanning-surfaces} reduces to the following: one starts with one polygon for each word $\Gamma_i$, and for each letter $\lambda_\ell$ one glues in a Haar patch determined by a pair of bijections $(\alpha_\ell,\beta_\ell)\in S_{n_\ell,n_\ell}$. This is exactly the construction of a map from the strand diagram of $\Gamma$: the word-polygons are the yellow faces, and the Haar patches are the blue faces. Hence the resulting coarse surface classes are canonically identified
with the elements of $DBM(\Gamma)$. 

By \cite[Definition 3.5]{CPS25}, for $M\in DBM(\Gamma)$ and each $\ell\in[L]$, the partition $\mu_\ell(M)$ is given by the half-degrees of the blue faces glued into the strand diagram of $\lambda_\ell$. Equivalently, $\mu_\ell(M)$ is the cycle type of $\alpha_\ell^{-1}\beta_\ell$. Thus the Weingarten factor attached to $(\tau,\sigma)$ is exactly
\[
\prod_{\ell=1}^L \mathrm{Wg}_N(\mu_\ell(M)).
\]
Moreover, by \cite[Lemma 3.4]{CPS25}, the number of connected components of the strand diagram associated with the bijection data $(\alpha_\ell,\beta_\ell)_{\ell\in[L]}$ is exactly the number $V(M)$ of vertices of the corresponding map. In the present trace-only specialization, the index-contraction factor $\mathcal K_N(\tau,\sigma)$ is therefore equal to $N^{V(M)}$. Since
\[
\Omega_N(\Xi)N^{\chi(\Xi)}
\]
is the total contribution of the corresponding pair $(\tau,\sigma)$, it follows that
\[
\Omega_N(\Xi)N^{\chi(\Xi)} = \left(\prod_{\ell=1}^L \mathrm{Wg}_N(\mu_\ell(M))\right)N^{V(M)},
\]
as claimed. 
\end{proof}

\begin{corollary}[\cite{CPS25}, Proposition 3.8]\label{cor:cps}
Let $\Gamma=(\Gamma_1,\dots,\Gamma_k)$ be a balanced collection of words on the alphabet $\{\lambda_1,\dots,\lambda_L\}$. For each $\ell\in\{1,\dots,L\}$, let $n_\ell$ be the number of occurrences of $\lambda_\ell$ in $\Gamma$. For $n\ge 0$, denote by $\mathrm{Wg}_{N}:S_n\to\mathbb{C}$ the ordinary unitary Weingarten function. If $\mu\vdash n$, write $\mathrm{Wg}_{N}(\mu)$ for its common value on permutations of cycle type $\mu$. If one defines
\[
\widetilde{\mathrm{Wg}}_{N}(\pi):=
N^{n+|\pi|}\mathrm{Wg}_{N}(\pi),
\qquad
|\pi|:=n-\#\mathrm{cycles}(\pi),
\]
then
\[
\mathbb{E}\left[\prod_{i=1}^k \Tr(\Gamma_i(U))\right]=\sum_{M\in \mathrm{DBM}(\Gamma)}\left(\prod_{\ell=1}^L \widetilde{\mathrm{Wg}}_{N}(\mu_\ell(M))\right)N^{\chi(M)-k}.
\]
\end{corollary}

\begin{proof}
Apply Proposition~\ref{prop:character_surface} in the trace-only specialization described in Lemma~\ref{lem:trace-specialization-DBM}. By that lemma, the sum over coarse surface
classes may be rewritten as a sum over $M\in DBM(\Gamma)$, and one obtains
\[
\mathbb E\left[\prod_{i=1}^k \Tr(\Gamma_i(U))\right]=\sum_{M\in DBM(\Gamma)}\left(\prod_{\ell=1}^L \mathrm{Wg}_N(\mu_\ell(M))\right)N^{V(M)}.
\]
This is already exactly the CPS map expansion before normalization of the Weingarten function. 

Now recall the identities stated in \cite{CPS25}:
\[
E(M)=2\sum_{\ell=1}^L n_\ell,\qquad F(M)=k+\sum_{\ell=1}^L \ell(\mu_\ell(M)),
\]
where $\ell(\mu)$ denotes the number of parts of the partition $\mu$. If $\pi_\ell\in S_{n_\ell}$ has cycle type $\mu_\ell(M)$, then
\[
|\pi_\ell|=n_\ell-\ell(\mu_\ell(M)),
\]
hence
\[
\sum_{\ell=1}^L \bigl(n_\ell+|\pi_\ell|\bigr)=2\sum_{\ell=1}^L n_\ell-\sum_{\ell=1}^L \ell(\mu_\ell(M))=E(M)-F(M)+k.
\]
Therefore, by the definition
\[
\widetilde{\mathrm{Wg}}_{N}(\pi):=N^{n+|\pi|}\mathrm{Wg}_{N}(\pi),
\]
we get
\[
\prod_{\ell=1}^L \widetilde{\mathrm{Wg}}_N(\mu_\ell(M))=\left(\prod_{\ell=1}^L \mathrm{Wg}_N(\mu_\ell(M))\right)
N^{E(M)-F(M)+k}.
\]
Using Euler's identity
\[
\chi(M)=V(M)-E(M)+F(M),
\]
we conclude that
\[
\left(\prod_{\ell=1}^L \widetilde{\mathrm{Wg}}_N(\mu_\ell(M))\right)N^{\chi(M)-k}=\left(\prod_{\ell=1}^L \mathrm{Wg}_N(\mu_\ell(M))\right)N^{V(M)}.
\]
Substituting this into the previous formula yields the desired identity. 
\end{proof}

For two multiplicity fields
\[
K^+,K^-:P(\Lambda)\to \mathbb N,
\]
define the associated collection of words
\[
\Gamma_{\Lambda,\Lc}(K^+,K^-):=(\ell_1,\dots,\ell_k)\sqcup\bigsqcup_{p\in P(\Lambda)} (\partial p)^{\sqcup K^+(p)}\sqcup\bigsqcup_{p\in P(\Lambda)} \bigl((\partial p)^{-1}\bigr)^{\sqcup K^-(p)} .
\]
We write
\[
|K|:=\sum_{p\in P(\Lambda)} \bigl(K^+(p)+K^-(p)\bigr),\qquad(K^+,K^-)!:=\prod_{p\in P(\Lambda)} K^+(p)!K^-(p)!,
\]
and
\[
\beta^{K^+,K^-}:=\prod_{p\in P(\Lambda)}\left(\frac{\beta_p}{2}\right)^{K^+(p)+K^-(p)}.
\]

We denote by $\EPE_{\Lambda}(\Lc;K^+,K^-)$ the set of edge-plaquette embeddings $(M,\phi)$ with boundary $\Lc=(\ell_1,\dots,\ell_k)$ and oriented plaquette counts $(K^+,K^-)$, meaning that:
\begin{itemize}
\item $M$ is a compact oriented surface with $k$ ordered boundary components;
\item $\phi:M\to \Lambda^{(2)}$ is cellular;
\item the $i$-th boundary component maps to the loop $\ell_i$;
\item for each plaquette $p$, exactly $K^+(p)$ plaquette-faces of $M$ are mapped homeomorphically onto $p$ with the chosen orientation, and exactly $K^-(p)$ plaquette-faces are mapped homeomorphically onto $p$ with the opposite orientation.
\end{itemize}
For $(M,\phi)\in \EPE_{\Lambda}(\Lc;K^+,K^-)$ and $e\in E(\Lambda)$, let $\mu_e(\phi)$ be the partition whose parts are half the boundary-degrees of the edge-faces mapped to $e$.

\begin{corollary}
\label{cor:wilson-epe}
For the Wilson action,
\[
Z\E[W_{\Lambda,\Lc}(U)]=\sum_{K^+,K^-:P(\Lambda)\to \mathbb N}\frac{\beta^{K^+,K^-}}{(K^+,K^-)!}\sum_{(M,\phi)\in \EPE_{\Lambda}(\Lc;K^+,K^-)}\left(\prod_{e\in E(\Lambda)}\widetilde{\mathrm{Wg}}_N(\mu_e(\phi))\right)N^{\chi(M)-k},
\]
where $k$ is the number of loops in the family $\Lc$.
\end{corollary}

\begin{proof}
Expand each Wilson plaquette weight into its trace power series:
\[
Q_{\beta_p}^{\mathrm W}(U_{\partial p})=\sum_{a,b\ge 0}\frac{(\beta_p/2)^{a+b}}{a!b!}\Tr(U_{\partial p})^a \Tr(U_{\partial p}^{-1})^b.
\]
Multiplying over all plaquettes gives
\[
Z\E[W_{\Lambda,\Lc}(U)]=\sum_{K^+,K^-:P(\Lambda)\to \mathbb N}\frac{\beta^{K^+,K^-}}{(K^+,K^-)!}\int_{\U(N)^{E(\Lambda)}}\prod_{w\in \Gamma_{\Lambda,\Lc}(K^+,K^-)}\Tr(U_w) dU.
\]
Fix a pair $(K^+,K^-)$. The integrand is now a product of traces only, so we are in the trace-only specialization of the surface construction of Section~\ref{sec:surface-sum}.

In this specialization, every word is carried by the fundamental representation or its dual. Hence the local mixed-tensor rank is $1$ at every insertion, the local Brauer index set is trivial, and no horizontal Brauer bands appear. In particular, the defect term vanishes. The refined surface construction of Theorem~\ref{thm:fixed-alpha-punctured-sum} therefore reduces to polygons for the loop and plaquette words together with Haar patches, and the resulting surfaces are precisely compact oriented surfaces with ordered boundary $\Lc$, endowed with a cellular map to $\Lambda^{(2)}$ whose plaquette-faces map homeomorphically onto plaquettes with multiplicities prescribed by $(K^+,K^-)$. In other words, the resulting objects are exactly the edge-plaquette embeddings $(M,\phi)\in \EPE_{\Lambda}(\Lc;K^+,K^-)$.

Moreover, in the same trace-only specialization, Corollary~\ref{cor:cps} identifies the local Haar contribution with the product of renormalized unitary Weingarten weights
\[
\prod_{e\in E(\Lambda)} \widetilde{\mathrm{Wg}}_N(\mu_e(\phi)),
\]
and the global power of $N$ is $N^{\chi(M)-k}$, where $k=\#\Lc$.

Substituting this trace-only surface expansion into the Wilson trace series yields the expected formula.
\end{proof}

\begin{remark}
Corollary~\ref{cor:wilson-epe} is the Wilson specialization of the refined surface-sum machinery underlying Theorem~\ref{thm:surface-sum-Wilsonloops}. It is more geometric than the dual-bipartite-map specialization of Corollary~\ref{cor:cps}, and should be viewed as the natural bridge between the present paper and the edge-plaquette-embedding formulation of Cao--Park--Sheffield.
\end{remark}

\subsection{Master loop equation}

Let us now turn to the master loop equation. For $p\in P^+(e)$, define the plaquette-resolved deformation multisets
\[
D^{+,p}_{i,e}(\Lc):=\{(\ell_1,\dots,\ell_{i-1},\ell_i\oplus_x p,\ell_{i+1},\dots,\ell_k):x\in C_i(e;\Lc)\},
\]
and
\[
D^{-,p}_{i,e}(\Lc):=\{(\ell_1,\dots,\ell_{i-1},\ell_i\ominus_x p,\ell_{i+1},\dots,\ell_k):x\in C_i(e;\Lc)\},
\]
counted with multiplicity. Define also the fully summed surgery multisets
\[
D^\pm(\Lc):=\bigsqcup_{i=1}^k \bigsqcup_{e\in D_i(\Lc)} \bigsqcup_{p\in P^+(e)} D^{\pm,p}_{i,e}(\Lc),\qquad S^\pm(\Lc):=\bigsqcup_{i=1}^k \bigsqcup_{e\in D_i(\Lc)} S^\pm_{i,e}(\Lc),
\]
\[
M_U^\pm(\Lc):=\bigsqcup_{i=1}^k \bigsqcup_{e\in D_i(\Lc)} M^\pm_{U,i,e}(\Lc),
\]
all counted with multiplicity.

\begin{proposition}\label{prop:wilson-master}
Let $\Lc=(\ell_1,\dots,\ell_k)$ be a loop family, and define $M_{\Lambda,N,\beta}(\Lc):=\mathbb E[W_\Lc(U)]$ for the unnormalized Wilson moment, and $\Phi_{\Lambda,N,\beta}(\Lc)=N^{-\#\Lc}M_{\Lambda,N,\beta}(\Lc)$ for the normalized Wilson loop expectation, where $\#\Lc$ is the number of loops in $\Lc$. Assume that the coupling is constant for each plaquette: $\beta_p=\beta$ for all $p\in P(\Lambda)$. Then
\begin{align}
N\ell(\Lc)\Phi_{\Lambda,N,\beta}(\Lc)&=\frac{\beta}{2}
\sum_{\Lc'\in D^-(\Lc)} \Phi_{\Lambda,N,\beta}(\Lc')-\frac{\beta}{2}
\sum_{\Lc'\in D^+(\Lc)} \Phi_{\Lambda,N,\beta}(\Lc')\notag\\
&\quad+N\sum_{\Lc'\in S^-(\Lc)} \Phi_{\Lambda,N,\beta}(\Lc')-N\sum_{\Lc'\in S^+(\Lc)} \Phi_{\Lambda,N,\beta}(\Lc')\notag\\
&\quad+\frac1N\sum_{\Lc'\in M_U^-(\Lc)} \Phi_{\Lambda,N,\beta}(\Lc')-\frac1N\sum_{\Lc'\in M_U^+(\Lc)} \Phi_{\Lambda,N,\beta}(\Lc').
\label{eq:fully-summed-wilson-mle-normalized}
\end{align}
\end{proposition}

\begin{proof}
For each oriented edge $e$, the same integration-by-parts argument as in Proposition~\ref{prop:first-order-haar-ibp} gives $0=I_{\Delta,e}+I_{\mathrm{mix},e}.$ For the loop-Laplacian term, Proposition~\ref{prop:universal-loop-laplacian-term} yields
\begin{align*}
I_{\Delta,e}
&=(\mathscr L_eM_{\Lambda,N,\beta})(\Lc) \\
&=-N\sum_{i:e\in D_i(\Lc)} m_i(e;\Lc)M_{\Lambda,N,\beta}(\Lc)+\sum_{i:e\in D_i(\Lc)}\left(\sum_{\Lc'\in S^-_{i,e}(\Lc)} M_{\Lambda,N,\beta}(\Lc')-\sum_{\Lc'\in S^+_{i,e}(\Lc)} M_{\Lambda,N,\beta}(\Lc')
\right) \\
&\qquad\qquad+\sum_{i:e\in D_i(\Lc)}\left(\sum_{\Lc'\in M^-_{U,i,e}(\Lc)} M_{\Lambda,N,\beta}(\Lc')-\sum_{\Lc'\in M^+_{U,i,e}(\Lc)} M_{\Lambda,N,\beta}(\Lc')\right).
\end{align*}

It remains to identify the mixed term. Under the Wilson action,
\[
Q^W_{\beta}(U_{\partial p})=\exp\left(\frac{\beta}{2}\big(\Tr(U_{\partial p})+\Tr(U_{\partial p}^{-1})\big)\right),
\]
hence, for each incidence $(e,p)$ with $p\in P^+(e)$,
\[
\partial_{e,a}Q^W_{\beta}(U_{\partial p})=\frac{\beta}{2}Q^W_{\beta}(U_{\partial p})\partial_{e,a}\left(\Tr(U_{\partial p})+\Tr(U_{\partial p}^{-1})\right).
\]
Therefore the mixed term at $(e,p)$ is
\[
\frac{\beta}{2}\mathbb E\left[\sum_{a=1}^{N^2}(\partial_{e,a}W_\Lc)\partial_{e,a}\bigl(\Tr(U_{\partial p})+\Tr(U_{\partial p}^{-1})\bigr)\right].
\]

Now use the magic formulas~\eqref{eq:magic} as in the proof of Proposition~\ref{prop:universal-loop-laplacian-term}. For each active occurrence $x\in C_i(e;\Lc)$, contracting the marked loop strand with the marked plaquette strand produces exactly the two local plaquette deformations of $\ell_i$ by the oriented plaquette boundary $\partial p$.  If the marked loop and plaquette occurrences have the same orientation, the magic formula contributes a minus sign and produces the positive deformation.  If they have opposite orientations, the derivative of the inverse occurrence supplies an additional minus sign, so the resulting negative deformation has positive sign.  Thus the negative deformation enters the mixed term with sign $+$ and the positive deformation with sign $-$. Summing over all active occurrences of $e$ or $e^{-1}$ in the whole loop family, one obtains
\[
I_{\mathrm{mix},e}=\frac12\sum_{i:e\in D_i(\Lc)}\sum_{p\in P^+(e)}\beta\left(\sum_{\Lc'\in D^{-,p}_{i,e}(\Lc)} M_{\Lambda,N,\beta}(\Lc')-\sum_{\Lc'\in D^{+,p}_{i,e}(\Lc)} M_{\Lambda,N,\beta}(\Lc')\right).
\]
Therefore
\begin{align*}
&N\sum_{i:e\in D_i(\Lc)}m_i(e;\Lc)M_{\Lambda,N,\beta}(\Lc)\\
&\quad=\frac{\beta}{2}\sum_{i:e\in D_i(\Lc)}\sum_{p\in P^+(e)}
\left(
\sum_{\Lc'\in D^{-,p}_{i,e}(\Lc)}M_{\Lambda,N,\beta}(\Lc')
-
\sum_{\Lc'\in D^{+,p}_{i,e}(\Lc)}M_{\Lambda,N,\beta}(\Lc')
\right)\\
&\qquad+
\sum_{i:e\in D_i(\Lc)}
\left(
\sum_{\Lc'\in S^-_{i,e}(\Lc)}M_{\Lambda,N,\beta}(\Lc')
-
\sum_{\Lc'\in S^+_{i,e}(\Lc)}M_{\Lambda,N,\beta}(\Lc')
\right)\\
&\qquad+
\sum_{i:e\in D_i(\Lc)}
\left(
\sum_{\Lc'\in M^-_{U,i,e}(\Lc)}M_{\Lambda,N,\beta}(\Lc')
-
\sum_{\Lc'\in M^+_{U,i,e}(\Lc)}M_{\Lambda,N,\beta}(\Lc')
\right).
\end{align*}

Now sum over all oriented edges $e$ belonging to some $D_i(\Lc)$. Since
\[
\sum_{e\in D_i(\Lc)} m_i(e;\Lc)=|\ell_i|\qquad\text{for every }i,
\]
the left-hand side becomes
\[
N\sum_{i=1}^k |\ell_i|M_{\Lambda,N,\beta}(\Lc)=N\ell(\Lc)M_{\Lambda,N,\beta}(\Lc).
\]
By definition of the multisets $D^\pm(\Lc)$, $S^\pm(\Lc)$ and $M_U^\pm(\Lc)$ as disjoint unions with multiplicity, since deformations preserve the number of loops, splittings increase it by one, and mergers decrease it by one, dividing by \(N^{\#\Lc}\) yields \eqref{eq:fully-summed-wilson-mle-normalized}.
\end{proof}

Note that there is a small difference in the statement of Proposition~\ref{prop:wilson-master} compared to \cite{SSZ24}: to recover the exact result of \cite{SSZ24}, $\beta$ should be replaced by $N\beta$, which corresponds to another convention related to the so-called \emph{'t Hooft regime}.

\appendix
\section{Adaptation to other compact classical groups}\label{sec:appendix}

In this appendix we explain more concretely how the strategy of the paper should extend to the other compact classical groups. Our purpose is not to reprove every statement in full detail, but to isolate the finite list of representation-theoretic inputs on which Sections~\ref{sec:surface-sum}--\ref{sec:master-loop} depend, and to indicate how these inputs are replaced for $\SU(N)$, $\SO(N)$ and $\Sp(N)$.

At a structural level, the proofs of the unitary case only use the following ingredients.

\medskip

\noindent
\textbf{(I) Fourier expansion of the plaquette action.}
For every compact group $G$, every central square-integrable function admits a character expansion by Peter--Weyl. This yields the analogue of Theorem~\ref{thm:state_sum_Wilson} verbatim.

\medskip

\noindent
\textbf{(II) Tensor realization of irreducible characters.}
For each irreducible label $\lambda$ of $G$, one needs a finite-dimensional tensor space $T_\lambda$, a representation $\rho_\lambda$ of $G$ on $T_\lambda$, and a $G$-equivariant projector $P_\lambda$ such that
\[
\chi_\lambda(g)=\Tr_{T_\lambda}\!\bigl(P_\lambda \rho_\lambda(g)\bigr).
\]

\medskip

\noindent
\textbf{(III) Diagrammatic expansion of the projector.}
One needs a finite diagrammatic basis of the commutant algebra of the relevant tensor power, together with an expansion
\[
P_\lambda=\sum_{\tau} c_{\lambda,N}(\tau)\,\rho(\tau)
\]
in this basis. In the unitary case this basis is given by walled Brauer diagrams.

\medskip

\noindent
\textbf{(IV) Haar integration in the same diagrammatic language.}
One needs a Weingarten-type expansion for Haar integrals of matrix coefficients, expressed in terms of the same diagrammatic basis. Once (III) and (IV) are available, the refined expansions of Section~\ref{sec:surface-sum} follow by exactly the same pattern: one inserts the projector expansion, integrates
letter by letter, and regroups the resulting finite sums.

\medskip

\noindent
\textbf{(V) Local Lie-algebra contraction identities.}
To derive the master loop equation, one only needs the analogue of the magic formulas~\eqref{eq:magic} for the defining representation of the group under consideration. Once these are known, the proof of the loop part is the usual group-dependent cut-and-join computation, while the mixed term is again a local incidence computation on one loop strand and one plaquette strand.

\medskip

Once these ingredients are granted, the remaining formal manipulations follow the same pattern: Section~\ref{sec:surface-sum} produces a decorated-surface expansion from (II)--(IV), Section~\ref{sec:local-channel} only uses the resulting local resolutions, finite tensor-product scalarization and edgewise Haar projection, and Section~\ref{sec:master-loop} only uses (V) together with the same locality mechanism. Thus the problem of adapting the whole formalism to another compact classical group reduces to identifying the correct replacements of the tensor model, of the commutant diagrams, and of the Weingarten kernels.

\subsection{The special unitary group}

The case of $\SU(N)$ is the closest to the unitary theory. Irreducible representations of $\SU(N)$ are still realized inside tensor constructions built from the defining representation $V=\mathbb C^N$ and its dual $V^*$, and the relevant commutant is still governed by the mixed Schur--Weyl formalism. Concretely, one may choose for each dominant weight of $\SU(N)$ a highest-weight representative for $\U(N)$, for instance with last coordinate equal to $0$, and realize the corresponding $\SU(N)$-module inside a mixed tensor space $T_{n,m}=V^{\otimes n}\otimes (V^*)^{\otimes m}.$ The associated isotypic projector is then described by the same kind of mixed-tensor representation-theoretic data as in Section~\ref{sec:surface-sum}.

There is, however, one point that should be stated explicitly: compared with $\U(N)$, the commutant for $\SU(N)$ is enlarged by the possible presence of determinant-type intertwiners coming from the invariant volume tensor. Equivalently, in addition to the usual walled Brauer contractions, one may have local $\varepsilon$-tensor insertions when the total covariant and
contravariant degrees are not balanced in the same way as in the unitary case. Thus the correct diagrammatic basis for $\SU(N)$ should be viewed as the unitary walled-Brauer basis, possibly enlarged by finitely many determinant channels.

After making this enlargement, the strategy of Sections~\ref{sec:surface-sum}--\ref{sec:master-loop} follows by the same formal steps, modulo the replacement of the representation-theoretic inputs.
Indeed:
\begin{itemize}
\item the state-sum expansion is still obtained by Peter--Weyl on $\SU(N)$;
\item the topological coefficients are still integrals of products of irreducible characters realized inside mixed tensor spaces;
\item the local projector expansions are still finite and local, now in the enlarged $\SU(N)$ commutant basis;
\item Haar integration still acts letterwise and produces local Weingarten coefficients attached to the same local channels;
\item the defect-ratio formalism of Section~\ref{sec:local-channel} is unaffected, since it only uses the existence of finite local resolutions and edgewise compression;
\item the master loop equation of Section~\ref{sec:master-loop} again reduces to the usual $\SU(N)$ cut-and-join operator together with a local spectral recoupling term at the active plaquette.
\end{itemize}

These observations indicate how analogues of Theorems~\ref{thm:state_sum_Wilson}--\ref{thm:universal-coefficientwise-master-equation} should be formulated for $\SU(N)$ once the enlarged commutant expansion, including the determinant-type channels, is established with the required finite-$N$ identities. We leave that group-specific verification to future work.

\subsection{The orthogonal and symplectic groups}

For $\SO(N)$ and $\Sp(N)$, the required replacements are different but equally standard. The defining representation is self-dual, so the covariant/contravariant distinction disappears. As a result, the mixed tensor spaces of the unitary case are replaced by ordinary tensor powers $V^{\otimes r},$ and the walled Brauer algebra is replaced by the ordinary Brauer algebra.

More precisely, if $G_N$ denotes either $\SO(N)$ or $\Sp(N)$, the commutant of the $G_N$-action on $V^{\otimes r}$ is described by Brauer diagrams. The analogue of the unitary traceless mixed projector is the projector onto the appropriate harmonic/traceless summand with respect to the invariant bilinear form preserved by $G_N$. Thus the role of Section~\ref{sec:mixed-schur-weyl} is played by the usual orthogonal/symplectic Schur--Weyl duality: irreducible characters are realized as traces of $G_N$-equivariant projectors on tensor powers, and these projectors admit finite expansions in the Brauer basis.

Once this replacement is made, the rest of the machinery is parallel:
\begin{itemize}
\item in the Weingarten expansion, ordinary Brauer diagrams replace walled Brauer diagrams;
\item the local contractions and coevaluations are taken with respect to the orthogonal or symplectic invariant form;
\item the edgewise Haar integrals are governed by the orthogonal or symplectic Weingarten kernels;
\item the local channel data of Section~\ref{sec:local-channel} now live on ordinary tensor legs, without a wall separating covariant and contravariant sectors;
\item the coefficientwise master loop equation again follows from the corresponding Lie-algebra contraction formulas, with the cut-and-join part replaced by the orthogonal/symplectic one.
\end{itemize}

Thus one expects analogues of the local resolution theorem, the defect-ratio formula and the coefficientwise master loop equation for $\SO(N)$ and $\Sp(N)$ once the corresponding Brauer-projector expansions and group-specific contraction formulas are inserted. We do not verify those finite-$N$ replacements here.

\medskip

On the geometric side, the surface expansion should be understood through the thickening of ordinary Brauer diagrams rather than walled Brauer diagrams. This would yield the analogue of Section~\ref{sec:surface-sum}, with the same finite local defect bookkeeping. If one wishes to formulate the resulting decorated-surface expansion in an explicitly geometric language, one should use the natural orthogonal/symplectic version of the Brauer-surface model rather than force the unitary oriented one; the precise topological packaging is secondary here, since the local channel and master-equation parts depend only on the underlying finite diagrammatic calculus. It is worth mentioning that, as usually seen in topological expansions involving the orthogonal and symplectic groups, the surfaces obtained through this procedure might not be orientable, in contrast to those obtained for unitary groups.

\medskip

In summary, for all compact classical groups the expected adaptation problem is local and representation-theoretic: one must replace the tensor model, the diagrammatic commutant basis, the corresponding Weingarten kernel, and the Lie-algebra contraction formulas. Once these group-specific inputs have been established with the required finite-$N$ identities, the arguments of the present paper indicate how the state-sum, decorated-surface organization, local channel model, defect-ratio representation, and coefficientwise master loop equation should be reconstructed by the same formal mechanism as in the unitary case.

\bibliographystyle{alpha}
\bibliography{Exact_Expansions_YM}

\end{document}